\documentclass[11pt]{article}
\usepackage[margin=1.25in]{geometry}
\usepackage{amssymb,amsfonts,amsmath,amsthm,amscd,dsfont,mathrsfs,bbold,pifont}
\usepackage{blkarray}
\usepackage{graphicx,float,psfrag,epsfig,color}
\usepackage{comment}
\usepackage{subcaption}
\usepackage{multicol}

\usepackage{algorithm}
\usepackage[noend]{algpseudocode}

\makeatletter
\def\BState{\State\hskip-\ALG@thistlm}
\makeatother

	\usepackage{microtype}
	\usepackage[pdftex,pagebackref=true,colorlinks]{hyperref}
	\hypersetup{linkcolor=[rgb]{.7,0,0}}
	\hypersetup{citecolor=[rgb]{0,.7,0}}
	\hypersetup{urlcolor=[rgb]{.7,0,.7}}

\newcommand{\Cnote}[1]{\begin{center}\fbox{\begin{minipage}{35em}
                        {{\bf Colin Note:} {#1}} \end{minipage}}\end{center}}

\newcommand\independent{\protect\mathpalette{\protect\independent}{\perp}} 
\def\independent#1#2{\mathrel{\rlap{$#1#2$}\mkern2mu{#1#2}}} 

\newcommand{\diam}{\mathrm{diam}}
\newcommand{\hbm}{\mathrm{HBM}}
\newcommand{\gbm}{\mathrm{GBM}}
\newcommand{\ssbm}{\mathrm{SSBM}}

\newcommand{\mR}{\mathbb{R}}

\newcommand{\pp}{\mathbb{P}}
\newcommand{\E}{\mathbb{E}}

\newcommand{\e}{\varepsilon}

\newcommand{\1}{\mathbb{1}}
\newcommand{\snr}{\mathrm{SNR}}
\newcommand{\sbm}{\mathrm{SBM}}

\usepackage{braket}

\newcommand{\abs}[1]{\left\vert #1 \right\vert}

\newcommand{\parens}[1]{\left( #1 \right)}

\newtheorem{proposition}{Proposition}
\newtheorem{theorem}{Theorem}

\newtheorem{lemma}{Lemma}
\newtheorem{corollary}{Corollary}
\newtheorem{definition}{Definition}
\newtheorem{remark}{Remark}

\newtheorem{conjecture}{Conjecture}

\begin{document}

\title{Graph powering and spectral robustness}

\author{Emmanuel Abbe\thanks{This research was partly supported by the NSF CAREER Award CCF-1552131 and NSF CSOI CCF-0939370} $^{1,3}$, Enric Boix$^{1,2}$, Peter Ralli$^{1}$, Colin Sandon$^{1,2}$\\
\vspace{-.3cm}\\
$^1$Princeton University \qquad $^2$MIT \qquad $^3$EPFL}

%\author{Emmanuel Abbe\thanks{Program in Applied and Computational Mathematics, and Department of Electrical Engineering, Princeton University, \texttt{eabbe@princeton.edu}. This research was partly supported by the NSF CAREER Award CCF-1552131, the ARO grant W911NF-16-1-0051, and the Google Faculty Research Award.} \\ Princeton University 
%\and  Colin Sandon\thanks{Department of Mathematics, Princeton University, USA, \texttt{sandon@princeton.edu}.} 
%}

\date{}
%\date{February 28, 2015}
\maketitle

\begin{abstract}
Spectral algorithms, such as principal component analysis and spectral clustering, typically require careful data transformations to be effective: upon observing a matrix $A$, one may look at the spectrum of $\psi(A)$ for a properly chosen $\psi$.
The issue is that the spectrum of $A$ might be contaminated by non-informational top eigenvalues, e.g., due to scale` variations in the data, and the application of $\psi$ aims to remove these. 
%In simple cases, $\psi$ may consist in truncating/normalizing large entries (e.g., the normalized Laplacian), or completing/smoothing small entries.  

Designing a good functional $\psi$ (and establishing what good means) is often challenging and model dependent. 
This paper proposes a simple and generic construction for sparse graphs, $$\psi(A) = \1((I+A)^r \ge1),$$ where $A$ denotes the adjacency matrix and $r$ is an integer (less than the graph diameter). 
This produces a graph connecting vertices from the original graph that are within distance $r$, and is referred to as graph powering. It is shown that graph powering regularizes the graph and decontaminates its spectrum in the following sense:
(i) If the graph is drawn from the sparse Erd\H{o}s-R\'enyi ensemble, which has no spectral gap, it is shown that graph powering produces a ``maximal'' spectral gap, with the latter justified by establishing an Alon-Boppana result for powered graphs; (ii) If the graph is drawn from the sparse SBM, graph powering is shown to achieve the fundamental limit for weak recovery (the KS threshold) similarly to \cite{massoulie-STOC}, settling an open problem therein. Further, graph powering is shown to be significantly more robust to tangles and cliques than previous spectral algorithms based on self-avoiding or nonbacktracking walk counts \cite{massoulie-STOC,Mossel_SBM2,bordenave,colin3}. This is illustrated on a geometric block model that is dense in cliques. 

%The construction is derived from Bayesian principles and connected to robust cuts. 

%Finally, we establish an Alon-Boppana result for graph powering of regular graphs. This allows to formalize the fact that graph powering achieves a ``maximal'' spectral gap for Erdos-Renyi random graphs, making them almost Ramanujan.  
\end{abstract}

\newpage

\tableofcontents

\newpage

\section{Introduction}
\subsection{Spectral data analysis and robustness}
A large variety of algorithms exploit the spectrum of graph\footnote{Possibly graphs with edge weights.} operators.
This concerns most methods of unsupervised learning that relies on spectral decomposition, e.g., principal component analysis, clustering or linear embeddings. 
The common base of spectral algorithms is to first obtain an Euclidean embedding of the data (which may a priori have no relation to a metric space), and then use this embedding for further tasks. Namely, given a $n$-vertex graph $G$ with adjacency matrix $A_G$,
\begin{enumerate}
\item Construct an operator $M_G=\psi(A_G)$;
\item Take the top $k$ eigenvectors of $M_G$ to create the $n \times k$ matrix $\Psi_k$, and use $\Psi_k(i)$ as the $k$-dimensional embedding for the data point $i \in [n]$. 
\end{enumerate} 
In clusterings, one typically looks for $k$ much smaller than $n$, cutting off potentially significant matrix norm, and running the $k$-means algorithm on the embedded points to obtain clusters \cite{ulrike}. In word embeddings, one may preserve almost all the matrix norm in order to approximate each word co-occurrence \cite{word_survey}.

Popular choices for $M_G$ (depending on applications) are the adjacency matrix $A$, the Laplacian $D - A$, the normalized Laplacian $I-D^{-1/2}AD^{-1/2}$ (or the random walk Laplacian $D^{-1}A$), and various regularized versions of the above using trimming/thresholding operations, or smoothing/completion operations. More recently, operators based on self-avoiding or nonbacktracking walks have become popular in the context of block models \cite{massoulie-STOC,Mossel_SBM2,bordenave,colin3}; see further discussion of these below. 
The function $\psi$ can also take specific forms such as the log-PMI function in word embeddings that uses both normalization and the application of the logarithm entrywise \cite{church,word_survey,pramod_word}, or more sophisticated (non-positive) forms such as in phase retrieval \cite{marco_phase}.  
A long list of other forms is omitted here. 

Why is it important to apply a transformation $\psi$? Consider graph clustering; if one takes $M_G=A$ directly, the top eigenvector is likely to localize on the largest degree vertex of the graph, which is not the type of macroscopic structure that clustering aims to extract.
This is a well-known issue, illustrated in Figure \ref{high-degree} using the spectral algorithm on $A$ to cluster the SBM in the sparse regime. Moreover, pruning (also called trimming or thresholding) the largest degree nodes as done in \cite{joseph,levina-reg,coja-sbm,Vu-arxiv,sbm-groth,new-vu} does not allow to circumvent this issue in the sparse SBM \cite{kawamoto2015limitations}`; either outliers would persist or too much of the graph would be destroyed.  
Similarly, in word embeddings, without mitigating the most popular words such as ``the", the embedding would assign full dimension on these. On the flip side, if one takes normalized Laplacians for clustering, one may overcorrect the large degree nodes and output clusters that are now at the periphery of the graph, such as ``tails'' of the graph --- see Figure \ref{lap-cut} for an example on the SBM. These are discussed in more detail in Section \ref{derivation}. In particular, neither the pruned adjacency matrix nor the normalized Laplacian achieve the threshold for weak recovery in the SBM. In general, transformations $\psi$ crucially serve to ``regularize'' the graph in order to obtain useful embeddings that capture ``macroscopic'' structures.

\begin{figure}[H]
\begin{center}
\includegraphics[scale=.3]{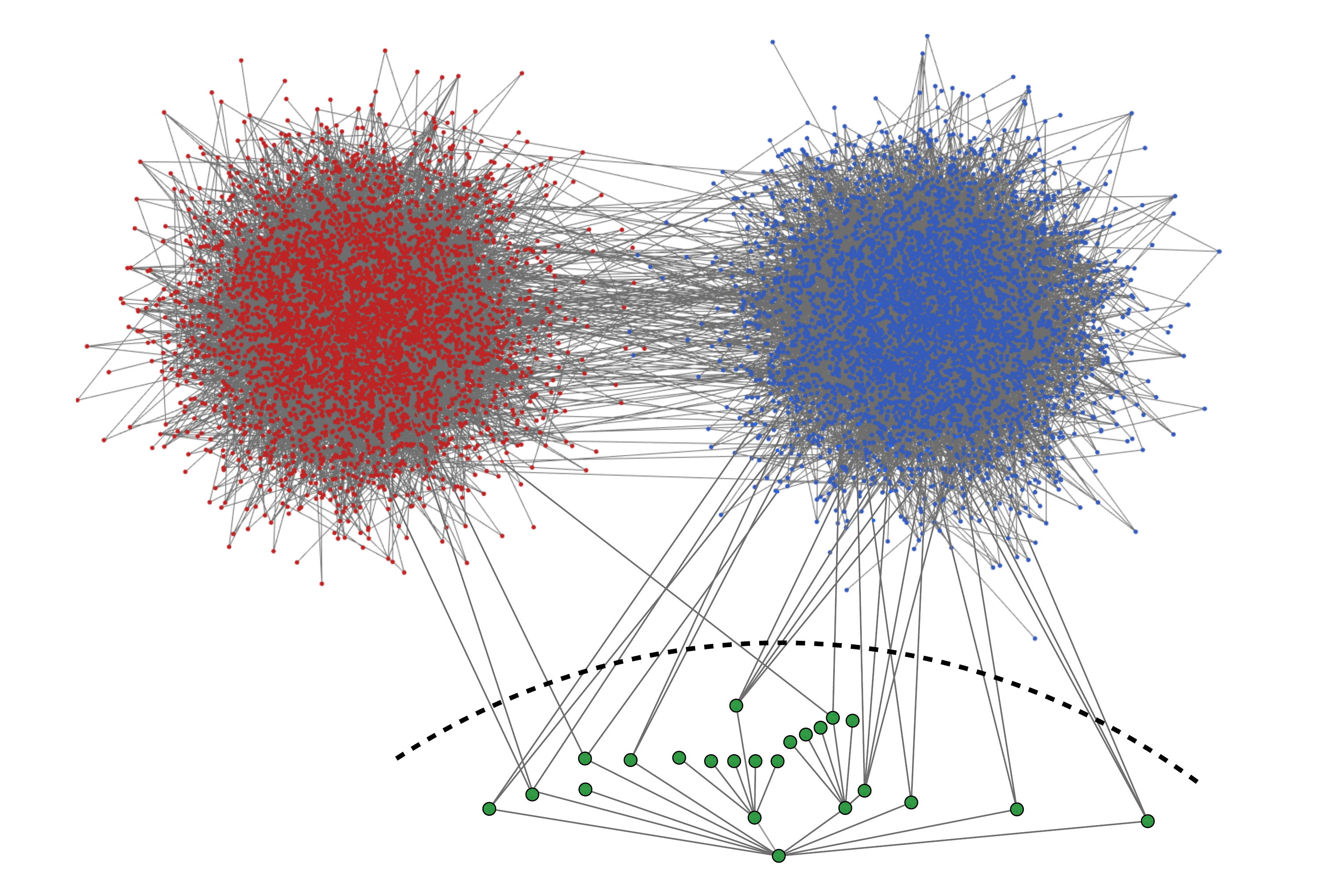}
\caption{The communities obtained with the spectral algorithm on the adjacency matrix $A$ in a sparse symmetric SBM above the KS threshold ($n=100000, a=2.2, b=0.06$): one community corresponds to the neighborhood of a high-degree vertex, and all other vertices are put in the second community.}
\label{high-degree}
\end{center}
\end{figure}
\begin{figure}[H]
\begin{center}
\includegraphics[scale=.3]{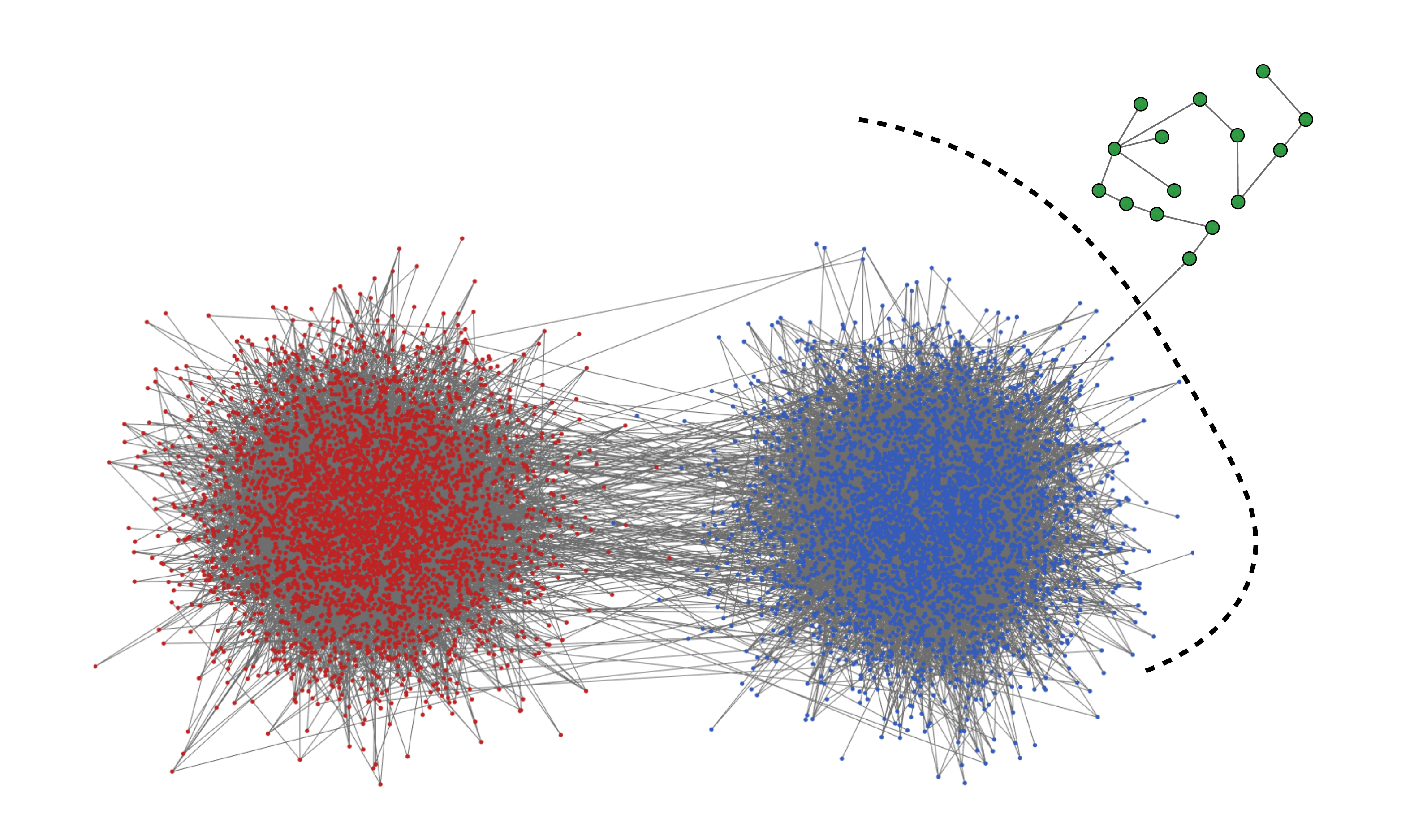}
\caption{The communities obtained with the spectral algorithm on the Laplacian matrix $L$ in a sparse symmetric SBM above the KS threshold ($n=100000, a=2.2, b=0.06$): one community corresponds to a ``tail'' of the graph (i.e., a small region connected by a single edge to the giant component), and all other vertices are put in the second community. The same outcome takes place for the normalized Laplacian.}
\label{lap-cut}
\end{center}
\end{figure}

The robustness of spectral methods has been studied in various contexts, mainly for objective-based clustering, using adversarial corruptions \cite{kane_robust,pravesh_2,pravesh_robust} and studying noise perturbations \cite{ulrike,li,karrer_robust,bala_nips}, among others. The goal of this paper is to formalize the notion of robustness in the context of community detection/clustering for probabilistic graph models, and to obtain a transformation $\psi$ that is as simple and as robust as possible. The derived transformation can also be viewed as a method to extract spectral gaps in graphs that have a spectral gap hiding under local irregularities, such as the Erd\H{o}s-R\'enyi random graph. 

\subsection{Our benchmarks}
We refer to next section for the formal statements. Our goal is to find an ``optimal'' $\psi$ that satisfies the following:
\begin{enumerate}
\item[(1)] Achieving a ``maximal'' spectral gap on Erd\H{o}s-R\'enyi (ER) random graphs, e.g., comparable to that of random regular graphs, as evidence that $\psi$ is a good graph regularizer; 
\item[(2)]  Achieving the Kesten-Stigum threshold for weak recovery in the SBM, as evidence that $\psi$ is good for community detection;
\item[(3)]  Achieving the threshold for weak recovery in a geometric block model, as evidence that $\psi$ is robust to cliques and tangles common in clustering problems. 
\end{enumerate}
We provide further explanation on these points:\\

\noindent
(1) The first point refers to the following: a random regular graph of degree $d$ has with high probability top two eigenvalues $\lambda_1 =d$ and $\lambda_2 \sim 2 \sqrt{d-1}$ \cite{friedman}. This is a ``maximal'' gap in the sense that Alon-Boppana shows that any regular graph must have a second eigenvalue larger than $(1-o_{\diam(G)}(1)) 2 \sqrt{d-1}$. Ramanujan families of graphs are defined as families achieving the lower-bound of $2 \sqrt{d-1}$, and so random regular graphs are almost Ramanujan with high probability.

In contrast, an ER random graph of expected degree $d$ has top two eigenvalues both of order $\sqrt{\log(n)/\log\log(n)}$ with high probability, and these correspond to eigenvectors localized on the vertices of largest degree. So ER graphs are not almost Ramanujan (if one puts aside the fact that they are not regular and applies the above Ramanujan definition in a loose way.) This is however restrictive.  

Are there more general definitions of ``maximal gaps'' that are more suitable to irregular graphs? Some known proposals are based on the universal cover \cite{hoory} and the nonbacktracking operator. 

In particular, it is shown in 
\cite{terras} that a regular graph is Ramanujan if and only if its nonbacktracking matrix has no eigenvalue of magnitude within $(\sqrt{\rho}, \rho)$, where $\rho$ is the Perron-Frobenius eigenvalue of the nonbacktracking matrix. 
This gives a possibility to extend the definition of Ramanujan to irregular graphs, requiring no eigenvalue of magnitude within $(\sqrt{\rho}, \rho)$ for its nonbacktracking matrix, which we refer to as NB-Ramanujan.
In this context, it was proved in \cite{bordenave} that ER$(n,c)$ is NB-Ramanujan with high probability. 
Further, \cite{bordenave} also shows that NB achieves our objective (2). However, the nonbacktracking matrix has some downsides:
(i) it is asymmetric, and some of the intuition from spectral graph theory is lost;
(ii) it is not robust to cliques and tangles, and fails to achieve point (3) as discussed below.

Thus we will look for other solutions than the NB operator that achieve (2) and (3), but this will require us to revisit what a ``maximal'' spectral gap means for objective (1). An intuitive notion is that $\psi(A)$ is a good graph regularizer if it gives the same type of spectral gap whether $A$ is random regular (hence already regularized) or from the ER ensemble (subject to local density variations).

(2) The stochastic block model is a random graph with planted clusters, and the weak recovery problem consists of recovering the clusters with a positive correlation. In the case of two symmetric clusters, as considered here, the fundamental limit for efficient algorithms is the so-called Kesten-Stigum threshold. Recently, a few algorithms have been proved to achieve this threshold \cite{massoulie-STOC,Mossel_SBM2,bordenave,colin3}. However, these are not robust to cliques and tangles as mentioned earlier and discussed next. 

(3) Diverse versions of geometric block models have recently been studied to integrate the presence of short loops in the graph \cite{our_grid,abishek2,arya,powering1} (see \cite{abbe_fnt} for furthe references). We use here a Gaussian-mixture block model (see Section \ref{gbm_section}). This model has the advantage of having a simple expression for the weak recovery threshold (conjectured here), allowing us to evaluate whether the proposed methods are optimal or not. In particular, this simple model breaks the algorithms of \cite{massoulie-STOC,Mossel_SBM2,bordenave,colin3,colin3cpam}, making eigenvectors concentrate possibly on tangles (i.e., dense regions with cliques intertwined) --- see Figure \ref{gbm-comp}.

\begin{figure}[h]
\begin{center}
\includegraphics[scale=.4]{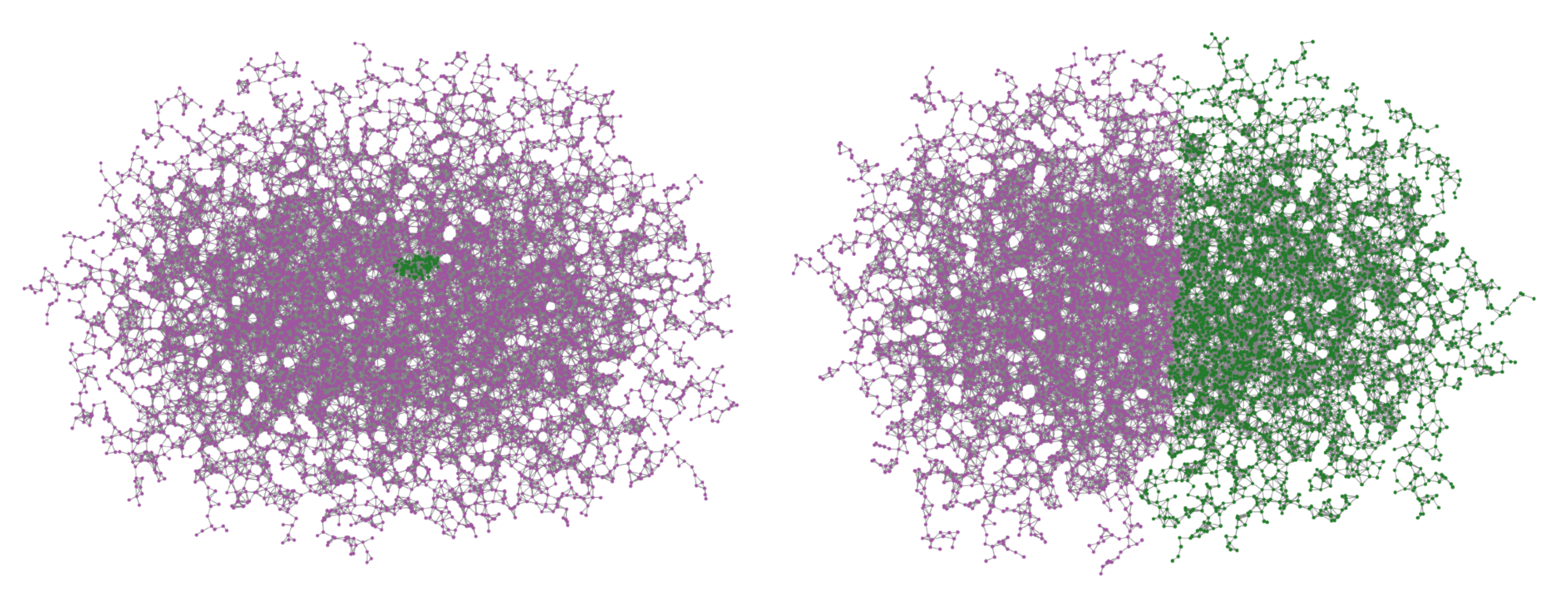}
\caption{A graph drawn from the GBM$(n,t,s)$ (see definition in Section \ref{gbm_section}), where $n/2$ points are sampled i.i.d.\ from an isotropic Gaussian in dimension 2 centered at $(-s/2,0)$ and $n/2$ points are sampled i.i.d.\ from an isotropic Gaussian in dimension 2 centered at $(s/2,0)$, and any points at distance less than $t/\sqrt{n}$ are connected (here $n=10000$, $s=2$ and $t=10$).  In this case the spectral algorithm on the NB matrix gives the left plot, which puts a small fraction of densely connected vertices (a tangle) in a community, and all other vertices in the second community. Graph powering gives the right plot in this case.}
\label{gbm-comp}
\end{center}
\end{figure}
 
\subsection{Graph powering}
The solution proposed in this paper is as follows. 

\begin{definition}[Graph powering]
We give two equivalent definitions: 
\begin{itemize}
\item Given a graph $G$ and a positive integer $r$, define the $r$-th graph power of $G$ as the graph $G^{(r)}$ with the same vertex set as $G$ and where two vertices are connected if there exists a path of length $\le r$ between them in $G$. 

Note: $G^{(r)}$ contains the edges of $G$ and adds edges between any two vertices at distance $\le r$ in $G$. Note also that one can equivalently ask for a walk of length $\le r$, rather than a path. 

\item If $A$ denotes the adjacency matrix of $G$ with 1 on the diagonal (i.e., put self-loops to every vertex of $G$), then the adjacency matrix $A^{(r)}$ of $G^{(r)}$ is defined by 
\begin{align}
A^{(r)} = \mathbb{1} (A^r \ge 1).
\end{align}
Note: $A^r$ has the same spectrum as $A$ (up to rescaling), but the action of the non-linearity $\mathbb{1} (\cdot \ge 1)$ gives the key modification to the spectrum. 
\end{itemize}
\end{definition}

\begin{definition} For a graph $G$, an $r$-power-cut in $G$ corresponds to a cut in $G^{(r)}$, i.e., 
\begin{align}
\partial_r(S) = \{ u \in S, v \in S^c: (A^{(r)})_{u,v}=1\}, \quad S \subseteq V(G).
\end{align}
\end{definition}

Note that powering is mainly targeted for sparse graphs, and is useful only if the power $r$ is not too small and not too large. If it is too small, the powered graph may not be sufficiently different from the underlying graph. If it is too large, say $r\ge \mathrm{diameter}(G)$, then powering turns any graph to a complete graph, which destroys all the useful information. However, powering with $r$ less than the diameter but large enough will be effective on the benchmarks (1)-(2)-(3). As a rule of thumb, one may take $r=\sqrt{\diam(G)}$.
We now discuss the main insight behind graph powering.\\

{\it (1) Power-cuts as Bayes-like cuts.}
The spectral algorithms based on $A$, $L$ or $L_{\mathrm{norm}}$ can be viewed as relaxations of the MAP estimator, i.e., the min-bisection:
\begin{align}
\max_{x \in \{-1,+1\}^n,\, x^t 1^n = 0} x^t A x .
\end{align}
However, the MAP estimator is the right benchmark only when aiming to maximize the probability of recovering the entire communities. It is not the right objective in the regime where one can only partially recover the communities, which is the sparse regime of interest to us in this paper. We illustrate this distinction on the following example --- see also Figure \ref{bad-tree}.

\begin{figure}[h]
\begin{center}
\includegraphics[scale=.38]{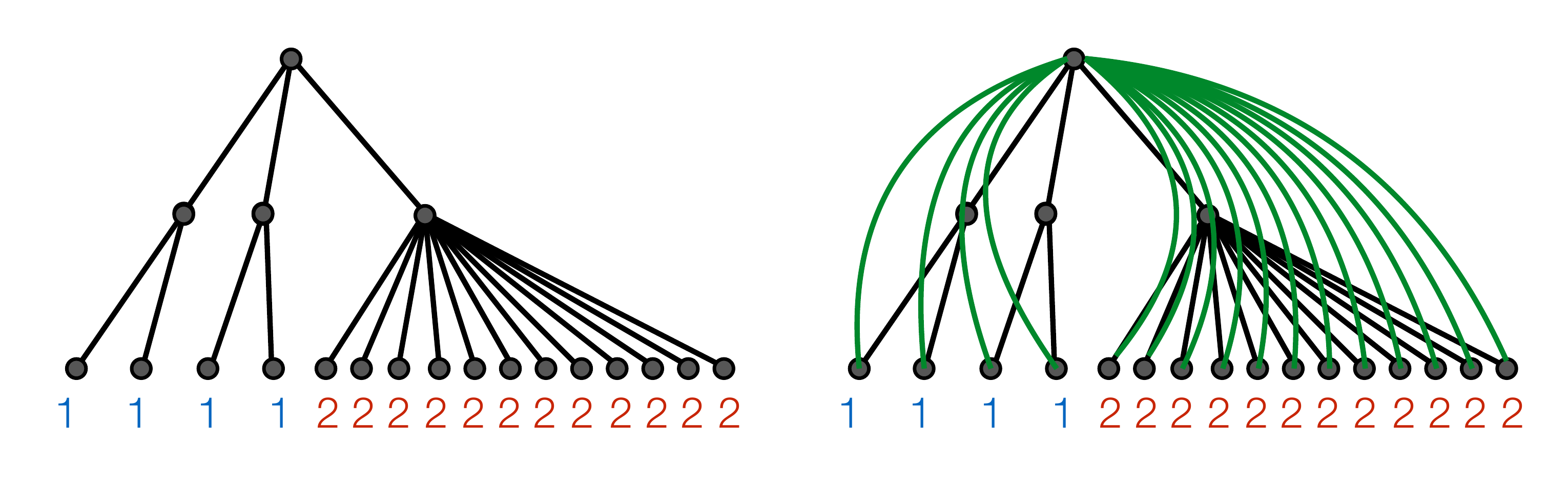}
\caption{In the left graph, assumed to come from $\ssbm(n,2,3/n,2/n)$, the root vertex is labelled community 1 from the ML estimator given the leaf labels, which corresponds to the min-cut around that vertex. 
In contrast, the Bayes optimal estimator puts the root vertex in community 2, as the belief of its right descendent towards community 2 is much stronger than the belief of its two left descendents towards community 1. This corresponds in fact to the min-power-cut obtained from the right graph, where 2-power edges are added by graph powering (note that only a subset of relevant edges are added in the Figure).}
\label{bad-tree}
\end{center}
\end{figure}

Imagine that a graph drawn from $\ssbm(n,2,3/n,2/n)$ contained the following induced subgraph: $v_0$ is adjacent to $v_1$, $v_2$, and $v_3$; $v_1$ and $v_2$ are each adjacent to two outside vertices that are known to be in community $1$, and $v_3$ is adjacent to a large number of vertices that are known to be in community $2$. Then $v_1$ and $v_2$ are more likely to be in community $1$ than they are to be in community $2$, and $v_3$ is more likely to be in community $2$ than it is to be in community $1$. So, the single most likely scenario is that $v_0$, $v_1$, and $v_2$ are in community $1$ while $v_3$ is in community $2$. In particular, this puts $v_0$ in the community that produces the sparsest cut (1 edge in the cut vs 2 edges in the other case). However, $v_3$ is almost certain to be in community $2$, while if we disregard any evidence provided by their adjacency to $v_0$, we would conclude that $v_1$ and $v_2$ are each only about $69\%$ likely to be in community $1$. As a result, $v_0$ is actually slightly more likely to be in community $2$ than it is to be in community $1$.

Power-cuts precisely help with getting feedback from vertices that are further away, making the cuts more ``Bayes-like'' and less ``MAP-like,'' as seen in the previous example where $v_1$ is now assigned to community $2$ using $2$-power-cuts rather than community $1$ using standard cuts. Note that this is also the case when using self-avoiding or nonbacktracking walk counts, however these tend to overcount in the presence of loops. For example, in the graph of Figure \ref{counts}, the count of NB walks is doubled around the 4-cycle, and this gets amplified in models like the GBM as discussed earlier (recall Figure \ref{counts}); in contrast, graph powering projects the count back to 1 thanks to the non-linearity $\mathbb{1} (\cdot \ge 1)$.\\

\begin{figure}[h]
\begin{center}
\includegraphics[scale=.38]{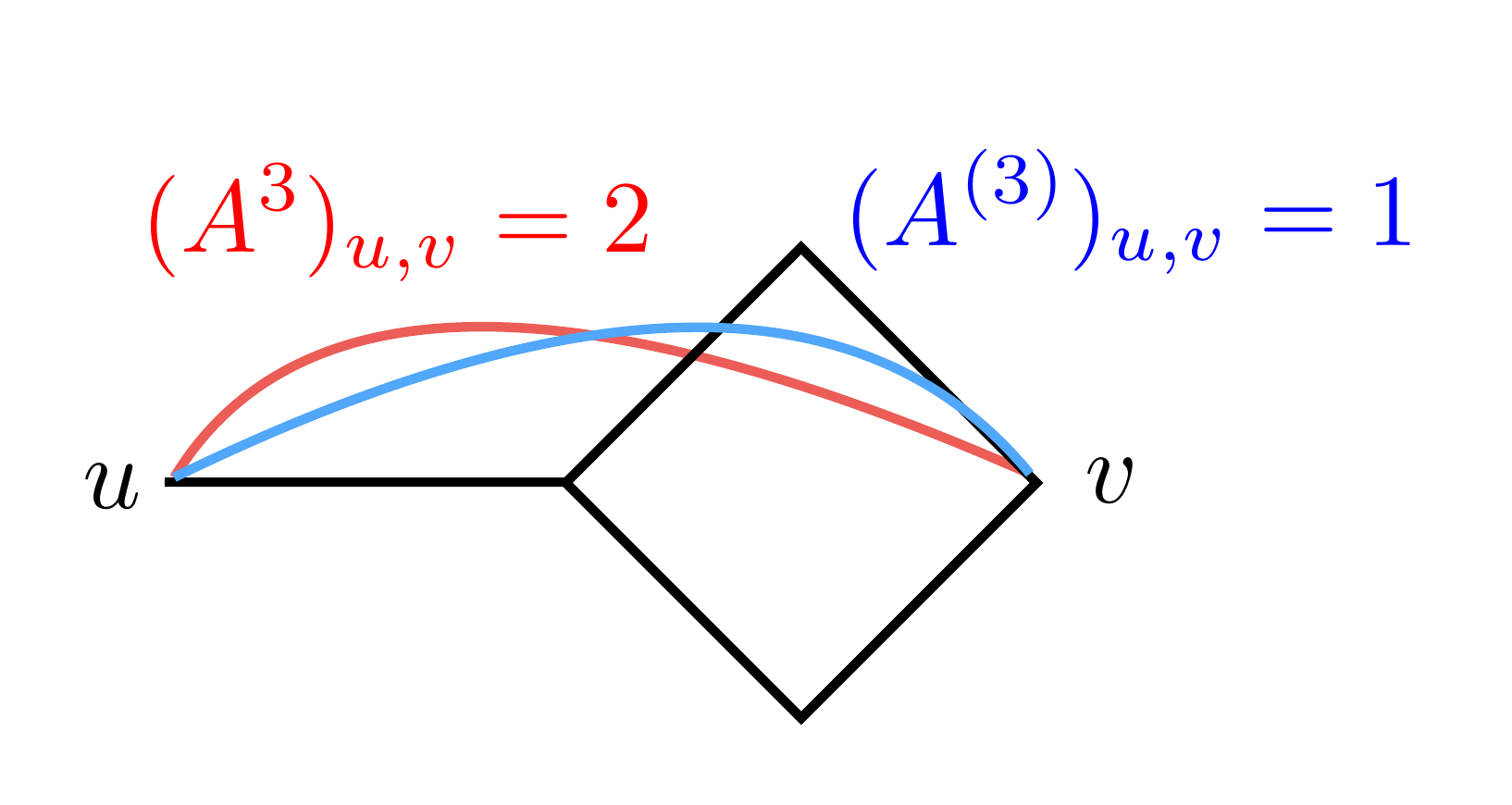}
\caption{Classical powering $A^3$ v.s.\ graph powering $A^{(3)}$. In this example, the number of nonbacktracking walks of length $3$ between $u$ and $v$ is 2 and not 1 as for graph powering.}
\label{counts}
\end{center}
\end{figure}

{\it (2) Powering to homogenize the graph.} Powering helps to mitigate degree variations, and more generally density variations in the graph, both with respect to high and low densities. Since the degree of all vertices is raised with powering, density variations in the regions of the graph do not contrast as much. Large degree vertices (as in Figure \ref{high-degree}) do not stick out as much and tails (as in Figure \ref{lap-cut}) are thickened, and the more macroscopic properties of the graph can prevail.\\

One could probably look for a non-linearity function that is `optimal' (say for the agreement in the SBM) rather than $\mathbb{1} (\cdot \ge 1)$; however this is likely to be model dependent, while this choice seems both natural and generic. A downside of powering is that it densifies the graph, so one would ideally combine graph powering with degree normalizations to reduce the power (since powering raises the density of the graph, normalization may no longer have the issues mentioned previously with tails) or some graph sparsification (such as \cite{sparsify}). 

Note that powering and sparsifying do not counter each other: powering adds edges to ``complete the graph'' in sparse regimes where edges should be present, while sparsifying prunes down the graph by adding weights on representative edges. Finally, one may also peel out leaves and paths, and use powered Laplacians to further reduce the powering order; see Section \ref{implementations} for further discussions on the implementations.

%TODO: another possibility is distance matrix, but less intuitive..

\section{Main results}\label{results_section}
We believe that the recent algorithms \cite{massoulie-STOC,Mossel_SBM2,bordenave,colin3} proved to achieve objective (2) fail on objective (3); 
this is backed in Section \ref{derivation}.
SDPs are likely to succeed in (3); they are already shown to be robust to a vanishing fraction of edge perturbations in \cite{montanari_sen}, but they do not achieve (2) \cite{montanari_sen,ankur_SBM,adel_SBM}, and are more costly computationally.
It is also unclear how weak the notion of ``weak Ramanujan'' from \cite{massoulie-STOC} for the matrix of self-avoiding walk counts is, i.e., could the spectral gap be larger for such operators? 

Our first contribution in this paper is to provide a principled derivation on how one can achieve (2) and (3) simultaneously. This is covered in Section \ref{derivation}, leading to graph powering. We then make more precise the notion of ``achieving a maximal spectral gap on Erd\H{o}s-R\'enyi graphs'', by establishing an Alon-Boppana result for graph powering. This will allow us to conclude that graph powering achieves (1), up to a potential exponent offset on the logarithmic factor. 
Finally, we show that graph powering achieves objective (2) and give evidence that graph powering achieves objective (3) in Section \ref{gbm_section}, without providing formal proofs for the latter\footnote{Formal proofs would require in particular percolation estimates that depart from the current focus of the paper.}.

We now state the results formally.

\subsection{Weak recovery in stochastic block models}
We consider different models of random graphs with planted clusters. 
In each case, an ensemble $M(n)$ provides a distribution on an ordered pair of random variables $(X,G)$,
where $X$ is an $n$-dimensional random vector with i.i.d.\ components,
corresponding to the community labels of the vertices,
and $G$ is an $n$-vertex random graph,
connecting the vertices in $V=[n]$ depending on their community labels. 
%In the geometric block model defined in Section \ref{}, we will also have auxiliary random variables attached to the vertices. 
The goal is always to recover $X$ from $G$, i.e., to reconstruct the communities by observing the connections. 
The focus of this paper is on the sparse regime and the weak recovery problem, defined below. 

\begin{definition}
In the case of $k$ communities, an algorithm $\hat{X}:2^{[n] \choose 2} \to [k]^n$ recovers communities with accuracy $f(n)$ in $M(n)$ if, for $(X,G) \sim M(n)$ and $\Omega_i:=\{v \in [n]: X_v= i\}$, $i \in [k]$,
\begin{align}
\pp\left\{\max_{\pi\in S_k} \frac{1}{k}\sum_{i=1}^k \frac{|\{v\in\Omega_i: \pi(\hat{X}_v)=i\}|}{|\Omega_i|}\geq f(n)\right\} = 1-o(1),
\end{align} 
We say that an algorithm solves weak recovery if it recovers communities with accuracy $1/k+\Omega(1)$.
\end{definition}

\begin{definition}
Let $n$ be a positive integer (the number of vertices), $k$ be a positive integer (the number of communities), $p=(p_1,\dots,p_k)$ be a probability vector on $[k]:=\{1,\dots,k\}$ (the prior on the $k$ communities) and $W$ be a $k\times k$ symmetric matrix with  entries in $[0,1]$ (the connectivity probabilities). The pair $(X,G)$ is drawn under $\mathrm{SBM}(n,p,W)$ if $X$ is an $n$-dimensional random vector with i.i.d.\ components distributed under $p$, and $G$ is an $n$-vertex simple graph where vertices $i$ and $j$ are connected with probability $W_{X_i,X_j}$, independently of other pairs of vertices. We also define the community sets by $\Omega_i=\Omega_i(X) : = \{v \in [n] : X_v = i\}, i \in [k].$
\end{definition}

In this paper, we refer to the above as the ``general stochastic block model'' and use ``stochastic block model'' for the version with two symmetric and sparse communities, as follows.

\begin{definition}
$(X,G)$ is drawn under $\mathrm{SBM}(n,a,b)$ if $k=2$, $W$ takes value $a/n$ on the diagonal and $b/n$ off the diagonal, and if the community prior is $p=\{1/2,1/2\}$. That is, $G$ has two balanced clusters, with edge probability $a/n$ inside the clusters and $b/n$ across the clusters.  
\end{definition}
Note that in this case, if defining $X, \hat{X}$ to take values in $\{-1,+1\}^n$, we have that weak recovery is solvable if and only if 
\begin{align}
|\langle X,\hat{X}(G) \rangle | =  \Omega(1)
\end{align}
with high probability.
\begin{definition}
We denote by $\mathrm{ER}(n,d)$ the distribution of Erdos-Renyi random graphs with $n$ vertices and edge probability $d/n$. 
\end{definition}

The following was conjectured in \cite{massoulie-STOC}; we prove it here using results from \cite{massoulie-STOC} itself. 

\begin{theorem}[Spectral separation for the distance matrix; conjectured in \cite{massoulie-STOC}]\label{ks1}
Let $(X,G)$ be drawn from $\sbm(n,a,b)$ with $(a+b)/2 > 1$.
Let $A^{[r]}$ be the $r$-distance matrix of $G$ ($A^{[r]}_{ij}=1$ if and only if $d_G(i,j) = r$), and $r= \e \log(n)$ such that $\e > 0$, $\e \log (a+b)/2 <1/4$. 
Then, with high probability, 

\begin{multicols}{2}
\begin{enumerate}
\item[A.] If $\left(\frac{a+b}{2}\right) < \left(\frac{a-b}{2}\right)^2$, then
\begin{enumerate}
\item[1.]  $\lambda_1(A^{[r]}) \asymp \left(\frac{a+b}{2} \right)^{r}$,
\item[2.] $\lambda_2(A^{[r]}) \asymp \left(\frac{a-b}{2} \right)^{r}$,
\item[3.] $|\lambda_3(A^{[r]})| \leq \left(\frac{a+b}{2} \right)^{r/2} \log(n)^{O(1)}$.
\end{enumerate}
\item[B.]If $\left(\frac{a+b}{2}\right) > \left(\frac{a-b}{2}\right)^2$, then
\begin{enumerate}
    \item[1.] $\lambda_1(A^{[r]}) \asymp \left(\frac{a+b}{2} \right)^{r}$,
    \item[2.] $|\lambda_2(A^{[r]})| \leq \left(\frac{a+b}{2} \right)^{r/2} \log(n)^{O(1)}$.
    \item[~]
\end{enumerate}
\end{enumerate}
\end{multicols}

Furthermore, $\phi_2(A^{[r]})$ with the rounding procedure of \cite{massoulie-STOC} achieves weak recovery whenever $\left(\frac{a+b}{2}\right) < \left(\frac{a-b}{2}\right)^2$, i.e., down to the KS threshold.
\end{theorem}

The following theorem is identical to the previous one but holds for the graph powering matrix.  
\begin{theorem}[Spectral separation for graph powering]\label{ks2}
Let $(X,G)$ be drawn from $\sbm(n,a,b)$ with $(a+b)/2 > 1$.
Let $A^{(r)}$ be the adjacency matrix of the $r$-th graph power of $G$, and $r= \e \log(n)$ such that $\e > 0$, $\e \log (a+b)/2 <1/4$. 
Then, with high probability,

\begin{multicols}{2}
\begin{enumerate}
\item[A.] If $\left(\frac{a+b}{2}\right) < \left(\frac{a-b}{2}\right)^2$, then
\begin{enumerate}
\item[1.]  $\lambda_1(A^{(r)}) \asymp \left(\frac{a+b}{2} \right)^{r}$,
\item[2.] $\lambda_2(A^{(r)}) \asymp \left(\frac{a-b}{2} \right)^{r}$,
\item[3.] $|\lambda_3(A^{(r)})| \leq \left(\frac{a+b}{2} \right)^{r/2} \log(n)^{O(1)}$.
\end{enumerate}
\item[B.]If $\left(\frac{a+b}{2}\right) > \left(\frac{a-b}{2}\right)^2$, then
\begin{enumerate}
    \item[1.] $\lambda_1(A^{(r)}) \asymp \left(\frac{a+b}{2} \right)^{r}$,
    \item[2.] $|\lambda_2(A^{(r)})| \leq \left(\frac{a+b}{2} \right)^{r/2} \log(n)^{O(1)}$.
    \item[~]
\end{enumerate}
\end{enumerate}
\end{multicols}

Furthermore, $\phi_2(A^{(r)})$ with the rounding procedure of \cite{massoulie-STOC} achieves weak recovery whenever $\left(\frac{a+b}{2}\right) < \left(\frac{a-b}{2}\right)^2$, i.e., down to the KS threshold.
\end{theorem}
Note that one may also take smaller values of $r$, such as $r=\sqrt{\diam(G)}$, or $r$ that is $\omega(\log\log(n))$ to surpass the high degree vertices, and $o(\log(n))$ (or a small enough constant on the log as in the theorem) to avoid powering too much. We refer to Section \ref{implementations} for discussion of the choice of $r$.

%\begin{corollary}
%Outputting $\phi_A(2)$ or $\phi_B(2)$ solves weak recovery whenever $(a-b)/2 > \sqrt{(a+b)/2}$, i.e., down to the KS threshold.
%\end{corollary}

\subsection{Alon-Boppana for graph powering}
We now investigate how large the spectral gap produced by graph powering on ER random graphs is. Note first the following statement, which follows from Theorem \ref{ks2} by setting $a=b$.
\begin{corollary}\label{corol_ks2}
Let $G$ be drawn from ER$(n,d)$, $A^{(r)}$ be the adjacency matrix of the $r$-th graph power of $G$ and $r= \e \log(n)$ with $\e>0$, $\e \log (a+b)/2 <1/4$. Then, with high probability, 
\begin{enumerate}
\item  $\lambda_1(A^{(r)}) \asymp d^{r}$,
\item $|\lambda_2(A^{(r)})| \leq \sqrt{d}^{\, r} \log(n)^{O(1)}$.
\end{enumerate}
\end{corollary}

Let us compare this gap to the one of Ramanujan graphs. Recall first that the Alon-Boppana result \cite{nilli} for $d$-regular graphs gives
\begin{enumerate}
\item  $\lambda_1(A) =d$, 
\item $\lambda_2(A) \ge (1-o_{\diam(g)}(1)) 2 \sqrt{d-1}$.
\end{enumerate}
and Ramanujan graphs, or more precisely Ramanujan families of graphs, achieve the lower-bound of $2 \sqrt{d-1}$. 
Their existence is known from \cite{lps}, and Friedman \cite{friedman} proved that random $d$-regular graphs are almost Ramanujan, i.e., with high probability their second largest eigenvalue satisfies
\begin{align}
\lambda_2(A) \le  2 \sqrt{d-1} + o(1).
\end{align}  
To argue that powering turns ER$(n,d)$ into a graph of maximal spectral gap (factoring out its irregularity), we first need to understand how large of a spectral gap the powering of any regular graph can have. 
Powering regular graphs may not necessarily produce regular graphs. 
So we cannot apply directly the Alon-Boppana bound.
Yet, let us assume for a moment that the $r$-th power of a $d$-regular graph is regular with degree $d^r$, then Alon-Boppana would give that the second largest eigenvalue of a $r$-powered graph is larger than $2 \sqrt{d^r-1} \sim 2 \sqrt{d}^{\, r}$. In contrast, our Corollary \ref{corol_ks2} gives that the $r$-th power of ER$(n,d)$ has its second largest eigenvalue of order $\sqrt{d}^{\, r} \log(n)^{O(1)}$. 

This additional logarithmic factor could suggest that powering may not give a tight generalization of Ramanujan, and, thus, may not be an ``optimal graph regularizer.'' In \cite{massoulie-STOC}, for the $r$-distance rather than the $r$-powered graph, this `slack' is absorbed in the terminology of `weak' Ramanujan. However, how weak is this exactly? The above reasoning does not take into account the fact that $G^{(r)}$ is not any regular graph, but a {\it powered} graph (i.e., not any graph is the power of some underlying graph). So it is still possible that powering must concede more in the spectral gap that the above loose argument suggests; we next show that this is indeed the case.
\begin{theorem}\label{abp}
Let $\{G_n\}_{n \ge 1}$ be a sequence of graphs such that $\diam(G_n)=\omega(1)$, and $\{r_n\}_{n \ge 1}$ a sequence of positive integers such that $r_n=\e \cdot \diam(G_n)$. Then, 
\begin{align}
\lambda_2(G_n^{(r_n)}) \ge (1-o_\e(1)) (r_n+1) \hat{d}_{r_n}^{\, r_n/2}(G_n),
\end{align}
where
\begin{align}
 \hat{d}_r(G) &= \left( \frac{1}{r+1} \sum_{i=0}^r \sqrt{\delta^{(i)}(G)\delta^{(r-i)}(G)} \right)^{2/r},\\
 \delta^{(i)}(G) &= \min_{(x,y) \in E(G)} |\{ v \in V(G): d_G(x,v)=i,d_G(y,v)\ge i \}|.
\end{align}
\end{theorem}

Moreover, we can approximate $\hat{d}_r$ for random $d$-regular graphs.
\begin{lemma}\label{lemma_reg}
Let $G$ be a random $d$-regular graph and $r = \e \log(n)$, where $\e \log d <1/4$. Then, with high probability, 
$$\hat{d}_r(G) = (1+o_d(1)) d.$$
\end{lemma}

%\begin{align}
%\lambda_2(A^{(r)}) \asymp \log(n) \sqrt{c}^{\, r}.
%\end{align}
%TODO: Do we still get this for ER$(n,c)$ when $r = \e \log(n)$, $\e$ small enough? 
Of course, a powered ER graph is not strictly regular, but Theorem \ref{abp} and Lemma \ref{lemma_reg} say that even if we had regularity, we could not hope to get a better bound than that of Corollary \ref{corol_ks2} for large degrees, except for the exponent on the logarithmic factor that we do not investigate.

Note also that $\hat{d}_r(G) = \parens{1+o_d(1)}d$ if $G$ is $d$-regular and $r$ is less than half the girth of $G$. One can then show the following.
\begin{lemma}\label{lemma_girth}
If $2r < \mathrm{girth}(G)$, $G$ is Ramanujan implies that $G^{(r)}$ is $r$-Ramanujan, in the sense that $\lambda_2(G^{(r)})= \parens{1+o_d(1)}(r+1) d^{r/2}$.
\end{lemma}
We do not know whether this is true for general Ramanujan graphs, nor whether the converse holds.

\section{Geometric block model and robustness to tangles}\label{gbm_section}
In this section we provide numerical evidences that graph powering achieves the fundamental limit for weak recovery in the Gaussian-mixture block model (GBM). We also give a possible proof sketch, but leave the result as a conjecture. We start by defining and motivating the model.

\begin{definition}
Given a positive integer $n$ and $s,t \ge 0$, we define $\gbm(n,s,t)$ as the probability distribution over ordered pairs $(X,G)$ as follows. First, each vertex $v \in V$ is independently assigned a community $X_v\in\{1,2\}$ with equal probability (as for the SBM). Then, each vertex $v$ is independently assigned a location $U_v$ in $\mathbb{R}^2$ according to $\mathcal{N}(((-1)^{X_v}s/2,0), I_2)$. Finally, we add an edge between every pair of vertices $u,v$ such that $\|U_u - U_v \| \le t/\sqrt{n}$.
\end{definition}
The parameter $s$ is the separation between the two isotropic Gaussian means, the first centered at $(-s/2,0)$ and the second at $(s/2,0)$.
In order to have a sparse graph with a giant, we work with the scaling $t/\sqrt{n}$ for a constant $t$. There is no known closed form expression for the threshold in order to have a giant, but such a threshold exists \cite{penrose}. We will simply assume that $t$ is above that threshold.

Note that the expected adjacency matrix of the GBM, conditioned on the vertex labels, has the same rank-2 block structure as the SBM: it takes value $\E(A_{ij}|X_i=X_j)$ for all $i,j$ on the diagonal blocks and $\E(A_{ij}|X_i\ne X_j)$ for all $i,j$ on the off-diagonal block. However, the sampled realizations are very different for the GBM. 

The main point of introducing this GBM is to have a simple model that accounts for having many more loops than the SBM does.
The SBM gives a good framework to understand how to cluster sparse graphs with some degree variations, where ``abstract'' edges occur frequently; i.e., when connection is not necessarily based on metric attributes (x can be friends with y for a certain reason, y can be friends with z for a different reason, while x and z have nothing in common). These abstract edges turn the SBM into a sparse graph with small diameter, which is an important feature in various applications, sometimes referred to as the ``small-world'' phenomenon. However, the  opposite effect also takes place when connections are based on metric attributes, i.e., if x and y are similar (or close), y and z are similar too, then x and z must also be similar to some extent. This creates many more short loops than in the SBM. 

One possibile way to study the effect of short loops is to consider adversaries that can modify the graph in certain ways. For example, monotone adversaries that only add edges inside clusters or only remove edges across clusters (for the assortative case), or adversaries having a certain budget of edges they can alter. The problem with such adversaries is that they typically ensure robustness to many fewer cliques than may be observed in applications, due to the classical downside of worst-case models where the possible obstructions may not be typical in applications. The goal of the GBM is to have a simple model for geometric connections and tangles, albeit with the usual downside of generative models.

Note that this version of the geometric block model is different than the version studied recently in \cite{abishek2}. The model of \cite{abishek2} introduces the geometry in a different manner: each vertex has a geometric attribute that does not have any community bias, e.g., a point process on $\mR^2$ or points uniformly drawn on the torus (or it could have been an isotropic Gaussian distribution on $\mR^2$), and each vertex has an independent equiprobable community label; then two vertices connect in proportion to their geometric distance (e.g., only if close enough), and the bias comes from their abstract community label. In contrast, in our GBM, vertices connect only based on their locations, and these locations are encoding the communities. This makes the GBM simpler to analyze than the model of \cite{abishek2}. In fact, we claim the following.    
\begin{conjecture}
Let $s,t \ge 0$ and $t$ such that $\gbm(n,s,t)$ has a giant component. 
Weak recovery is efficiently solvable in $\gbm(n,s,t)$ if and only if $s >0$.
\end{conjecture}
%Note that the assumption in the claim allows for two disconnected giant components, which represents an easy case, or a single giant component merging the two communities. 
\begin{proof}[Justification]
Obviously weak recovery is not solvable if $s=0$, so the claim follows by showing that weak recovery is efficiently solvable if $s>0$. 
A first simple case is when the graph has two separate giant components, in which case we can simply assign each component to a different community. 
When there is a single giant component, one does not need an elaborate algorithm either. 

As a first case, take a vertex uniformly at random among all $n$ vertices and assign that vertex to community 1; then assign the $n/2$ closest vertices to that vertex (in the graph distance) to community 1, and the rest to community 2. With probability $1/2+\Omega(1)$, this will already solve weak recovery. To succeed with probability $1-o(1)$, one can pick two anchor vertices that are ``far away'' in graph-distance, and assign each to a different community and the rest of the vertices to the same community as their closest anchor vertex. 

To make ``far away'' more precise, one could take two vertices of maximal (finite) distance. As soon as $s>0$, these will not be aligned with the $y$-axis giving two appropriate anchor nodes. To prove the result, it may however be easier to focus on the highest degree vertices. If the parameters are such that there are two regions of maximum vertex density in $\mathbb{R}^2$, we could simply pick two very high degree vertices that are far away from each other and use them as anchors. This requires however concentration results for the distances among pair of vertices. 

In fact, if we had a mixture of two uniform distributions of overlapping support, one could use a triangulation algorithm to reconstruct approximately the vertex locations (up to overall rotations and translations), using concentration results for distances in random geometric graphs (cf. \cite{yao2011large}). For example, for a uniform distribution in the unit square, for all pairs of vertices $u$ and $v$ in the giant, the graph distance $d_G(u,v)$ would be well concentrated around $c\|U_u - U_v\|$, for some constant $c$. A variant of this would then be needed for the non-homogeneous (Gaussian) model used here. These techniques are however out of scope for the current paper. 
\end{proof}

We next give with Figure \ref{fig:graphpoweringgbm} numerical results supporting that graph powering achieves the above threshold. 
\begin{figure}[H]
\begin{center}
\includegraphics[scale=.6]{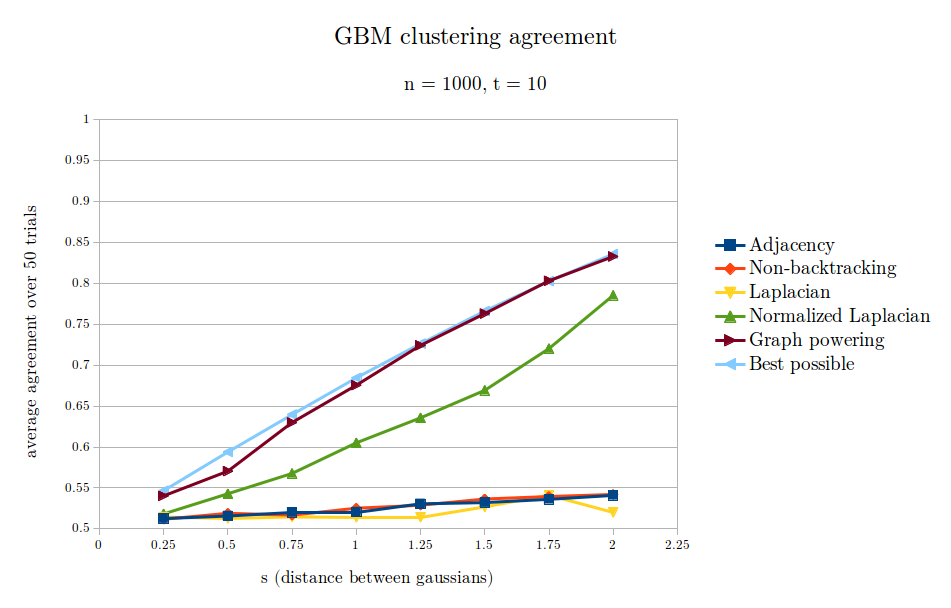}
\caption{The truth-agreement values of clusters calculated by different methods on $G \sim \gbm(n,s,t)$ with $n = 1000$, $t = 10$. Each entry is the average agreement over 50 runs. The spectral clustering on the powered adjacency matrix has powering parameter $r = 0.3 d$, where $d$ is the diameter. ``Best possible'' indicates the information-theoretic upper bound for agreement, if the locations of the vertices in the plane were known to the algorithm.} \label{fig:graphpoweringgbm}
\end{center}

\end{figure}

%\begin{table}[H] \label{tab:graphpoweringgbm}
%\begin{center}
%\begin{tabular}{ |c||c|c|c| } 
%\multicolumn{4}{c}{Agreement Values on $\gbm$} \\
% \hline
%  $s$& Adjacency & Graph powering & Best possible \\ 
%\hline
%0.25 &	0.513 & 0.518 &	0.534 &	0.547 \\
%\hline
%0.50 &	0.516 &	0.573 &	0.594 \\
%\hline
%0.75 &	0.520 &	0.625 &	0.640 \\
%\hline
%1.00 &	0.521 &	0.680 &	0.684 \\
%\hline
%1.25 &	0.531 &	0.723 &	0.726 \\
%\hline
%1.50 &	0.532 &	0.751 &	0.766 \\
%\hline
%1.75 &	0.536 &	0.796 &	0.803 \\
%\hline
%2.00 &	0.541 &	0.834 &	0.836 \\
%
%Randwalk
%
%0.543
%0.568
%0.605
%0.636
%0.669
%0.720
%0.785
%
%
% \hline
%\end{tabular}
%\end{center}
%\end{table}

\begin{conjecture}\label{conj_gbm}
%Let $G\sim \gbm(n,s,t/\sqrt{n})$.
Let $s,t \ge 0$ and $t$ such that $\gbm(n,s,t)$ has a giant component. 
Taking the second largest eigenvector of the powered adjacency matrix $A^{(r)}$ with $r=\e \cdot \diam(G)$, for $\e$ small enough, and 
rounding it by the sign or by the median (dividing the vertices into those with above- and below-median sums of the entries) solves weak recovery in $\gbm(n,s,t/\sqrt{n})$ whenever weak recovery is solvable.
\end{conjecture}
We refer to Section \ref{justifications} for a proof plan for this conjecture, and to Section \ref{implementations} for discussions on how to reduce $r$, such as to $r=\sqrt{\diam(G)}$.

\section{Comparisons of algorithms}
In this section, we compare some of the main algorithms for community detection on the SBM and GBM. We also introduce a hybrid block model (HBM), which superposes an SBM and a GBM. The HBM has the advantage of simultaneously having a short diameter, having abundant tangles, and being sparse.  We then illustrate the fact that graph powering is robust to the superposition of both types of edges (geometric and abstract edges), while other algorithms suffer on either of the other types of edges. 

%We could also change the probability distribution that determines what location a vertex is assigned to in order to alter the distribution of vertex degrees or otherwise influence the resulting graph. However, in any geometric block model, essentially every edge that is not a cut edge is in a small cycle, which may be unrealistic in the other direction. We thus superpose the SBM and GBM as follows.

\begin{definition}
Given a positive integer $n$ and $a,b,s,t,h \ge 0$, we define $\hbm(n,a,b,s,t,h_1,h_2)$ as the probability distribution over ordered pairs $(X,G)$ as follows. Let $(X,G_1) \sim \sbm(n,a,b)$, and $G_2$ drawn from $\gbm(n,s,t)$ with $X$ for the community labels, such that $G_1,G_2$ are independent conditionally on $X$. For each pair of vertices $u,v$, independently keep the edge from $G_1$ with probability $h_1$ and the edge from $G_2$ with probability $h_2$, and merge the edges if both are kept. Call the resulting graph $G$.
\end{definition}

\begin{table}[H] \label{tab:graphpoweringgbm}
\begin{center}
\begin{tabular}{ |c|c| } 
\multicolumn{2}{c}{Agreement Values on $\hbm(n,a,b,s,t,h_1,h_2)$} \\
\hline
Method &	Agreement (average over 50 trials) \\
\hline
Adjacency &	0.508 \\
\hline
Non-backtracking &	0.508 \\
\hline
Laplacian &	0.507 \\
\hline
Normalized Laplacian &	0.630 \\
\hline
Graph powering &	\textbf{0.666} \\
 \hline
\end{tabular}
\caption{The $\hbm(n,a,b,s,t,h_1,h_2)$ parameters are $n = 4000$, $s = 1$, $t = 10$, $h_1 = h_2 = 0.5$, $a = 2.5$, $b = 0.187$. The graph powering parameter is $r = 0.4 d$, where $d$ is the diameter.}
\end{center}

\end{table}

% \begin{table}[H] \label{tab:graphpoweringhbm}
% \begin{center}
% \begin{tabular}{ |c|c|c|c|c| } 
% \multicolumn{5}{c}{Agreement Values on $\sbm(n,a,b)$} \\
% \hline
% & $a = 2.5,$ &	$a = 3,$ &	$a = 3.5,$ &	$a = 4,$ \\
% & $b = 0.187$ &	$b = 0.399$ &	$b = 0.632$ &	$b = 0.882$ \\
% \hline
% \hline
% Adjacency	 & 0.743 &	0.623 &	0.818 &	0.887 \\
% \hline
% Laplacian &	0.602 &	0.806 &	0.532 &	0.515 \\
% \hline
% Normalized laplacian &	0.901 &	0.912 &	0.839 &	0.929 \\
% \hline
% Non-backtracking &	0.908 &	0.885 &	0.922 &	0.933 \\
% \hline
% Graph powering &	0.940 &	0.942 &	0.937 &	0.940 \\
% \hline
% \end{tabular}
\begin{figure}[H]
\begin{center}
\includegraphics[scale=.6]{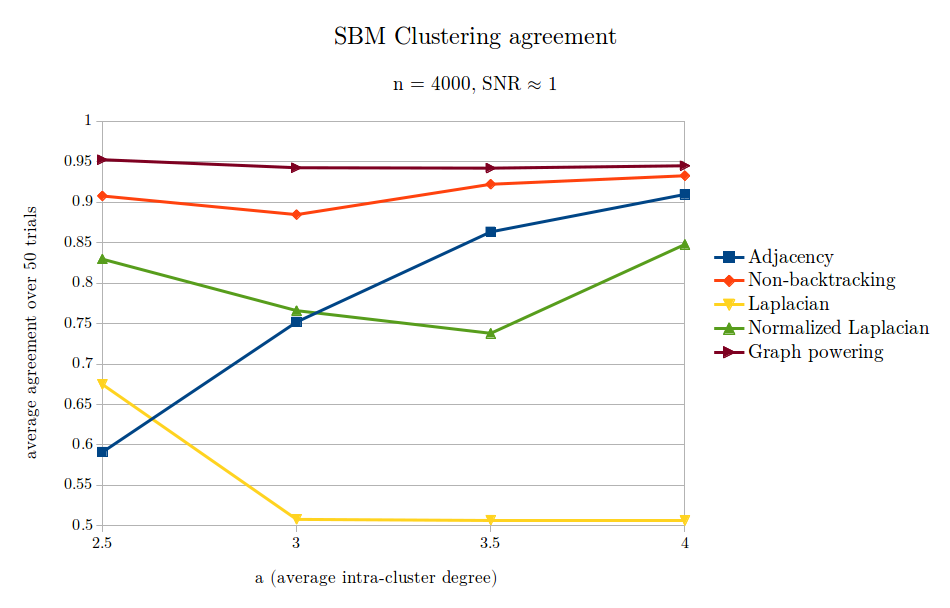}
\caption{The $\sbm(n,a,b)$ parameters are $n = 4000$, $a$ given by the $x$-axis, and $\snr = (a-b)^2 / (2(a+b))$ chosen to be very close to $1$ (the KS-threshold for weak recovery). The graph powering parameter is $r = 0.5 d$, where $d$ is the diameter. Although the normalized Laplacian and the adjacency matrix may seem to cluster the graph relatively well, when $n$ is increased the performance of these methods drops dramatically. For instance, when $n = 100000$, the normalized Laplacian method consistently scores below 0.51 agreement on the parameter ranges in the above graph. Similarly, for $n = 100000$, $\snr \approx 1$ and $a \leq 3$, the adjacency method scores below 0.55 agreement.}\label{sbm_num}
\end{center}

\end{figure}
% \end{center}

% \end{table}

\section{Derivation of graph powering}\label{derivation}

\subsection{Bayes vs.\ ML and beliefs vs.\ cuts}

The most obvious approach to community recovery when we have a generative model would be to simply return the {\it most likely} communities given the graph in question, which reduces to finding a minimal balanced cut in the symmetric SBM. However, this does not always produce good results as demonstrated by the example discussed in Figure \ref{bad-tree}.  Fundamentally, focusing on the single most likely community assignment forces one to treat each vertex as being $100\%$ likely to be in the designated community when attempting to determine the communities of other vertices, and that is not the right way to handle uncertainty. To the degree that our goal is to recover communities with some given degree of accuracy, we want to assign the vertices to communities in the way that maximizes the probability that we get at least that fraction of them right. Typically, we can come fairly close to that by determining how likely each vertex is to be from each community, and then assigning it to the most likely community. More formally, we could reasonably want to guess a community $x_v$ for each $v$ in such a way as to maximize $\pp(X_v=x_v|G=g)$.

To further motivate this approach, consider a weakly symmetric SBM, i.e., a two-community SBM with two communities drawn from a prior $(p_1,p_2)$, $p_1 \ne p_2$, and connectivity $Q/n$ such that $PQ$ has two rows with the same sum. In other words, the expected degree of every vertex is the same (and weak recovery is non-trivial), but the model is slightly asymmetrical and one can determine the community labels from the partition. In this case, we can work with the agreement between the true labels and the algorithm reconstruction without use of the relabelling $\pi$, i.e., 
\begin{align}
A(X,\hat{X}(G))= \sum_{v=1}^n \1(X_v = \hat{X}_v(G)).
\end{align}
Consider now an algorithm that maximizes the expected agreement, i.e, 
\begin{align}
\E A(X,\hat{X}(G))= \sum_{v=1}^n \pp (X_v = \hat{X}_v(G)).
\end{align}
To solve weak recovery, one needs a non-trivial expected agreement, and to maximize the above, one has to maximize each term given by 
\begin{align}
\pp (X_v = \hat{X}_v(G))= \sum_{g}\pp (X_v = \hat{X}_v(g)|G=g)  \pp (G=g),
\end{align}
i.e., $\hat{X}_v(g)$ should take the maximal value of the function $ x_v \mapsto \pp (X_v = x_v |G=g)$. In other words, we need the marginal $\pp (X_v = \cdot |G=g)$. Note that in the symmetric SBM, i.e., $\sbm(n,a,b)$, this marginal is $1/2$. Hence we need to break the symmetry between communities in order for these marginals to be meaningful.

%\begin{remark}
%Strictly speaking, if $(X,G)\sim\sbm(n,a,b)$ then every vertex is equally likely to be from either community because the probability distribution is symmetric under relabeling of communities. To make nontrivial statements about the probability that a given vertex is in community $1$, we first have to break the symmetry. 
Technically, any method of deciding which community is which given a partition of the vertices into communities would break the symmetry, but not all of them break it in a useful way. To illustrate this, let $(X,G)\sim\sbm(n,10,0)$ for some odd $n$. The overwhelming majority of the vertices in $G$ are in one of two giant components that each contain solely vertices from one community. As such, we should be able to declare the community of the vertices from one of these components community $1$ in order to assign every vertex in a giant component to the correct community with probability $1-o(1)$. However, if we were to declare that community $1$ was the community with an odd number of vertices, then we would not be able to determine which of the giant components contained vertices from community $1$ due to uncertainty about which community the isolated vertices were in. As a result, every vertex would be in community $1$ with probability $1/2+o(1)$. We could declare that community $1$ was the larger community. At that point, the vertices from the largest component would be in community $1$ with probability $1/2+\Omega(1)$. However, there would be a nonvanishing probability that enough of the isolated vertices were in the same community as the smaller giant component to make that giant component the one with vertices from community $1$. So, every vertex would be in community $1$ with probability $1-\Omega(1)$ and community $2$ with probability $1-\Omega(1)$. 

The point is, in any case where community detection is possible, there is a partially known partition of the vertices into communities. However, some methods of deciding which set from the partition should be given which label result in the uncertainty in the partition leading to uncertainty in which set from the partition is which community. One reasonable approach to breaking the symmetry would be to pick the $\sqrt{n}$ vertices of highest degree, and then declare that community $1$ was the community that contained the largest number of these vertices. The fact that these vertices have high degrees means that they have large numbers of edges that provide information on their communities. As a result, we would hopefully be able to determine which of them were in the same communities as each other with accuracy $1-o(1)$. That in turn would tend to allow us to determine which of them were in community $1$ with an accuracy of $1-o(1)$. Finally, that would fix which community was community $1$ in a way that would be unlikely to add significant amounts of ambiguity about other vertices' communities. When applicable for the rest of this paper, we will assume that the symmetry has been broken in a way that does not introduce unnecessary uncertainty.
%\end{remark}

\subsection{Belief propagation and nonbacktracking linearization}
One approach to estimate the marginal $P[X_{v_0}=x_{v_0}|G=g]$ is to proceed iteratively starting from some initial guesses. We will derive this for the general SBM, as some aspects are more easily explained in the general context. The main ideas in this section already appear in papers such as \cite{redemption}; we provide here a slightly different and more detailed account, including the symmetry breaking and the influence of non-edges in details.       

\begin{lemma}
Let $(X,G)\sim\sbm(n,p,Q/n)$, $v_0\in G$, and $g$ be a possible value of $G$. Let $v_1,... v_m$ be the vertices that are adjacent to $v_0$, and define the vectors $q_1,...,q_m$ such that for each community $i$ and vertex $v_j$, 
\[(q_j)_i=P[X_{v_j}=i|G\backslash \{v_0\}=g\backslash \{v_0\}].\]
If the probability distributions for $X_{v_1}, X_{v_2},..., X_{v_m}$ ignoring $v_0$ are nearly independent in the sense that for all $x_1,...,x_m$ we have that
\begin{align*}&P\left[X_{v_1}=x_1,X_{v_2}=x_2,...,X_{v_m}=x_m|G\backslash \{v_0\}=g\backslash \{v_0\}\right]\\
&=(1+o(1))\prod_{i=1}^m P\left[X_{v_i}=x_i|G\backslash \{v_0\}=g\backslash \{v_0\}\right],\end{align*}
then with probability $1-o(1)$, $G$ is such that for each community $i$, we have that
\[P[X_{v_0}=i|G=g]=(1+o(1))\frac{p_i\prod_{j=1}^m (Q q_j)_i}{\sum_{i'=1}^k p_{i'}\prod_{j=1}^m (Q q_j)_{i'}}.\]
\end{lemma}

This follows from Bayes' Rule.

\begin{remark}
The requirement that
\begin{align*}
&P\left[X_{v_1}=x_1,X_{v_2}=x_2,...,X_{v_m}=x_m|G\backslash \{v_0\}=g\backslash \{v_0\}\right]\\
&=(1+o(1))\prod_{i=1}^m P\left[X_{v_i}=x_i|G\backslash \{v_0\}=g\backslash \{v_0\}\right]
\end{align*}
is necessary because otherwise dependencies between the communities of the $v_i$ would complicate the calculation. It is reasonable because if $(X,G)\sim\sbm(n,p,Q/n)$, $v_0$ is a random vertex in $G$, and $v_1,... v_m$ are the neighbors of $v_0$, then the following holds. With probability $1-o(1)$, for every $i$ and $j$ every path between $v_i$ and $v_j$ in $G\backslash\{v_0\}$ has a length of $\Omega(\log n)$. So, we would expect that knowing the community of $v_i$ would provide very little evidence about the community of $v_j$.
\end{remark}

Given preliminary guesses of the vertices' communities, we might be able to develop better guesses by applying this repeatedly. However, in order to estimate $P[X_v=i|G=g]$ one needs to know $P[X_{v'}=i'|G\backslash \{v\}=g\backslash\{v\}]$ for all $i'$ and all $v'$ adjacent to $v$. Then in order to compute those, one needs to know $P[X_{v''}=i'|G\backslash\{v,v'\}=g\backslash\{v,v'\}]$ for all $v'$ adjacent to $v$, $v''$ adjacent to $v'$, and community $i'$. The end result is that if we wanted to apply the formula recursively with $t$ layers of recursion, we would need an estimate of $P[X_{v_t}=i|G\backslash \{v_0,...,v_{t-1}\}=g\backslash\{v_0,...,v_{t-1}\}]$, $(q_{v_0,v_1,...,v_t})_i$ for every path $v_0,...,v_t$ in $G$. Then we could refine these estimates by means of the following algorithm.

\vspace{1 cm}
\noindent
{\em PathBeliefPropagationAlgorithm(t,q, p,Q, G):}
\begin{enumerate}
\item Set $q^{(0)}=q$.

\item For each $0<t'< t$, each $v\in G$, and each community $i$, set
\begin{align*}
&(q^{(t')}_{v_0,...,v_{t-t'}})_i\\
&=\frac{p_i\prod_{v_{t-t'+1}:(v_{t-t'},v_{t-t'+1})\in E(G), v_{t-t'+1}\not\in (v_0,...,v_{t-t'})}(Q q^{(t'-1)}_{v_0,...,v_{t-t'},v_{t-t'+1}})_i}{\sum_{i'=1}^k p_{i'}i\prod_{v_{t-t'+1}:(v_{t-t'},v_{t-t'+1})\in E(G), v_{t-t'+1}\not\in (v_0,...,v_{t-t'})}(Q q^{(t'-1)}_{v_0,...,v_{t-t'},v_{t-t'+1}})_{i'}}.\end{align*}

\item For each $v\in G$ and each community $i$, set 
\[(q^{(t)}_{v})_i=\frac{p_i\prod_{v': (v,v')\in E(G)} (Q q^{(t-1)}_{v,v'})_i}{\sum_{i'=1}^k p_{i'}\prod_{v': (v,v')\in E(G)}(Q q^{(t-1)}_{v,v'})_{i'}}.\]

\item Return $q^{(t)}$.
\end{enumerate}
\vspace{1 cm}

Of course, that requires nontrivial initial estimates for the probability that each vertex is in each community. In the stochastic block model, whatever approach we used to break the symmetry between communities will generally give us some information on what communities the vertices are in. For instance, if we have a graph drawn from $\sbm(n,a,b)$, we select the $\sqrt{n}$ vertices of highest degree, and we declare that community $1$ is the community that contains the majority of these vertices, then we can conclude that each of these vertices is in community $1$ with probability $1/2+\Theta(1/\sqrt[4]{n})$. So, we could set $(q_{v_0,v_1,...,v_t})_i$ equal to $1/2+1/\sqrt[4]{4\pi^2n}$ when $i=1$ and $v_t$ is one of the selected vertices, $1/2-1/\sqrt[4]{4\pi^2n}$ when $i=2$ and $v_t$ is one of the selected vertices, and $1/2$ when $v_t$ is any other vertex. Unfortunately, this would not work. These prior probabilities give every vertex at least a $1/2$ probability of being in community $1$ and if $a>b$, the existence of an edge between two vertices is always evidence that they are in the same community. As a result, there is no way that this algorithm could ever conclude that a vertex has a less than $1/2$ chance of being in community $1$, and it will probably end up concluding that all of the vertices in the main component are in community $1$. The fundamental problem is that the abscence of an edge between $v$ and $v'$ provides evidence that these vertices are in different communities, and this algorithm fails to take that into account. Generally, as long as our current estimates of vertice's communities assign the right number of vertices to each community, each vertex's nonneigbors are balanced between the communities, so the nonedges provide negligible amounts of evidence. However, if they are not taken into account, then any bias of the estimates towards one community can grow over time. 

The obvious solution to this would be to add terms to the formula for $q^{(t')}_{v_0,...,v_{t-t'}}$ that account for all $v_{t-t'+1}$ that are not adjacent to $v_{t-t'}$. However, that would mean that instead of merely having a value of $q_{v_0,...,v_t}$ for every path of length $t$ or less we would need to have a value of $q_{v_0,...,v_t}$ for every $(t+1)$-tuple of vertices. That would increase the memory and time needed to run this algorithm from the already questionable $ne^{O(t)}$ to the completely impractical $O(n^{t+1})$. A reasonable way to make the algorithm more efficient is to assume that  $P[X_{v''}=i'|G\backslash\{v,v'\}=g\backslash\{v,v'\}]\approx P[X_{v''}=i'|G\backslash\{v'\}=g\backslash\{v'\}]$, which should generally hold as long as there is no small cycle containing $v$, $v'$, and $v''$. We can also reasonably assume that if $v$ is not adjacent to $v'$ then $P[X_{v'}=i'|G\backslash\{v\}=g\backslash\{v\}]\approx P[X_{v'}=i'|G=g]$. Using that to simplify the proposed approach, we would need an intial estimate of $P[X_{v'}=i|G\backslash\{v\}=g\backslash\{v\}]$, $(q_{v,v'})_i$, for each community $i$ and each adjacent $v$ and $v'$, and an initial estimate of $P[X_{v}=i|G=g]$, $(q_v)_i$ for each community $i$ and vertex $v$. Then we could develop better guesses by means of the following algorithm.

\vspace{1 cm}
\noindent
{\em AdjustedBeliefPropagationAlgorithm(t,q, p,Q, G):}
\begin{enumerate}
\item Set $q^{(0)}=q$.

\item For each $0<t'< t$: \begin{enumerate}

\item For each $(v,v')\in E(G)$, and each community $i$, set
\begin{align*}
&(q^{(t')}_{v,v'})_i\\
&=\frac{p_i\prod\limits_{v'': (v',v'')\in E(G), v''\ne v} (Q q^{(t'-1)}_{v',v''})_i\prod\limits_{v'': (v',v'')\not\in E(G),v''\ne v'} [1-(Q q^{(t'-1)}_{v''})_i/n]}{\sum_{i'=1}^k p_{i'}\prod\limits_{v'': (v',v'')\in E(G), v''\ne v}(Q q^{(t'-1)}_{v',v''})_{i'}\prod\limits_{v'': (v',v'')\not \in E(G),v''\ne v'}[1-(Q q^{(t'-1)}_{v''})_{i'}/n]}.\end{align*}

\item For each $v\in G$, and each community $i$, set
\[(q^{(t')}_{v})_i=\frac{p_i\prod\limits_{v': (v,v')\in E(G)} (Q q^{(t'-1)}_{v,v'})_i\prod\limits_{v': (v,v')\not\in E(G),v'\ne v} [1-(Q q^{(t'-1)}_{v'})_i/n]}{\sum_{i'=1}^k p_{i'}\prod\limits_{v': (v,v')\in E(G)}(Q q^{(t'-1)}_{v,v'})_{i'}\prod\limits_{v': (v,v')\not \in E(G),v'\ne v}[1-(Q q^{(t'-1)}_{v'})_{i'}/n]}.\]

\end{enumerate}
\item For each $v\in G$, and each community $i$, set
\[(q^{(t)}_{v})_i=\frac{p_i\prod\limits_{v': (v,v')\in E(G)} (Q q^{(t-1)}_{v,v'})_i\prod\limits_{v': (v,v')\not\in E(G),v'\ne v} [1-(Q q^{(t-1)}_{v'})_i/n]}{\sum_{i'=1}^k p_{i'}\prod\limits_{v': (v,v')\in E(G)}(Q q^{(t-1)}_{v,v'})_{i'}\prod\limits_{v': (v,v')\not \in E(G),v'\ne v}[1-(Q q^{(t-1)}_{v'})_{i'}/n]}.\]

\item Return $q^{(t)}$.
\end{enumerate}

\vspace{1 cm}

\begin{remark}
If one computes $q^{(t')}$ using the formula exactly as written, this will take $O(tn^2)$ time. However, one can reduce the run time to $O(tn)$ by computing $\prod_{v''\in G} [1-(Qq_{v''}^{(t'-1)})_i/n]$ once for each $i$ and $t'$ and then dividing out the terms for all $v''$ that should be excluded from the product in order to calculate each of the $(q^{(t')}_{v,v'})_i$ and the $(q^{(t')}_{v})_i$.
\end{remark}

Alternately, if one has good enough initial guesses of the communities not to need to worry about the estimates becoming biased towards one community, one can use the following simplified version.

\vspace{1 cm}
\noindent
{\em BeliefPropagationAlgorithm(t,q, p,Q, G):}
\begin{enumerate}
\item Set $q^{(0)}=q$.

\item For each $0<t'< t$, each $v\in G$, and each community $i$, set
\[(q^{(t')}_{v,v'})_i=\frac{p_i\prod_{v'': (v',v'')\in E(G), v''\ne v} (Q q^{(t'-1)}_{v',v''})_i}{\sum_{i'=1}^k p_{i'}\prod_{v'': (v',v'')\in E(G), v''\ne v}(Q q^{(t'-1)}_{v',v''})_{i'}}.\]

\item For each $(v,v')\in E(G)$ and each community $i$, set 
\[(q^{(t)}_{v})_i=\frac{p_i\prod_{v': (v,v')\in E(G)} (Q q^{(t-1)}_{v,v'})_i}{\sum_{i'=1}^k p_{i'}\prod_{v': (v,v')\in E(G)}(Q q^{(t-1)}_{v,v'})_{i'}}.\]

\item Return $q^{(t)}$.
\end{enumerate}
\vspace{1 cm}

This algorithm ends with a probability distribution for the community of each vertex that is hopefully better than the initial guesses, and it can be run efficiently. However, it has several major downsides. First of all, it uses a non-trivial (e.g., non-linear) formula to calculate each sucessive set of probability distributions. In particular, one needs to know the parameters of the model in order to run the algorithm. Worse, it may not extend well to other models. For instance, if we wanted to convert it to an algorithm for the GBM, then instead of just keeping track of a probability distribution for each vertex's community, we would need to keep track of a probability distribution for each vertex's community and geometric location. The most obvious way to do that would require infinite memory, and it could be highly nontrivial to determine how to track a reasonable approximation of this probability distribution. As such, it seems that we need a more general algorithm.

%In order to generalize it, we will want to remove unnecessary details that are limiting its range of applicability. 
As a first step towards this, note that we will probably have fairly low confidence in our original guesses of the vertices' communities; e.g., if these come from a random initial guess. If vertices from some communities have higher expected degrees than others, we can quickly gain significant amounts of evidence about the vertices' communities by counting their degrees, but otherwise our beliefs about all of the vertices' communities will be fairly close to their prior probabilities for a while. As such, it may be useful to focus on the first order approximation of our formula when our beliefs about the communities of $v_1,...,v_m$ are all close to the prior probabilities for a vertex's community, and vertices in each community have the same expected degree, $d$. In this case, every entry of $Qp$ must be equal to $d$. So, we have that
\begin{align*}
P[X_{v_0}=i|G=g]&\approx \frac{p_i\prod_{j=1}^m (Q q_j)_i}{\sum_{i'=1}^k p_{i'}\prod_{j=1}^m (Q q_j)_{i'}}\\ \\
&=\frac{p_i\prod_{j=1}^m [d+(Q (q_j-p))_i]}{\sum_{i'=1}^k p_{i'}\prod_{j=1}^m [d+(Q (q_j-p))_{i'}]}\\ \\
&\approx \frac{p_i[1+\sum_{j=1}^m (Q (q_j-p))_i/d]}{\sum_{i'=1}^k p_{i'}[1+\sum_{j=1}^m [(Q (q_j-p))_{i'}/d]}\\ \\
&= \frac{p_i[1+\sum_{j=1}^m (Q (q_j-p))_i/d]}{1+\sum_{j=1}^m p\cdot Q (q_j-p)/d}\\ \\
&=p_i+p_i\sum_{j=1}^m (Q (q_j-p))_i/d.\\
\end{align*}

We can then rewrite BeliefPropagationAlgorithm using this approximation in order to get the following:

\vspace{1 cm}
\noindent
{\em LinearizedBeliefPropagationAlgorithm(t, p, Q, q, G):}
\begin{enumerate}
\item Set $\epsilon^{(0)}_{v,v'}=q_{v,v'}-p$ for all $(v,v')\in E(G)$.

\item For each $0<t'< t$, and each $v\in G$, set
\[\epsilon^{(t')}_{v,v'}=\sum_{v'': (v',v'')\in E(G), v''\ne v} PQ\epsilon^{(t'-1)}_{v',v''}/d.\]

\item For each $(v,v')\in E(G)$, set 
\[q^{(t)}_{v}=p+\sum_{v': (v,v')\in E(G)} PQ\epsilon^{(t-1)}_{v,v'}/d.\]

\item Return $q^{(t)}$.
\end{enumerate}
\vspace{1 cm}

Once again, one may want to take the evidence provided by nonedges into account as well, in which case one gets the following.

\vspace{1 cm}
\noindent
{\em AdjustedLinearizedBeliefPropagationAlgorithm(t, p, Q, q, G):}
\begin{enumerate}
\item Set $\epsilon^{(0)}_{v,v'}=q_{v,v'}-p$ for all $(v,v')\in E(G)$.

\item Set $\epsilon^{(0)}_v=q_v-p$ for all $v\in G$.

\item For each $0<t'< t$:
\begin{enumerate}

\item For each $(v,v')\in E(G)$, set
\[\epsilon^{(t')}_{v,v'}=\sum_{v'': (v',v'')\in E(G), v''\ne v} PQ\epsilon^{(t'-1)}_{v',v''}/d-\sum_{v'':(v',v'')\not\in E(G),v''\ne v'} PQ\epsilon_{v''}^{(t'-1)}/n.\]

\item For each $v\in G$, set
\[\epsilon^{(t')}_{v}=\sum_{v': (v,v')\in E(G)} PQ\epsilon^{(t'-1)}_{v,v'}/d-\sum_{v': (v,v')\not\in E(G),v'\ne v} PQ\epsilon^{(t'-1)}_{v'}/n.\]

\end{enumerate}
\item For each $v\in G$, set 
\[q^{(t)}_{v}=p+\sum_{v': (v,v')\in E(G)} PQ\epsilon^{(t-1)}_{v,v'}/d-\sum_{v': (v,v')\not\in E(G),v'\ne v} PQ\epsilon^{(t-1)}_{v'}/n.\]

\item Return $q^{(t)}$.
\end{enumerate}
\vspace{1 cm}

So far, this does not seem like a particularly big improvement. However, this paves the way for further simplifications. First, we define the graph's nonbacktracking walk matrix as follows.
\begin{definition}
Given a graph, the graph's nonbacktracking walk matrix, $B$, is a matrix over the vector space with an orthonormal basis consisting of a vector for each directed edge in the graph. $B_{(v_1,v_2),(v'_1,v'_2)}$ is $1$ if $v'_2=v_1$ and $v'_1\ne v_2$, otherwise it is $0$. In other words, $B$ maps a directed edge to the sum of all directed edges starting at its end, except the reverse of the original edge.
\end{definition}

Then we define its adjusted nonbacktracking walk matrix as follows.
\begin{definition}
Given a graph $G$, and $d>0$, the graph's adjusted nonbacktracking walk matrix, $\widehat{B}$ is the $(2|E(G)|+n)\times (2|E(G)|+n)$ matrix such that for all $w$ in the vector space with a dimension for each directed edge and each vertex, we have that $\widehat{B} w=w'$, where $w'$ is defined such that \[w'_{v,v'}=\sum_{v'': (v',v'')\in E(G), v''\ne v} w_{v',v''}/d-\sum_{v'':(v',v'')\not\in E(G),v''\ne v'} w_{v''}/n\]
for all $(v,v')\in E(G)$ and
\[w'_{v}=\sum_{v': (v,v')\in E(G)} w_{v,v'}/d-\sum_{v': (v,v')\not\in E(G),v'\ne v} w_{v'}/n\] 
for all $v\in G$.
\end{definition}

These definitions allows us to state the following fact:

\begin{theorem}
When AdjustedLinearizedBeliefPropagationAlgorithm is run, for every $0<t'<t$, we have that
\[\epsilon^{(t')}=\left (\widehat{B}\otimes PQ\right)^{t'}\epsilon^{(0)}.\]
\end{theorem}

\begin{proof}
It is clear from the definition of $\widehat{B}$ and the propagation step of AdjustedLinearizedBeliefPropagationAlgorithm that $\epsilon^{(t')}=\left (\widehat{B}\otimes PQ\right)\epsilon^{(t'-1)}$ for all $0<t'<t$. The desired conclusion follows immediately.
\end{proof}

\begin{comment}
for each $0<t'<t$, 

We procede by induction on $t'$. Obviously, $\epsilon^{(0)}=(\widehat{B}\otimes PQ)^{0}\epsilon^{(0)}$. Now, assume that this holds for $t'-1$. For every $(v,v')\in E(G)$, we have that
\begin{align*}
\epsilon^{(t')}_{v,v'}&=\sum_{v'': (v',v'')\in E(G), v''\ne v} PQ\epsilon^{(t'-1)}_{v',v''}/d-\sum_{v'': (v',v'')\not\in E(G),v''\ne v'} PQ\epsilon^{(t'-1)}_{v',v''}/n\\
&=\sum_{v'': (v',v'')\in E(G), v''\ne v} (1/d+1/n)PQ\epsilon^{(t'-1)}_{v',v''}-\sum_{v''\ne v,v''\ne v'} PQ\epsilon^{(t'-1)}_{v',v''}/n\\
&=\sum_{v'':(v',v'')\in E(G)} B_{(v',v''),(v,v')}  (PQ)\epsilon^{(t'-1)}_{v',v''}/d\\
&=[(e_{(v,v')}B)\otimes (PQ)]\sum_{(v_1,v_2)\in E(G)} e_{(v_1,v2)}\otimes \epsilon^{(t'-1)}_{v_1,v_2}/d\\
\end{align*}
So,
\begin{align*}
\epsilon^{(t')}&=\left(\frac{1}{d} B\otimes PQ\right) \epsilon^{(t'-1)}\\
&=\left(\frac{1}{d} B\otimes PQ\right)^{t'}\epsilon^{(0)}
\end{align*}
\end{comment}

In other words, AdjustedLinearizedBeliefPropagationAlgorithm is essentially a power iteration algorithm that finds an eigenvector of $\widehat{B}\otimes PQ$ with a large eigenvalue. 
\begin{remark}
If $\epsilon^{(0)}$ was selected completely randomly, then it would probably find the eigenvector with the highest eigenvalue. However, we actually set $\epsilon^{(0)}$ by starting with a set of probability distributions and then subtracting $p$ from each one of them. As a result, each component vector in $\epsilon^{(0)}$ has entries that add up to $0$, and due to the fact that every column in $PQ$ has the same sum, that means that each component vector in $\epsilon^{(t')}$ will add up to $0$ for all $t'$. As such, it cannot find any eigenvector that does not also have this property.
\end{remark}

$\widehat{B}\otimes PQ$ has an eigenbasis consisting of tensor products of eigenvectors of $\widehat{B}$ and eigenvectors of $PQ$, with eigenvalues equal to the products of the corresponding eigenvalues of $\widehat{B}$ and $PQ$. As such, this suggests that for large $t'$, $\epsilon^{(t')}$ would be approximately equal to a tensor product of eigenvectors of $\widehat{B}$ and $PQ$ with large corresponding eigenvalues. For the sake of concreteness, assume that $w$ and $\rho$ are eigenvectors of $\widehat{B}$ and $PQ$ such that $\epsilon^{(t')}\approx w\otimes \rho$. That corresponds to estimating that
\[P[X_v=i|G=g]\approx p_i+w_v \rho_i\]
for each vertex $v$ and community $i$. If these estimates are any good, they must estimate that there are approximately $p_in$ vertices in community $i$ for each $i$. In other words, it must be the case that the sum over all vertices of the estimated probabilities that they are in community $i$ is approximately $p_i n$. That means that either $\rho$ is small, in which case these estimates are trivial, or $\sum_{v\in G} w_v\approx 0$. Now, let the eigenvalue corresponding to $w$ be $\lambda$. If $\sum_{v\in G} w_v\approx 0$, then for each $(v,v')\in E(G)$, we have that
\begin{align*}
\lambda w_{v,v'} &\approx \sum_{v'':(v',v'')\in E(G),v''\ne v} w_{v',v''}/d\\
&= \sum_{v'':(v',v'')\in E(G)} B_{(v',v''),(v,v')} w_{v',v''}/d.
\end{align*}
So, the restriction of $w$ to the space spanned by vectors corresponding to directed edges will be approximately an eigenvector of $B$ with an eigenvalue of $\approx \lambda/d$. Conversely, any eigenvector of $B$ that is balanced in the sense that its entries add up to approximately $0$ should correspond to an eigenvector of $\widehat{B}$. So, we could try to determine what communities vertices were in by finding some of the balanced eigenvectors of $B$ with the largest eignvalues, adding together the entries corresponding to edges ending at each vertex, and then running $k$-means or some comparable algorithm on the resulting data. The eigenvector of $B$ with the largest eigenvalue will have solely nonnegative entries, so it will not be balanced. However, it is reasonable to expect that its next few eigenvectors would be relatively well balanced. This approach has a couple of key advantages over the previous algorithms. First of all, one does not need to know anything about the graph's parameters to find the top few eigenvectors of $B$, so this algorithm works on a graph drawn from an SBM with unknown parameters. Secondly, the calculations necessary to reasonably approximate the top few eigenvectors of $B$ will tend to be simpler and quicker than the calculations required for the previous algorithms. Note that balanced eigenvectors of $B$ will be approximately eigenvectors of $B-\frac{\lambda_1}{n}\mathbb{1}$, where $\lambda_1$ is the largest eigenvalue of $B$ and $\mathbb{1}$ is the matrix thats entries are all $1$. As such, we could look for the main eigenvectors of $B-\frac{\lambda_1}{n}\mathbb{1}$ instead of looking for the main balanced eigenvectors of $B$.

\subsection{Curbing the enthusiasm around tangles without promoting tails}\label{curb}

This is likely to work in the SBM, however many other models have larger numbers of cycles, which creates problems for this approach. In order for the formula for the probability distribution of $X_v$ to be accurate, we need probability distributions for the communities of the neighbors of $v$ based on $G\backslash\{v\}$ that are roughly independent of each other. In the SBM there are very few cycles, so requiring nonbacktracking is generally sufficient to ensure that the estimates of the communities of the neighbors of $v$ are essentially unaffected by the presence of $v$ in $G$, and that these probability distributions are nearly independent of each other. In a model like the GBM, on the other hand, the neighbors of $v$ will typically be very close to each other in $G\backslash\{v\}$. 

This will cause two major problems. First of all, it means that the probability distributions of the communities of the neighbors of $v$ will tend to be highly dependent on each other, in violation of the assumptions for the formula to be valid. More precisely, they will have a strong positive dependency which will cause the formula to double-count evidence and become overconfident in its conclusions about the community of $v$. The worse problem is that the formula underlying this algorithm assumes that for most communities $i$ and vertices $v$, $v'$, and $v''$ with $v$ and $v''$ adjacent to $v'$, we have that $P[X_{v}=i|G\backslash\{v',v''\}=g\backslash\{v',v''\}]\approx P[X_{v}=i|G\backslash\{v'\}=g\backslash\{v'\}]$. In the SBM, small cycles are rare, so this will usually hold. However, in the GBM, there are enough small cycles that if $v$ and $v''$ are both adjacent to $v'$, then there is likely to be a short path connecting them in $G\backslash\{v'\}$. As a result, $v''$ is probably providing significant evidence about the community of $v$, which means that this approximation is unlikely to hold. More precisely, our previous beliefs about the community of $v$ will influence our beliefs about the community of $v''$ which will influence our beliefs about the community of $v'$ in turn. In a particularly dense tangle, this feedback could result in an out of control positive feedback loop which would result in us rapidly becoming highly confident that the vertices in that region were in some random community. Then we would classify all other vertices in the graph based on their proximity to that tangle. So, we believe that this approach would fail on the GBM in the following sense.

\begin{conjecture} \label{conj3}
Let $G\sim \gbm(n,s,t)$, with $s$ and $t$ chosen such that $G$ has a single giant component with high probability. Also, let $B$ be the nonbacktracking walk matrix of the main component of $G$. Then computing the eigenvector of $B$ with the second largest eigenvalue and dividing the vertices of $G$ into those with above-median sums of their entries in the eigenvector and those with below-median sums of entries does not recover communities with optimal accuracy.
\end{conjecture}

\begin{proof}[Justification]
We believe that the eigenvectors of $G$ with the largest eigenvalues will each have a single small region of the graph such that the entries of the graph corresponding to edges in that region will have large absolute values. Then, the entries corresponding to directed edges outside the region will decrease exponentially with distance from that region. As a result, dividing the vertices into those that have above-median and below-median sums of entries will divide the vertices into those that are geometrically near some such region and those that are not. This is a suboptimal way to assign vertices to communities, so the algorithm will have suboptimal accuracy.

More formally, the average vertex in $G$ will have a degree of $O(1)$. However, for any geometric point within a distance of $s$ of the origin, the probability distribution of the number of vertices in the circle of radius $t/2\sqrt{n}$ centered on that point is approximately a poisson distribution with a mean of $\Omega(1)$. It is possible to find $\Theta(n)$ such circles that do not overlap, and the circle's contents are thus nearly independent of each other. Furthermore, we can do so in such a way that half of these points have $x$-coordinates less than $-s/4$ and the other half have $x$-coordinates that are greater than $s/4$.

As a result, with probability $1-o(1)$, there will be at least one such circle that contains $\Omega(\log(n)/\log(\log(n)))$ vertices with $x$-coordinate less than $-s/4$, and at least one such circle that contains $\Omega(\log(n)/\log(\log(n)))$ vertices with $x$-coordinate greater than $s/4$. That implies that there exists $\eta=\Omega(\log(n)/\log(\log(n)))$ such that with probability $1-o(1)$, $G$ contains $2$ complete graphs with at least $\eta$ vertices such that there is no path of length less than $s\sqrt{n}/2t$ between them. Now, pick one of these complete graphs and let $w$ be the vector in $\mathbb{R}^{2|E(G)|}$ that has all entries corresponding to directed edges in the graph set to one and all other entries set to $0$. Next, let $w'$ be the vector that has all entries corresponding to directed edges in the other complete graph set to one and all other entries set to $0$. For all nonnegative integers $t$, we have that $w\cdot B^t w/w\cdot w\ge (\eta-2)^t$ and $w'\cdot B^t w'/w'\cdot w'\ge (\eta-2)^t$. This implies that the largest eigenvalue of $B$ must be at least $\eta-2$. Next, observe that there must exist $c,c'=O(1)$ that are not both $0$ such that $cw+c'w'$ can be expressed as a linear combination of eigenvectors of $B$ other than its eigenvector of largest eigenvalue. Also, for all $0\le t< \sqrt{n}/r$ we have that $w\cdot B^t w'=w'\cdot B^t w=0$. So, 
\[(cw+c'w')\cdot B^t (cw+c'w')/(cw+c'w')\cdot (cw+c'w')\ge (\eta-2)^t\]
for all $0\le t<s\sqrt{n}/2t$. That means that either the second largest eigenvector of $B$ is also at least $(\eta-2)/2$ or the representation of $cw+c'w'$ as a sum of eigenvectors of $B$ involves multiple eigenvectors of magnitude $\Omega(2^{s\sqrt{n}/2t}/n^2)$ that cancel out almost completely. The former is far more plausible than the latter.

Now, let $w$ be an eigenvector of $B$ with an eigenvalue of $\lambda\ge (\eta-2)/2$. Also, let $(v_0,v_1)$ be a random edge in $G$. Next, define $v_2$ to be the vertex adjacent to $v_1$ that maximizes the value of $|w_{v_1,v_2}|$, $v_3$ to be the vertex adjacent to $v_2$ that maximizes the value of $|w_{v_2,v_3}|$ and so on. For each $i$, we have that
\begin{align*}
|w_{v_{i-1},v_i}|&=\frac{1}{\lambda}|\sum_{v':(v_{i},v')\in E(G),v'\ne v_{i-1}} w_{v_i,v'}|\\
&\le \frac{1}{\lambda}\sum_{v':(v_{i},v')\in E(G),v'\ne v_{i-1}} |w_{v_i,v'}|\\
&\le \frac{\mathrm{deg}(v_i)-1}{\lambda} |w_{v_i,v_{i+1}}|.\\
\end{align*}
So, if we define $m$ to be the first integer such that $v_m$ has a degree of at least $(\eta-2)/4$, then it must be the case that $|w_{v_0,v_1}|\le 2^{1-m}|w_{v_{m-1},v_m}|$. The graph probably only has $O(\log(n))$ cliques of maximum size. Also, if we consider the largest eigenvectors of the restrictions of $B$ to the vicinity of two different cliques of maximum size, then we would expect them to differ by a factor of $1+\Omega(1/\log^2(n))$ simply due to random variation in what other vertices are nearby. As a result, these eigenvalues are sufficiently unbalanced that each of the main eigenvectors of $B$ will probably be centered on a single dense region of the graph. So, assigning vertices to communities based on whether their edges' entries in that eigenvector have an above or below median sum would essentially assign vertices to communities based on their proximity to that region. The correct way to assign vertices to communities based on their geometric location is to assign all vertices with positive $x$-coordinates to one community and those with negative $x$-coordinates to the other. However, this algorithm would assign all vertices that are sufficiently distant from the central region to the other community, even if they are further from the $y$-axis than the central region is. As a result, it will classify vertices with suboptimal accuracy.
%\Enote{Insert diagram of the community assignment if necessary}
\end{proof}

\begin{remark}
There are several variants of this approach that we could have used, such as computing the $m$ eigenvectors of $B$ with the largest eigenvalues for some $m>2$ and then running some sort of clustering algorithm on the coordinates given by the entries in these vectors corresponding to these vertices. However, we do not think that these variants are likely to do better. All of the eigenvectors would assign negligible values to all but a tiny fraction of the vertices, and many clustering algorithms would classify all of the vertices that are essentially at the origin as being in the same community. Any algorithm that does so would classify $n-o(n)$ vertices as being in the same community, and would thus fail to solve weak recovery. Many algorithms that avoid that would classify all vertices sufficiently close to a specific dense region of the graph as being in one community. One could conceivably hope to compare two eigenvectors centered on regions on opposite sides of the graph in order to recover communities with optimal accuracy. However, any algorithm that does so would probably be overfitted to the GBM.
\end{remark}

More generally, if we start with preliminary guesses of the vertices' communities and then use the formula above to refine them in such a model, the cycles will result in large amounts of unwanted feedback, we will quickly gain high confidence that the vertices in the largest positive feedback loop are in some random community, and then we will end up classifying the rest of the vertices based on their proximity to said feedback loop. There are several possible ways to try to fix this, each having their own downsides.

\begin{enumerate}
\item We could try to use a formula that takes dependencies between a vertex's neighbors into account. However, this would be substantially more complicated, and one would need information on the model's parameters to determine what the correct formula is.

\item We could use the graph's $r$-nonbacktracking walk matrix, which is a generalization of its nonbacktracking walk matrix that is defined over a vector space with a dimension for every directed path of length $r-1$. It maps each directed path of length $r-1$ to the sum of all paths that can be formed by adding a new vertex to the front of the path and deleting the last vertex. Using it instead of $B$ would result in an algorithm that works the same as the previous one, except that cycles of length $r$ or less would not cause feedback. However, this would not do anything to fix doublecounting of evidence caused by cycles. Also, to adequetely mitigate feedback in a model like the GBM, we would need to use a huge $r$, and the dimensionality of the $r$-nonbacktracking walk matrix would grow exponentially in $r$. As a result, doing this effectively would take an impractical amount of time, which would render the algorithm useless in practice.

\item We could simply eliminate all areas of the graph that have too many cycles by deleting edges or vertices as needed. The problem with this is that those edges and vertices provided information on which vertices are in which communities, and we can not necessarily afford to lose that information.

%\Cnote{The effectiveness of replacing cycles with stars depends on the situation. In the SBM there are not enough cycles to have much effect, so doing something with them will not change much. In the GBM there are cycles everywhere, which makes replacing enough of them to make a meaningful difference difficult. In the HBM, it might be  a good approach to find every set of vertices that are linked by small cycles and replace each such set with a star.}

\item We could raise the nonbacktracking walk matrix to some power, $r$, and then reduce all nonzero entries to $1$. Call the resulting matrix $B^{(r)}$. Using $B^{r}$ instead of $B$ would not change anything. Every entry of $B^{r}$ corresponding to a pair of vertices that are not near any small cycles would be $0$ or $1$ anyways. However, in a part of the graph where multiple small cycles are causing large amounts of feedback and doublecounting, entries of $B^{r}$ could be very large. As such, reducing every positive entry of $B^{(r)}$ to $1$ significantly mitigates the effects of feedback and doublecounting caused by small cycles. Note that while replacing $B$ with $B^{(r)}$ can be viewed as adding even more edges to regions of the graph with dense tangles, it adds edges everywhere else too, and the tangles will still end up with an edge density that is comparable to the rest of the graph, limiting their influence on the eigenvectors. However, this algorithm does not handle cycles perfectly. In a Bayes optimal algorithm, the existence of multiple paths between two vertices $v$ and $v'$ would provide more evidence that they are in the same community than either would alone, but not as much as the sum of the evidence they would have provided individually. Using $B^{(r)}$, the existence of multiple paths could provide any amount of evidence from only slightly more than one of them alone would have provided to far more than the sum of what they would have provided individually.

However, it is still an improvement over $B$ because in the densest tangles there will almost certainly be multiple length $r$ nonbacktracking walks between the same pairs of vertices, so it will mitigate some of the feedback. The feedback resulting from particularly dense tangles is the dominant factor influencing the eigenvectors of $B$ in models like the GBM, so countering it is critically important to developing an algorithm that recovers communities with optimal accuracy in such models. If we set $r$ to a constant fraction of the graph's diameter, then we believe this would do a good enough job of mitigating the effects of cycles to succeed on the GBM, while keeping the desirable properties of $B$ on the SBM. We discuss the approach in details in the next section; we conclude here with a last option.  

\item A final option would be to set our new estimates of what communities a vertex is likely to be in equal to an average of the previous estimates for its neighbors instead of their sum. This would ensure that no matter how many cycles there are in an area of the graph, one can never have enough feedback to make beliefs about the communities of vertices in a region amplify themselves over time. Essentially, we would be substituting the random walk matrix, $AD^{-1}$, or possibly the random nonbacktracking walk matrix, for the nonbacktracking walk matrix. Note that the random walk matrix has its name due to the fact that for all vertices $v$ and $v'$ and $t>0$ the probability that a length $r$ random walk starting at $v$ ends at $v'$ is given by $(AD^{-1})^t_{v',v}$.

However, this idea has some major problems. First of all, when guessing the community of a vertex whose neighborhood in the graph is a tree, it gives less weight to vertices connected to that vertex by paths through high degree vertices than to vertices connected to it by paths through low degree vertices. This is simply an incorrect formula, with the result that this will proably fail to solve weak recovery in some instances of the SBM on which weak recovery is efficiently solvable. Furthermore,  while it would seem like the averaging makes feedback a nonissue, this is not really true. While the averaging does ensure that the values in a local region of the graph cannot amplify over time, it also ensures that the values everywhere will converge to the overall average. So, if one region of the graph is isolated enough in the sense that the number of edges between that region and the rest of the graph is small relative to the number of edges in the region, it is possible that the vertices there will homogenize with the vertices outside the region more slowly than the graph as a whole will homogenize. If that happens, then eventually most of the variation in values will be determined by how close vertices are to that region. Recall Figure \ref{lap-cut} for such an outcome. In fact, we conjecture the following.

\begin{conjecture}
For all $a,b>0$, the algorithm that assigns vertices to communities by finding the $m$ dominant eigenvectors of the random walk matrix of the main component of the graph and then running $2$-means fails to solve weak recovery on $\sbm(n,a,b)$.
\end{conjecture}

\begin{proof}[Justification]
No eigenvalue of $AD^{-1}$ can be greater than $1$ because for any vector $w$, the absolute value of the largest entry in $D^{-1}A w$ can not be more than the absolute value of the largest entry in $w$. Also, the vector whose entries are all $1$ is an eigenvector of $D^{-1}A$ with an eigenvalue of $1$. Now, consider the matrix $D^{-1/2}AD^{-1/2}$. This is a conjugate of $AD^{-1}$ so it has the same eigenvalues, and its eigenvectors are what one gets when multiplying the eigenvectors of $AD^{-1}$ by $D^{-1/2}$. 

Next, let $w$ be the vector such that $w_v$ is $1$ if $v$ is in community $1$ and $-1$ if $v$ is in community $2$. Now, let $r=\lfloor \sqrt{\ln(n)}\rfloor$, $v_0$ be a vertex randomly selected with probability proportional to the square root of its degree, and $v_0,v_1,...,v_r$ be a random walk starting at $v_0$. With probability $1-o(1)$, the subgraph of $G$ induced by vertices within $r$ edges of $v_0$ is a tree. Call a subgraph of $G$ that can be disconected from the main component of the graph by deleting a single edge a branch. With probability $1-o(1)$, the combined sizes of all branches connected to $v_i$ by an edge will be $O(\log(\log(n)))$ for all $i$. So, the expected number of steps required for the random walk to move from $v_i$ to a vertex that is not part of any branch that $v_i$ is not also part of is also $O(\log(\log(n)))$. That means that the number of steps that do not involve detouring onto a branch will typically be $\Omega(\sqrt{\log(n)}/\log(\log(n)))$. Any step in the walk that does not involve a detour onto a branch is at least as likely to move farther away from $v_0$ as it is to move closer to it. As a result, for each $d$, we have that $P[d(v_0,v_r)=d]=O(\log^2(\log(n))/\sqrt[4]{\log(n)})$, where $d(v,v')$ is the minimum length of a path between $v$ and $v'$. For a fixed value of $d(v_0,v_r)$, we have that $v_0$ and $v_r$ are in the same community with probaiblity $1/2+(\frac{a-b}{a+b})^{d(v_0,v_r)}$. So, $v_0$ and $v_r$ are in the same community with probability $1/2+O(\log^2(\log(n))/\sqrt[4]{\log(n)})$. That means that $w\cdot (D^{-1/2}AD^{-1/2})^r w/n=O(\log^2(\log(n))/\sqrt[4]{\log(n)})$. That in turn means that any eigenvector of $D^{-1/2}AD^{-1/2}$ that has $\Omega(1)$ correlation with $w$ must have an eigenvalue of $1-\omega(1/\sqrt{\log(n)})$.

On the flip side, we can find $m$ induced paths in $G$ that that have lengths of $\Omega(\log(n))$ such that no vertex in the path except the last one has any edges except those in the path, these paths do not overlap, and no two of these paths have vertices that share a neighbor. Now, let $w_i$ be the vector that has all of its entries corresponding to a vertex on the $i$th one of these paths set to $1$ and all other entries set to $0$. For each $i$ we have that $w_i\cdot D^{-1/2}AD^{-1/2} w_i/w_i\cdot w_i=1-O(1/\log(n))$, while for all $i\ne j$, we have that $w_i\cdot w_j=0$ and $w_i\cdot D^{-1/2}AD^{-1/2} w_j=0$. That means that the top $m$ eigenvalues of $D^{-1/2}AD^{-1/2}$are all at least $1-O(1/\log(n))$. By the previous argument, that means that none of the top $m$ eigenvectors can have $\Omega(1)$ correlation with $w$. So, attempting to classify vertices by running $2$-means on coordinates defined by these vectors will not solve weak recovery.
\end{proof}
We also refer to the caption of Figure  \ref{sbm_num} for numerical results corroborating this claim. 
\end{enumerate}

\subsection{Graph powering performances}\label{justifications}
We now investigate the performance of the powered nonbacktracking graph, i.e., the matrix $B^{(r)}$, and its undirected counter-part $A^{(r)}$. We start by its performance on the GBM. 
\begin{conjecture}
For all $m>0$, $s>0$, $t>0$, and $0<r'<1/2$, the algorithm that finds the eigenvector of $B^{(r'd)}$ with the second largest eigenvalue and then divides the vertices into those with above- and below-median sums of the entries in the eigenvector corresponding to edges ending at them solves weak recovery on $\gbm(n,s,t)$ with optimal accuracy.
\end{conjecture}

Instead of attempting to justify this conjecture directly, we will start by giving justifications for the analogous conjecture in terms of $A^{(r)}$. Note that $A^{(r)}$ for $r$ large enough is generally similar to $B^{(r)}$, although it does not prevent backtracking the way $B^{(r)}$ does. 
For instance, in the absence of cycles, $(B^{(r)})^s_{(v_1,v_2),(v'_1,v'_2)}$ is a count of the number of length $rs+1$ paths from $(v'_1,v'_2)$ to $(v_1,v_2)$. In this case, $(A^{(r)})^s_{v,v'}$ is a count of the number of sequences of $s$ length $r$ paths such that the first path starts at $v'$, each subsequent path starts at the previous path's endpoint, and the last path ends at $v$.
However, the GBM has enough small cycles that allowing backtracking will not change much, so we believe that the above conjecture is true if the following version is.

\begin{conjecture}
For all $m>0$, $s>0$, $t>0$, and $0<r'<1/2$, the algorithm that finds the eigenvector of $A^{(r'd)}$ with the second largest eigenvalue and then divides the vertices into those with above and below median entries in the eigenvector corresponding to edges ending at them recovers communities on $\gbm(n,s,t)$ with optimal accuracy.
\end{conjecture}

\begin{proof}[Justification]
If graphs drawn from $\gbm(n,s,t)$ have no giant component with high probability, then recovering communities on $\gbm(n,s,t)$ with accuracy greater than $1/2$ is impossible, so the algorithm trivially recovers communities with optimal accuracy. If graphs drawn from $\gbm(n,s,t)$ typically have two giant components then the two eigenvectors of $A^{(r'd)}$ with the largest eigenvalues will be the ones that have positive entries for every vertex in one giant component and all other entries set to $0$. So, the algorithm will assign all vertices in one giant component to one community, which attains optimal accuracy. That leaves the case where graphs drawn from $\gbm(n,s,t)$ have a single giant component with high probability. For the rest of this argument, assume that we are in this case. Any small component of the graph will have $o(n^2)$ vertices, so any eigenvector with nonzero entries corresponding to vertices in such a component must have an eigenvalue of $o(n^2)$. There are multiple cliques of size $\Omega(n^2)$ in the giant component that do not have any edges between them, so the second largest eigenvalue of $A^{(r'd)}$ is $\Omega(n^2)$. Therefore, the eigenvector of $A^{(r'd)}$ with the second largest eigenvalue will have all of its entries corresponding to vertices outside the giant component set to $0$. From now on, ignore all vertices outside of the giant component of the graph.

%\Cnote{This part is rather rushed. Do you want more details on why we only need to worry about the case where there is a single giant component and eigenvectors of the main component?}

Let $w$ be the unit eigenvector of $A^{(r'd)}$ with the largest eigenvalue. Every entry in $A^{(r'd)}$ is nonnegative, so $w$ has all nonnegative entries. Furthermore, there is a path between any two vertices of a graph drawn from $\gbm(n,s,t)$, so every entry in $[A^{(r'd)}]^n$ is positive, which means every entry in $w$ is positive. Now, let $\lambda$ be the eigenvalue corresponding to $w$ and $w'$ be a unit vector of $A^{(r'd)}$. Also, let $E^{(r'd)}$ be the set of all pairs of vertices that have a path of length $r'd$ or less between them. Then it must be the case that
\begin{align*}
\lambda&=\sum_v \lambda (w'_v)^2\\
&=\sum_v \lambda (w'_v)^2\sum_{v':(v,v')\in E^{(r'd)}} w_{v'}/(\lambda w_v)\\
&=\sum_{(v,v')\in E^{(r'd)}} (w'_v)^2 \frac{w_{v'}}{w_v}.
\end{align*}
This means that
\begin{align*}
&w'\cdot A^{(r'd)} w'=\sum_{(v,v')\in E^{(r'd)}} w'_v w'_{v'}\\
&=\lambda+\sum_{(v,v')\in E^{(r'd)}} w'_v w'_{v'}-\frac{w_{v'}}{2w_v}(w'_v)^2-\frac{w_{v}}{2w_{v'}}(w'_{v'})^2\\
&=\lambda-\sum_{(v,v')\in E^{(r'd)}} \frac{w_v w_{v'}}{2}\left(\frac{w'_v}{w_v}-\frac{w'_{v'}}{w_{v'}}\right)^2.
\end{align*}

If $w'$ is the eigenvector of second greatest eigenvalue, then it must be orthogonal to $w$, and for any other unit vector $w''$ that is orthogonal to $w$,  it must be the case that
\[\sum_{(v,v')\in E^{(r'd)}} \frac{w_v w_{v'}}{2}\left(\frac{w'_v}{w_v}-\frac{w'_{v'}}{w_{v'}}\right)^2\le \sum_{(v,v')\in E^{(r'd)}} \frac{w_v w_{v'}}{2}\left(\frac{w''_v}{w_v}-\frac{w''_{v'}}{w_{v'}}\right)^2.\]
On another note, the orthogonality of $w$ and $w'$ implies that $w'$ must have both positive and negative entries. Also, since $w\cdot w=w'\cdot w'$, there must exist $v$ such that $|w'_v/w_v|\ge 1$. So, there must be a significant amount of variation in the value of $\frac{w'_v}{w_v}$, but the values of $\frac{w'_v}{w_v}$ must tend to be similar for nearby vertices. Also, vertices that are particularly close to each other geometrically will have mostly the same neighbors in terms of $E^{(r'd)}$, which will result in them having similar values. So, we would expect that $w'_v/w_v$ would be strongly positive on one side of the graph, strongly negative on the other, and that it will shift gradually between these extremes as one moves from one side to the other. Geometrically, the $x$ direction is the direction the giant component extends the farthest in, so we would expect that these sides would be defined in terms of $x$-coordinates. 

Intuitively, it seems like the entries of $w'$ would switch signs half way between these sides, which means at the $y$-axis. However, we need to consider the possibility that random variation in vertices would disrupt the symmetry in a way that prevents this from happening. Given a vertex $v$ at a given geometric location that gives it a nonvanishing probability of being in the giant component, and $r''<r'd$, the expected number of vertices exactly $r''$ edges away from $v$ in $G$ is $\Theta(r'')$, and the expected number of vertices within $r'd$ edges of $v$ is $\Theta(n^2)$. Now, consider the affects of deleting a random vertex $v'$ on the set of vertices within $r'd$ edges of $v$. It is possible that deleting $v'$ disconnects $v$ from the giant component entirely, in which case its deletion removes nearly the entire set. 

Now, assume that this does not happen. Some of the edges of $v'$ might be cut edges, but the components their removal cuts off from the giant component will typically have $O(1)$ vertices. We would expect that any two edges of $v'$ other than cut edges would be contained in some cycle. Furthermore, due to the abundance of small cycles in the GBM, we would generally expect that there exists some $m$ such that every such pair of edges is contained in a cycle of length at most $m$ and $E[m^2]=O(1)$. That means that for any $v''$ that is still in the giant component, the length of the shortest path from $v$ to $v''$ will be at most $m-4$ edges longer than it was before $v'$ was deleted. So, only vertices that were more than $r'd-(m-4)$ edges away from $v$ are in danger of being removed from the set. Furthermore, a minimum length path from $v$ to $v''$ will only pass through one vertex that is $r''$ edges away from $v$ for each $r''$. So, if $v'$ was $r''$ edges away from $v$, there is only a $O(1/r'')$ chance that the length of the shortest path from $v$ to $v''$ is even effected by the deletion of $v'$. So, the variance of the size of this set conditioned on the geometric location of $v$ and the assumption that $v$ is in the giant component is $O(n^3)$. 

That means that the standard deviation of the size of this set, and the size of the subset restricted to a given geometric region, is much smaller than the expected size of the set. Furthermore, all entries of $(A^{(r'd)})^{\lceil 1/r'\rceil}$ are positive, and we would expect 
that a power iteration method on $(A^{(r'd)})$ would only need a constant number of steps to obtain a reasonable approximation of $w'$. So, random variation in vertex placement has little effect on the behavior of $(A^{(r'd)})$ and we expect that the signs of the entries of $w'$ have $1-o(1)$ correlation with the signs of the corresponding vertices' $x$-coordinates. So, this algorithm would essentially assign the vertices with positive $x$-cordinates to one community and the vertices with negative $x$-cordinates to the other. That is the best one can do to classify vertices in the GBM, so we believe this algorithm will classify vertices with optimal accuracy.
\end{proof}

We likewise believe that this would solve weak recovery on the SBM whenever weak recovery is efficiently solvable. More precisely, we believe the following.

\begin{conjecture}
Choose $p$ and $Q$ such that in $\sbm(n,p,Q/n)$, vertices from every community have the same expected degree and there is an efficient algorithm that solves weak recovery on $\sbm(n,p,Q/n)$. There exists $0<r_0<1$ such that for all constant $r<r_0$ the algorithm that finds the eigenvector of $A^{(rd)}$ or $B^{(rd)}$ with the second largest eigenvalue and then divides the vertices into those with above- and below-median sums of the entries in the eigenvector corresponding to edges ending at them solves weak recovery on $\sbm(n,p,Q/n)$.
\end{conjecture}

\begin{proof}[Justification]
We argue for $B^{(rd)}$. 
We already believed that doing the same thing with the eigenvector of second largest eigenvalue of $B$ would solve weak recovery in this case. $B^{rd}$ has exactly the same eigenvectors as $B$, and for sufficiently small $r$, the difference between $B^{rd}$ and $B^{(rd)}$ should be negligable due to the rarity of cycles in the SBM. As such, we believe that taking the eigenvector of $B^{(rd)}$ with the second largest eigenvalue and then dividing its vertices into those that have above-median sums of the entries corresponding to their edges and those that have below-median sums of the entries corresponding to their edges solves weak recovery. 
\end{proof}

%\Cnote{I stuck with the median for this one as well, but honestly I think you could do basically anything with the second eigenvector and it would solve weak recovery.}
%
%\Cnote{I think that it would solve weak recovery on the HBM with sufficiently favorable parameters. Is it worth writing out another conjecture for that?}

So, we believe that assigning vertices to communities based on the top eigenvectors of $B^{(r)}$ for $r$ proportional to the graph's diameter would be an effective way to recover communities. The problem with this approach is that computing $B^{(r)}$ and finding its eigenvectors for such large values of $r$ is not very efficient.

\begin{remark}
Exactly how long this would take depends on our model and whether we really set $r=\Theta(d)$ or use a somewhat smaller value of $r$. On the SBM we do not need particularly high values of $r$ because there are not many tangles anyway. This would run in $O(n^{1+o(1)})$ time on the SBM for any $r=o(d)$, and $n^{1+\Theta(1)}$ time for $r=\Theta(d)$. On the GBM, the minimum workable value of $r$ would probably be around $d/\log(n)$. Exactly how long the algorithm would take depends heavily on how efficient we are at computing powers and eigenvectors, but it could reasonably take something along the lines of $\Theta(n^{3/2}r)$ time.
\end{remark}

Unfortunately, for practical values of $r$ the GBM will have large numbers of cycles of length greater than $2r$, and $B^{(r)}$ does not do anything to mitigate the feedback caused by these cycles. As a result, we believe that this algorithm would fail on the GBM in the following sense.

\begin{conjecture}
There exists $s_0>0$ such that the following holds. For all $m>0$, $s\ge s_0$, $t>0$, and $r=n^{o(1)}$ such that graphs drawn from $\gbm(n,s,t)$ have a single giant component with high probability, the algorithm that finds the eigenvector of $A^{(r)}$ or $B^{(r)}$ with the second largest eigenvalue and then divides the vertices into those that have above-median sums of the entries corresponding to edges ending at them and those that have below-median sums of the entries ending at them does not recover communities with optimal accuracy on $\gbm(n,s,t)$.
\end{conjecture}

\begin{proof}[Justification]
We argue for $B^{(r)}$ and assume that $r=\omega(1)$.
We believe that this algorithm would encounter nearly the same problem on the GBM as the one using $B$ would. More precisely, we believe that each of the main eigenvectors of $B^{(r)}$ would have some large entries corresponding to some of the directed edges in the vicinity of $(-s/2,0)$ or $(s/2,0)$, but not both. Then the absolute values of entries corresponding to other edges would decrease rapidly as one moved away from that region. As a result, we expect the algorithm to divide the vertices into those sufficiently close to the dominant region, and all other vertices.

More formally, set $s_0=8$ and let $s\ge s_0$. Next, for all $x$ and $y$, let $p(x,y)$ be the probability density function of vertices' geometric location. For any $x,y$ with $|x|<s/4$, we have that $p(s/2,0)/p(x,y)\ge e^2/2\approx 3.6$. Among other things, this means that $p(x,y)$ has maxima at $(\pm x_0,0)$ for some $x_0$ with $s/4<x_0<s/2$. If $(v,v')$ is an edge at $(x,y)$, then on average there will be $(1+o(1)) \pi^2 t^4 r^2 p^2(x,y)$ other edges that are geometrically close enough to $(v,v')$ that there could potentially be a path of length $r+1$ connecting them. However, there will not always be such a path between two potential edges, and it is more likely to exist where the vertices are denser. As such, there exists a nondecreasing function $f$ such that an average column in $B^{(r)}$ corresponding to an edge ending at $(x,y)$ will have $\pi^2 t^4 r^2 p^2(x,y)[f(p(x,y))+o(1)]$ nonzero entries. As a result, if we pick any point within $o(1)$ of $(\pm x_0,0)$ and then take a circle of radius $1/\ln(n)$ centered on that point, the graph induced by edges with vertices in that circle will have an average degree of $\pi^2 t^4 r^2 p^2(x_0,0)[f(p(x_0,0))+o(1)]$. Furthermore, we can find an arbitrarily large constant number of such circles such that no vertex in one of them has a path of length less than $\sqrt{n}/\ln(n)$ connecting it to a vertex in another. That strongly suggests that $B^{(r)}$ will have at least $m$ eigenvalues of $\pi^2 t^4 r^2 p^2(x_0,0)[f(p(x_0,0))+o(1)]$ based on the same logic used in the justification for Conjecture \ref{conj3}.

Now, let $B'$ be the restriction of $B^{(r)}$ to edges that's vertices both have positive $x$-coordinates. Next, let $w$ be its top eigenvector, scaled so that its entries are all positive. Also, for a given vertex $v$, let $(x_v,y_v)$ be the coordinates of $v$ and $\delta_v=\sqrt{(x_v-x_0)^2+y_v^2}$. We would expect that
\begin{align*}
\lambda \sum_{(v,v')} \delta^2_v w_{v,v'}&=\sum_{(v,v')} w_{v,v'}\sum_{(v'',v'''): B'_{(v'',v'''),(v,v')}=1} \delta^2_{v''}\\
&\approx\sum_{(v,v')} \left[\pi^2 t^4 r^2 p^2(x_v,y_v)f(p(x_v,y_v))\right] \left[\delta^2_v+2t^2r^2f(p(x_v,y_v))/3n\right]w_{v,v'}\\
&=\sum_{(v,v')} \left[\pi^2 t^4 r^2 p^2(x_v,y_v)f(p(x_v,y_v))\right]\delta^2_vw_{v,v'}\\
&\hspace{1 cm} +\sum_{(v,v')} \left[\pi^2 t^4 r^2 p^2(x_v,y_v)f(p(x_v,y_v))\right] \left[2t^2r^2f(p(x_v,y_v))/3n\right]w_{v,v'}\\
&\approx \sum_{(v,v')} \left[\pi^2 t^4 r^2 p^2(x_v,y_v)f(p(x_v,y_v))\right]\delta^2_vw_{v,v'}\\
&\hspace{1 cm} +\lambda\left[2t^2r^2f(p(x_v,y_v))/3n\right] \sum_{(v,v')}w_{v,v'}\\
&= \sum_{(v,v')} \lambda(1-\Theta(\delta_v^2))\delta^2_vw_{v,v'}+\Theta(\lambda t^2r^2/n) \sum_{(v,v')}w_{v,v'}.\\
\end{align*}
where the last equality follows from the expectation that $f(p(x_v,y_v))=f(p(x_0,0))-\theta(\delta_v^2)$. This implies that
\[\sum_{(v,v')} \delta_v^4 w_{v,v'}=\Theta\left(\sum_{(v,v')} r^2 w_{v,v'}/n\right).\]

So, most of the weight of $w$ is concentrated in the region of geometric space within $O(\sqrt[4]{r^2/n})$ of $(x_0,0)$. There are only $O(r\sqrt{n})$ vertices in this space, so random variation in vertex placement should give $\lambda$ a standard deviation of $\Theta(1/\sqrt[4]{r^2n}E[\lambda])$. That means that if we set $\lambda'$ equal to the top eigenvalue of the restriction of $B^{(r)}$ to edges whose vertices have negative $x$-coordinates, $\lambda'$ will probably differ from $\lambda$ by a factor of $(1+\Theta(1/\sqrt[4]{r^2n}))$. Furthermore, $(-x_0,0)$ is $\Theta(\sqrt{n})$ edges away from $(x_0,0)$. $\ln(\sqrt[4]{r^2n})=o(\sqrt{n}/r)$, which means that the eigenvector of $B^{(r)}$ with the second largest eigenvalue is unlikely to combine these to any significant degree. So, the second eigenvector of $B^{(r)}$ will be concentrated around $(-x_0,0)$ or $(x_0,0)$ but not both.

%\Cnote{I probably need to explain this better.}

That means that assigning vertices to communities based on whether or not the sum of the entries of the eigenvector corresponding to their edges is greater than the median would essentially assign vertices that are sufficiently close to $(\pm x_0,0)$ to one community and those that are farther away to the other community. This fails to attain optimal accuracy for the same reasons as in Conjecture \ref{conj3}.
\end{proof}

%\Enote{colin: I removed the remark below as we start having a lot of material}
%\begin{remark}
%It might be possible to attain optimal accuracy by somehow balancing an eigenvector with large values near $(x_0,0)$ with one near $(-x_0,0)$, but we are not sure if one would even encounter both of those, and the method would probably be fairly specific to the GBM.
%\end{remark}

There are a couple of key problems that contribute to the predicted failure of the algorithm in the scenario above. First of all, the graph is divided into high degree and low degree regions. Edges from a region almost always lead to the same region, or at least a similar region. As a result, eigenvectors focus almost completely on the high degree regions and assign negligable values to edges in low degree regions. Secondly, for large $s$ there are two different regions of maximum degree that are $\Theta(\sqrt{n})$ edges away from each other. Also, the top eigenvalue of the restriction of $B^{(r)}$ to one of the densest regions has $\omega(1/r\sqrt{n})$ random variance due to the positioning of vertices. As a result, each of the main eigenvectors of $B^{(r)}$ will tend to focus on one maximum degree region or the other rather than being symmetric or antisymmetric. Thus, the algorithm will tend to classify vertices based on their distance to a high degree region rather than classifying them based on what side of the graph they are on.  

\begin{comment}
On another note, $B^{(r)}$ maps each directed edge to the sum of all directed edges that can be reached from it by means of a path of length $r+1$. $A^{(r)}$ is generally similar to $B^{(r)}$, although it would allow some backtracking that $B^{(r)}$ filters out. For instance, in the absence of cycles, $(B^{(r)})^s_{(v_1,v_2),(v'_1,v'_2)}$ is a count of the number of length $rs+1$ paths from $(v'_1,v'_2)$ to $(v_1,v_2)$. In this case, $(A^{(r)})^s_{v,v'}$ is a count of the number of sequences of $s$ length $r$ paths such that the first path starts at $v'$, each subsequent path starts at the previous path's endpoint, and the last path ends at $v$. Note that PoweringAlgorithm simply changes the graph to one that has an adjacency matrix of $A^{(r)}$. We believe that $A^{(r)}$ suceeds at detecting communities on the stochastic block model in the following sense.
\end{comment}

\begin{conjecture}
Let $a,b\ge 0$ such that $(a-b)^2>2(a+b)$, $r=\Theta(\sqrt{\log (n)})$, and $G\sim\sbm(n,a,b)$. Next, let $w$ be the eigenvector of $A^{(r)}$ with eigenvalue of second largest absolute value. Dividing the vertices of $G$ into those with above median entries in $w$ and those with below median entries in $w$ solves weak recovery.
\end{conjecture}

\begin{proof}[Justification]
First, let $w'\in\mathbb{R}^n$ be the vector such that $w'_v=(-1)^{X_v}$ and $w''\in\mathbb{R}^n$ be the vector such that $w''_v=1$. Next, note that for small $t$, and a vertex $v$, the expected number of vertices $t$ edges away from $v$ is approximately $(\frac{a+b}{2})^t$, and approximately $(\frac{a-b}{2})^t$ more of these are in the same community as $v$ than are in the other community on average. For any $t>0$ and vertices $v$ and $v'$, we have that $(A^{(r)})^t_{v,v'}$ is the number of sequences of $t$ nonbacktracking walks of length $r$ or less such that the first nonbacktracking walk starts at $v$, each subsequent nonbacktracking walk starts where the previous one ended, and the last nonbacktracking walk ends at $v'$. If such a sequence has nonbacktracking walks with lengths of $r-c_1$, $r-c_2$,...,$r-c_t$ and there are $b_1$,$b_2$,...,$b_{t-1}$ edges of overlap between sucessive nonbacktracking walks, then the combined walk is essentially a nonbacktracking walk of length $rt-\sum c_i-2\sum b_i$ from $v$ to $v'$ with nonbacktracking walks of length $b_1$, $b_2$,..., $b_{t-1}$ branching off of it. So for small $t>0$, we have that
\begin{align*}
w''\cdot \left[A^{(r)}\right]^t w''/n&\approx \sum_{c_1,...,c_t,b_1,...,b_{t-1}\ge 0} \left(\frac{a+b}{2}\right)^{rt-\sum c_i-2\sum b_i}\prod_{i=1}^{t-1} \left(\frac{a+b}{2}\right)^{b_i}\\
&= \left(\frac{a+b}{2}\right)^{rt} \left(\frac{a+b}{a+b-2}\right)^{2t-1}.
\end{align*}
Similarly,
\begin{align*}
w'\cdot \left[A^{(r)}\right]^t w'/n&\approx \sum_{c_1,...,c_t,b_1,...,b_{t-1}\ge 0} \left(\frac{a-b}{2}\right)^{rt-\sum c_i-2\sum b_i}\prod_{i=1}^{t-1} \left(\frac{a+b}{2}\right)^{b_i}\\
&= \left(\frac{a-b}{2}\right)^{rt} \left(\frac{a-b}{a-b-2}\right)^{t}\left(\frac{(a-b)^2}{(a-b)^2-2(a+b)}\right)^{t-1}
\end{align*}
assuming that $(a-b)^2>2(a+b)$. Alternately, if $w_1,...,w_n$ is an orthonormal eigenbasis for $A^{(r)}$ with corresponding eigenvalues of $\lambda_1,...,\lambda_n$, then for all $t\ge 0$, we have that $w''\cdot \left[A^{(r)}\right]^t w''/n=\frac{1}{n}\sum_{i=1}^n \lambda_i^t (w_i\cdot w'')^2$. That means that
\begin{align*}
&\sum_{i=1}^n \lambda_i^2\left (\lambda_i-\left(\frac{a+b}{2}\right)^r\left(\frac{a+b}{a+b-2}\right)^2\right)^2 (w_i\cdot w'')^2\\
&=\sum_{i=1}^n [\lambda_i^4-2\left(\frac{a+b}{2}\right)^r\left(\frac{a+b}{a+b-2}\right)^2\lambda_i^3+\left(\frac{a+b}{2}\right)^{2r}\left(\frac{a+b}{a+b-2}\right)^4\lambda_i^2](w_i\cdot w'')^2\\
&=w''\cdot \left[A^{(r)}\right]^4 w''-2\left(\frac{a+b}{2}\right)^r\left(\frac{a+b}{a+b-2}\right)^2 w''\cdot \left[A^{(r)}\right]^3 w''\\
&+\left(\frac{a+b}{2}\right)^{2r}\left(\frac{a+b}{a+b-2}\right)^4 w''\cdot \left[A^{(r)}\right]^2 w''\\
&\approx n\left(\frac{a+b}{2}\right)^{4r} \left(\frac{a+b}{a+b-2}\right)^{7}-2n\left(\frac{a+b}{2}\right)^{4r} \left(\frac{a+b}{a+b-2}\right)^{7}+n\left(\frac{a+b}{2}\right)^{4r} \left(\frac{a+b}{a+b-2}\right)^{7}\\
&=0.
\end{align*}

So, any eigenvector that is significantly correlated with $w''$ must have an eigenvalue of approximately $0$ or approximately $(\frac{a+b}{2})^r(\frac{a+b}{a+b-2})^2$. We know that $\sum_{i=1}^n \lambda^2_i (w_i\cdot w'')=w''\cdot [A^{(r)}]^2 w''\approx n(\frac{a+b}{2})^{2r}(\frac{a+b}{a+b-2})^3$, so they cannot all have small eigenvalues. Therefore, $w''$ is correlated with an eigenvector of $A^{(r)}$ with an eigenvalue of approximately $(\frac{a+b}{2})^r(\frac{a+b}{a+b-2})^2$. By similar logic, $w'$ must be correlated with an eigenvector of $A^{(r)}$ with an eigenvalue of approximately $(\frac{a-b}{2})^r(\frac{a-b}{a-b-2})(\frac{(a-b)^2}{(a-b)^2-2(a+b)})$. For a high degree vertex, $v$, we would expect that
\begin{align*}
e_v \cdot (A^{(r)})^2 e_v&=||A^{(r)} e_v||_2^2\\
&\approx \sum_{i=1}^r \mathrm{deg}(v)\left(\frac{a+b}{2}\right)^{i-1}\\
&\approx \mathrm{deg}(v)\cdot \frac{2}{a+b-2} \left(\frac{a+b}{2}\right)^r.
\end{align*}
The largest degree of a vertex in $G$ is $\Omega(\log(n)/\log(\log(n)))$ with high probability, so this means that there will be a large number of eigenvectors with eigenvalues of up to $\Omega\left(\sqrt{\frac{2}{a+b-2} (\frac{a+b}{2})^r\frac{\log(n)}{\log(\log(n))}}\right)$ that are concentrated in small regions of the graph. We can upper bound the absolute value of any eigenvalue corresponding to an eigenvector that is concentrated in a small region of the graph by using the fact that such an eigenvector would be significantly correlated with the unit vector at one of the vertices it is focused on. Given a high degree vertex $v$, there are approximately $\mathrm{deg}(v)(\frac{a+b}{2})^{r'-1}$ vertices that are $r'$ edges away from $v$, and the corresponding entries in $(A^{(r)})^t$ will tend to be reasonably close to each other. So, we would expect that

\begin{align*}
e_v\cdot (A^{(r)})^{2t} e_v&=||(A^{(r)})^t e_v||_2^2\\
&=\sum_{v'\in G} [(A^{(r)})^t_{v,v'}]^2\\
&=\sum_{r'=0}^{rt} \sum_{v'\in G: d(v,v')=r'} [(A^{(r)})^t_{v,v'}]^2\\
&=e^{O(t)}\sum_{r'=0}^{rt} \frac{1}{|\{v'\in G:d(v,v')=r'\}} \left[\sum_{v'\in G: d(v,v')=r'} (A^{(r)})^t_{v,v'}\right]^2\\
&=e^{O(t)}\sum_{r'=0}^{rt} \mathrm{deg}(v)^{-1}\left(\frac{a+b}{2}\right)^{1-r'} \left[\sum_{v'\in G: d(v,v')=r'} (A^{(r)})^t_{v,v'}\right]^2\\
&=e^{O(t)}(rt)^{-O(1)}\left[ \sum_{r'=0}^{rt} \mathrm{deg}(v)^{-1/2}\left(\frac{a+b}{2}\right)^{(1-r')/2}\sum_{v'\in G: d(v,v')=r'} (A^{(r)})^t_{v,v'}\right]^2\\
&=e^{O(t)}(rt)^{-O(1)}\left[ \mathrm{deg}(v)^{-1/2} \sum_{v'\in G}\left(\frac{a+b}{2}\right)^{(1-d(v,v'))/2} (A^{(r)})^t_{v,v'}\right]^2\\
&=e^{O(t)}(rt)^{-O(1)}[ \mathrm{deg}(v)^{-1/2} \sum_{c_1,...,c_t,b_1,...,b_{t-1}} \left(\frac{a+b}{2}\right)^{(1-rt+\sum c_i+2\sum b_i)/2}\\
&\indent\indent \left(\frac{a+b}{2}\right)^{rt-\sum c_i-2\sum b_i}\prod_{i=1}^{t-1} \left(\frac{a+b}{2}\right)^{b_i}]^2\\
&=e^{O(t)}(rt)^{-O(1)}\left[ \mathrm{deg}(v)^{-1/2} \sum_{c_1,...,c_t,b_1,...,b_{t-1}} \left(\frac{a+b}{2}\right)^{rt/2+1/2-\sum c_i/2}\right]^2\\
&=e^{O(t)}(rt)^{-O(1)}\left[ \mathrm{deg}(v)^{-1/2} (r+1)^{t-1} \left(\frac{a+b}{2}\right)^{rt/2+1/2} \left(\frac{1}{1-\sqrt{2/(a+b)}}\right)^t\right]^2\\
&=e^{O(t)}(rt)^{-O(1)} \mathrm{deg}(v)^{-1} r^{2t-2} \left(\frac{a+b}{2}\right)^{rt+1} \left(1-\sqrt{2/(a+b)}\right)^{-2t}.\\
\end{align*}
This means that any eigenvector of $A^{(r)}$ that is concentrated in a small region of the graph will have an eigenvalue that is within a constant factor of $r \left(\frac{a+b}{2}\right)^{r/2}$, and we can show that for a random vector $w'''$ we will have
\begin{align*}
w'''\cdot (A^{(r)})^{2t} w'''/||w'''||_2^2&\approx E[||(A^{(r)})^t w'''||_2^2/||w'''||_2^2]\\
& \approx E[||(A^{(r)})^t e_v||_2^2]\\
&=\Theta\left(e^{O(t)}(rt)^{-O(1)} r^{2t-2} \left(\frac{a+b}{2}\right)^{rt+1}.\right)
\end{align*}
So, it appears that eigenvalues other than the dominant eigenvector and $w$ will tend to have eigenvalues of $O(r \left(\frac{a+b}{2}\right)^{r/2})$. This will always be less than $(\frac{a-b}{2})^r(\frac{a-b}{a-b-2})(\frac{(a-b)^2}{(a-b)^2-2(a+b)})$ for sufficiently large $n$. So, we expect that $w$ will be correlated with the communities. That suggests that dividing the vertices into those with above median entries in $w$ and those with below median entries in $w$ will solve weak recovery.
\end{proof}

\section{Powering implementations}\label{implementations}
We argued in previous sections that graph powering can circumvent the issues of nonbacktracking or normalized Laplacian operators on the discussed objectives, but if some regions of the graph have significantly higher degrees than others, graph powering may require impractically large values of $r$. We next discuss how one can combine graph powering with other processing steps to reduce the value of $r$.

\subsection{Cleaning and normalizing}

A first possibility is to combine graph powering on $A$ or $B$ with degree normalization, e.g., to use the normalized Laplacian of $A^{(r)}$. Since graph powering will provide a first round of regularization of the density variations in the graph, composing powering with degree normalization may not face the issues mentioned previously for the normalized Laplacian. 

As long as $r=\omega(1)$, we are effectively only weighting the information provided by a path based on the inverses of $o(1)$ of the degrees of the vertices on the path. That reduces the distortion to the flow of information this causes enough that aforementioned distortion will not cause the algorithm to fail at weak recovery on the SBM. The initial powering step also tends to render isolated regions of the graph significantly less isolated, which mitigates their effects on the eigenvectors.

However, it may not mitigate the effects of isolated regions enough. On the SBM there are paths of length $\Omega(\log(n))$ such that each such path only has one edge connecting it to the rest of the graph. Unless we use $r=\Omega(\log(n))$ with a sufficiently large constant, these paths will still be isolated enough after powering to result in eigenvectors with eigenvalues greater than any eigenvalue of an eigenvector corresponding to the communities. This would result in the algorithm failing at weak recovery. If we do use a sufficiently large value of $r$, then due to the structure of the SBM, the typical vertex would be within $r$ edges of $n^{\Omega(1)}$ other vertices, and the algorithm would take $n^{1+\Omega(1)}$ space and time, which would be notably suboptimal.

A reasonable solution to this would be to delete problematic sections of the graph before powering. Repeatedly deleting all leaves would be a good start because branches are one of the potential forms an isolated region could take and leaves generally do not provide much information on their neighbors' communities. However, there are other possible forms of isolated regions, such as long paths with no edges connecting their interiors to the rest of the graph. We would like to remove these as well, but they do provide evidence that the vertices on either side of the path are in the same community. As a result, for any constant $c$, deleting all such paths of length $c$ or more would discard enough evidence that it would render weak recovery impossible for the SBM with parameters sufficiently close to the threshold. So, we would need to pick some function $r'=\omega(1)$ and only delete isolated paths that have a length of $r'$ or greater.

This would probably work on the SBM or GBM. However, some other models have vertices with extremely high degree. These can potentially make community recovery easier, but they can also potentially create problems by adding large numbers of unhelpful edges once the graph is raised to a significant power. As an extreme example of this, a vertex that was adjacent to every other vertex would cause any nontrivial powering of the graph to be the complete graph. As such, it may help to remove vertices of excessively high degree. This may not always help, and it is important not to overdo it, but in some cases it may help to delete the $m$ highest degree vertices for some $m=o(n)$.

That still leaves the issue of ensuring that there are no remaining regions of the graph that are still too isolated after powering. In order to do that, we would like to ensure that for every vertex $v$, the majority of the vertices that are within $o(r)$ edges of $v$ have typical levels of connectivity to the rest of the graph. If every vertex has some constant probability of having significantly lower connectivity than normal, and the vertices $r$ edges away from $v$ are reasonably close to being independent, then we would want there to be $\omega(\log n)$ vertices $o(r)$ edges away from $v$ in order to ensure that. The number of vertices that are typically a given distance away from a designated vertex varies wildly depending on the model. In the SBM, the number of vertices a given distance away scales geometrically with the distance, while in the GBM the number of vertices is linear in the distance. More generally, define $f(m)$ to be the typical number of vertices $m$ edges away from a given vertex in the graph. If the graph has a diameter of $d$, then that presumably means that $f^{-1}(n)=\Theta(d)$, which suggests that perhaps $f^{-1}(\log(n))=\Theta(\log(d))$. In all models we have considered, this either gives an accurate estimate of $f^{-1}(\log(n))$ or overestimates it.

However, we want to ensure that every vertex has $\omega(\log(n))$ vertices $o(r)$ edges away from it, rather than merely ensuring that the typical vertices do. Our previous deletion of all leaves and all isolated paths of length $r'$ or less ensures that for every vertex $v$ remaining in the graph, $v$ has at least two edges. It also ensures that if one starts a nonbacktracking walk at any edge, one will never hit a dead end, and one will have a choice of edges at least every $r'$ steps. Steps that continue along an isolated path do not seems to contribute to increasing the number of reachable vertices, but the others do. So, that would suggest that perhaps one should require that $r=\omega(r'\log(d))$. Unfortunately, that does not quite work because in models like the GBM, one could essentially take a long isolated path and then add a few cycles to prevent it from being deleted. Luckily, there is a limited amount of room for such a construction to exist in an isolated region of the graph, which essentially means that if it has more than $\sum_{i=1}^{f^{-1}(\ln(m))} f(i)\approx \ln^2(m)$ vertices, it will have to have at least $m$ vertices near the edge of the region. So, it is sufficient to require that $r=\omega(\max(r'\log(d),\log^2(d)))$, $r=o(d)$, and $r'=\omega(1)$. There are many possible values within this range, but to keep things relatively simple, we can set $r=\Theta(\ln^3(d))$ and $r'=\Theta(\ln^{3/2}(d))$.

\subsection{The spectral meta-algorithm}
Based on the points from previous sections, we arrive at the following algorithm. 

\begin{definition}
For a connected graph $G$ and parameters $\tau,c \ge 0$, define $\psi(G)$ as the matrix output by the following procedure, where $r:=c\ln^3(\diam(G))$.
\begin{enumerate}

\item Cleaning:
\begin{itemize}
\item Delete all vertices of degree more than $\tau$.
\item Repeat the following until no vertex is deleted:
\begin{itemize}
\item Delete all the leaves from $G$;
\item Delete all $\sqrt{r}$-segments from $G$: If $G$ has a vertex $v$ such that every vertex within $\sqrt{r}$ edges of $v$ has degree $2$ or less, delete $v$.
\end{itemize}
\item If the graph has multiple components, delete all but the largest component.
\end{itemize}

\item Powering: 

Take the $r$-th power of the graph output in the previous step (i.e., add an edge between every pair of vertices that had a path of length $r$ or less between them).

\item Normalizing:

Construct the random walk matrix of the graph output in the previous step, i.e., $\psi(G):=D_{\tilde{G}}^{-1}A_{\tilde{G}}$ where $\tilde{G}$ is the graph output at the previous step.
\end{enumerate}
\end{definition}

This leads to the following algorithm.\\ 

\noindent
{\it Spectral meta-algorithm (for two clusters)}.\\
Input: a graph $G$, a choice of the two parameters for $\psi(G)$.\\
Algorithm: take the eigenvector $w$ of $\psi(G)$ corresponding to its second largest eigenvalue, and assign each vertex in $G$ with a positive entry in $w$ to one cluster, and all other vertices to the other cluster.

\subsection{Powering and sparsifying}
One downside that remains for the previous spectral meta-algorithm is that the powering step will typically yield a graph that has $\Omega(n)$ vertices and at least a polylogarithmic number of edges per vertex which could easily be inconveniently large. For the case of sparse graph, which is the main focus of this paper, this is typically not an issue (e.g., the complexity of the algorithm in Theorem \ref{ks2} is already quasi-linear). However, for denser graphs, this may be impractical. One possible approach to mitigating this would be to randomly delete most of the vertices. 

As long as the main eigenvalues of the graph's random walk matrix are sufficiently large relative to the eigenvalues corresponding to noise eigenvectors (i.e., eigenvectors resulting from the random variation in the graph) and we do not delete too large a fraction of the vertices, this will probably not affect the eigenvectors too much, and we can guess communities for the vertices that were deleted based on the communities of nearby vertices.

Of course, if we still have to compute the powered graph in the first place, this does not help all that much. However, we can address that by alternating between raising the graph to a power and deleting vertices. For a graph $G$ and integer $r$, if the distance between $v$ and $v'$ in $G$ is significantly less than $2r$, there will typically be a large number of vertices that are within $r$ edges of both $v$ and $v'$. As such, if $H$ is a graph formed by taking $G^{(r)}$ and then deleting most of the vertices, and $v,v'\in H$, then it is fairly likely that $H^{(2)}$ will contain an edge between $v$ and $v'$. Therefore, the result of alternating between powering a graph and deleting vertices at random is likely to be reasonably close to what we would have gotten if we had simply raised it to a high power first and then deleted most of the vertices. So, it should still allow us to assign the vertices to communities accurately while saving us the trouble of computing the full powered graph.

Another possibility would be to edge-sparsify the powered graph, as in \cite{sparsify}. Since it would be computationally expensive to first calculate $A^{(r)}$ and then sparsify it, the sparsification could be done by iteratively squaring the graph and then sparsifying to keep the degree low. One would have to prove that the result of this chain of powerings and sparsifications would indeed approximate the spectral properties of the graph power.\footnote{Other relevant papers for this type of approaches based mainly on the classical rather than graph power are \cite{peng2014efficient,yin2,murtagh2017derandomization,yin1}. See also \cite{yin3} for lower bounds on the expansion of classical powers.}

\section{Additions and open problems}
\subsection{Powering weighted graphs}
Throughout this analysis so far, we have assumed that our original graph is unweighted. To use our algorithms to recover communities on a weighted graph, we would first have to consider what the weights mean. The simplest possibility would be that the degree of evidence that an edge provides that its vertices are in the same community is proportional to its weight. For the rest of this section, we will assume that the weights are positive and have this property; other scenarios may require applying some function to the weights when they are used in the adjusted algorithm, or otherwise modifying it to account for the weights.

To use our algorithms to recover communities on a weighted graph, we would need to make several modifications. First of all, it would be possible for the graph to have structures that were essentially isolated paths or leaves, but that technically were not due to the presence of some extra edges with negligible weights. As such, we would want to modify the cleaning phase so that it deletes vertices that are sufficiently similar to leaves or centerpoints of isolated paths. One way to do that would be to have the algorithm do the following. If there is any vertex such that at least $1-1/\sqrt{r}$ of the total weight of its edges is concentrated on a single edge, the algorithm would consider it a leaf and delete it. Likewise, if there is a path of length $2\sqrt{r}+2$ such that every vertex in the interior of the path has at least $1-1/\sqrt{r}$ of the total weight of its edges concentrated on edges in the path, the algorithm would consider it an isolated path and delete its center vertex.

Secondly, powering would need to assign weights to the edges it is adding based on the weights of the edges in the path. To the degree that our goal in powering is to account for indirect evidence between communities while avoiding doublecounting and feedback, we would probably want to set the weight of the edge powering puts between two vertices equal to the weight of the strongest path between them. Generally, the degree of evidence a path provides that the vertices on either end are in the same community should be proportional to the product of the weights of its edges. As such, we should consider a path as having a weight equal to the product of the weights of its edges times some function of the path's length. The obvious choice there would be a function that is exponential in the length of the path. If it grows too slowly with length, we would only assign significant weights to short paths, in which case powering the graph would have little effect on it, and the problems it was added to fix might prevent the algorithm from working. If the function grows too quickly with length, then the algorithm would essentially ignore short paths, which would be suboptimal, although it could provide meaningful reconstructions. It seems like finding $c$ such that the typical vertex in the graph has a couple of edges of weight $c$ or more and then setting a path's length equal to the product of its edges' weights divided by $c$ to the path length would generally be relevant. However, it is safer to use a smaller value of $c$ if it is unclear what an appropriate value is.

 Finally, the normalization phase would work as usual. 

\subsection{Conjectures}
We showed in this paper that raising a graph to a power proportional to its diameter and then taking the second eigenvector of its adjacency matrix solves weak recovery on the SBM. We also believe that it solves weak recovery on the GBM as stated by Conjecture \ref{conj_gbm} and heuristically justified.  

However, since the power used is inconveniently large, we would prefer to use the spectral meta-algorithm because it allows us to use a power that is polylogarithmic in the graph's diameter. We believe that the spectral meta-algorithm also succeeds at weak recovery on these models as stated below.

\begin{conjecture}
Choose $a$ and $b$ such that there is an efficient algorithm that solves weak recovery on $\sbm(n,a,b)$. Then the spectral meta-algorithm solves weak recovery on $\sbm(n,a,b)$.
\end{conjecture}

\begin{conjecture}
For any $s$ and $r$ such that $\gbm(n,s,t)$ has a giant component, there exists $\delta>1/2$ such that the following holds. The spectral meta-algorithm recovers communities with accuracy $\delta$ on $\gbm(n,s,t)$ and there is no algorithm that recovers communities with accuracy $\delta+\Omega(1)$ on $\gbm(n,s,t)$.
\end{conjecture}

\begin{comment}
\begin{conjecture}
There exist positive constants $a$, $b$, $s$, and $r$, such that the spectral meta-algorithm solves weak recovery on $\hbm(n,s,r/\sqrt{n},a/n,b/n)$.
\end{conjecture}
\end{comment}

%\Cnote{Unfortunately, this is not going to solve weak recovery on the HBM because one can have something along the lines of a pair of isolated paths with corresponding vertices connected. It is still isolated enough to result in a localized eigenvector of higher eigenvalue than the one corresponding to the communities, but it will not get deleted during the cleaning step. To fix that we would need a more general method to delete large clusters of points that were mostly disconnected from the rest of the graph.}

\section{Proofs}

\subsection{Proof of Theorem \ref{ks1}}
Our proof of Theorem \ref{ks1} is based on the following theorem of Massouli\'e \cite{massoulie-STOC}, which is analogous to Theorems \ref{ks1} and \ref{ks2}, but for the matrix $A^{\{r\}}$ counting self-avoiding walks (i.e., paths) of length $r$:

\begin{theorem}[Spectral separation for the self-avoiding-walk matrix; proved in \cite{massoulie-STOC}]\label{ks3}
Let $(X,G)$ be drawn from $\sbm(n,a,b)$ with $(a+b)/2 > 1$.
Let $A^{\{r\}}$ be the length-$r$-self-avoiding-walk matrix of $G$ ($A^{\{r\}}_{ij}$ equals the number of self-avoiding walks of length $r$ between $i$ and $j$), and $r= \e \log(n)$ such that $\e > 0$, $\e \log (a+b)/2 <1/4$. 
Then, with high probability, for all $k \in \{r/2,\ldots,r\}$,

\begin{multicols}{2}
\begin{enumerate}
\item[A.] If $\left(\frac{a+b}{2}\right) < \left(\frac{a-b}{2}\right)^2$, then
\begin{enumerate}
\item[1.]  $\lambda_1(A^{\{k\}}) \asymp \left(\frac{a+b}{2} \right)^{k}$,
\item[2.] $\lambda_2(A^{\{k\}}) \asymp \left(\frac{a-b}{2} \right)^{k}$,
\item[3.] $|\lambda_3(A^{\{k\}})| \leq \left(\frac{a+b}{2} \right)^{k/2} \log(n)^{O(1)}$.
\end{enumerate}
\item[B.]If $\left(\frac{a+b}{2}\right) > \left(\frac{a-b}{2}\right)^2$, then
\begin{enumerate}
    \item[1.] $\lambda_1(A^{\{k\}}) \asymp \left(\frac{a+b}{2} \right)^{k}$,
    \item[2.] $|\lambda_2(A^{\{k\}})| \leq \left(\frac{a+b}{2} \right)^{k/2} \log(n)^{O(1)}$.
    \item[~]
\end{enumerate}
%    \item[~]
\end{enumerate}
\end{multicols}

Furthermore, for all $k \in \{r/2,\ldots,r\}$, $\phi_2(A^{\{k\}})$ with the rounding procedure of \cite{massoulie-STOC} achieves weak recovery whenever $\left(\frac{a+b}{2}\right) < \left(\frac{a-b}{2}\right)^2$, i.e., down to the KS threshold.
\end{theorem}

\paragraph{Remark on Theorem \ref{ks3}} Theorem \ref{ks3} does not appear in the above form in \cite{massoulie-STOC}. In particular, case $B$ of Theorem \ref{ks3} is not addressed by \cite{massoulie-STOC}, but it can be proved with the same techniques as case $A$ (in fact, it is simpler). The polylogarithmic factors in the bounds of \cite{massoulie-STOC} on $\lambda_1$ and $\lambda_2$ can removed with a more careful analysis, along the lines of the later work \cite{bordenave} on the spectrum of the non-backtracking operator. Similarly, the $n^{\epsilon}$ factor in the bound of \cite{massoulie-STOC} on $\lambda_3$ can be seen to be in fact a $(\log n)^{O(1)}$ factor. Finally, the bounds in \cite{massoulie-STOC} are stated for $A^{\{r\}}$ alone, not for $A^{\{k\}}$ for all $k \in \{r/2,\ldots,r\}$. However, the proof of \cite{massoulie-STOC} shows that there is a constant $C > 0$ such that for each individual $k \in \{r/2,\ldots,r\}$ the bounds on the top eigenvalues of $A^{\{k\}}$ hold with probability $\geq 1 - C(\log n)^{-2}$. A union bound over $k \in \{r/2,\ldots,r\}$ then gives Theorem \ref{ks3} as stated. (This union bound is needed for the proof of Theorem \ref{ks2} in the next section, but is not needed for the proof of Theorem \ref{ks1}.)

We will only prove case $A$ of Theorem \ref{ks1}, since the argument for case $B$ is similar and simpler. Writing $\alpha := (a+b)/2$ and $\beta := (a-b)/2$, the proof can be broken down into two steps:
\begin{enumerate}
\item We show that the distance-$r$ matrix is a small perturbation of the length-$r$-self-avoiding-walk matrix. In particular, we show that with high probability, the difference $B = B(r) := A^{[r]} - A^{\{r\}}$ has small spectral norm $\|B\|_2 = O(\alpha^{r/2} \log^3 n)$.
\item We use matrix perturbation theory to prove that the top eigenvalues and eigenvectors of $A^{[r]}$ behave like the top eigenvalues and eigenvectors of $A^{\{r\}}$. On the event that the bounds in Theorem \ref{ks3} and Step 1 hold, \begin{enumerate}
\item $\|B\|_2 = o(\lambda_2(A^{\{r\}})),$ so Weyl's inequality \cite{weyl1912asymptotische} gives $$\lambda_1(A^{[r]}) \asymp \lambda_1(A^{\{r\}}) \asymp \alpha^r,$$ $$\lambda_2(A^{[r]}) \asymp \lambda_2(A^{\{r\}}) \asymp \beta^r,$$ $$|\lambda_3(A^{[r]})| \leq |\lambda_3(A^{\{r\}})| + \|B\|_2 \leq \alpha^{r/2} (\log n)^{O(1)}$$.
\item $\|B\|_2 = o(\max\{\lambda_1(A^{\{r\}}) - \lambda_2(A^{\{r\}}), \lambda_2(A^{\{r\}}) - \lambda_3(A^{\{r\}})\}),$ so by the Davis-Kahan Theorem \cite{daviskahan}, $\phi_1(A^{[r]})$ and $\phi_2(A^{[r]})$ are asymptotically aligned with $\phi_1(A^{\{r\}})$ and $\phi_2(A^{\{r\}}),$ respectively, which is enough for the rounding procedure of \cite{massoulie-STOC} to achieve weak recovery.
As a reminder, the Davis-Kahan theorem states:
\begin{theorem}[Davis-Kahan Theorem]\label{thm:daviskahan}
Suppose that $\bar{H} = \sum_{j=1}^n \bar{\mu}_j\bar{u}_j\bar{u}_j^T$ and $H = \bar{H} + E$, where $\bar{\mu}_1 \geq \dots \geq \bar{\mu}_n$, $\|\bar{u}_j\|_2 = 1$ and $E$ is symmetric. Let $u_j$ be a unit eigenvector of $H$ corresponding to its $j$-th largest eigenvalue, and $\Delta_j = \min \{\bar{\mu}_{j-1} - \bar{\mu}_j, \bar{\mu}_j - \bar{\mu}_{j+1}\}$, where we define $\bar{\mu}_0 = + \infty$ and $\bar{\mu}_{n+1} = -\infty$. We have \begin{equation}\min_{s = \pm 1} \|su_j - \bar{u}_j\|_2 \lesssim \frac{\|E\|_2}{\Delta_j},\end{equation} where $\lesssim$ only hides an absolute constant.
\end{theorem}
\end{enumerate}
as desired.
\end{enumerate}

It remains to work out the details for the first step: to prove that $\|B\|_2 = O(\alpha^{r/2} \log^3 n)$ with high probability.

To understand the intuition behind our argument, it helps to imagine what would happen if the underlying graph were a tree instead of an SBM. In the tree case, there would be exactly one self-avoiding walk between every pair of vertices, and the length of this walk would be equal to the distance between the two vertices. In other words, in the tree case, the matrices $A^{\{r\}}$ and $A^{[r]}$ would be equal.  

While the SBM is (with high probability) not a tree, it is with high probability locally tree-like. This means that for small $r$, most vertices don't have cycles in their $r$-neighborhoods. Therefore, most vertices' $r$-neighborhoods are trees, and hence $A^{\{r\}} \approx A^{[r]}$.

The observation that the SBM is locally tree-like can be formalized:

\begin{lemma}[Lemma 4.2 of \cite{massoulie-STOC}]
Let $E_1$ be the event that no vertex has more than one cycle in its $r$-neighborhood. For $r = \e \log n$ and $\e \log \alpha < 1/4$,  $E_1$ occurs with high probability.
\end{lemma}

%\Pnote{bad citation here.  Should it be \cite{massoulie-STOC}?  This same paper is cited a few times below}

Conditioning on $E_1$, we can define the equivalence relation $\sim$ so that $v \sim w$ if and only if there is a cycle in the intersection of the $r$-neighborhoods of $v$ and $w$. This is a well-defined equivalence relation because every vertex has at most one cycle in its $r$-neighborhood. The relation $\sim$ is useful because of item (i) of the following proposition (proof postponed) connecting $\sim$ to the structure of $B$:
\begin{proposition}\label{claim:Misblockdiagonalandbounded}
Condition on $E_1$. Then for all $i,j \in V(G)$:
\begin{enumerate}
    \item[(i)] $B_{i,j} \neq 0 \implies i \sim j$.
    \item[(ii)] $B_{i,j} \neq 0 \implies$ there are at least two length-$(\leq r)$ paths from $i$ to $j$.
    \item[(iii)] $|B_{i,j}| \leq 1$.
    \item[(iv)] There are at most two length-$(\leq r)$ paths from $i$ to $j$.
\end{enumerate}
\end{proposition}

We condition on $E_1$ in the rest of the proof, since it holds with high probability. By item (i) of Proposition \ref{claim:Misblockdiagonalandbounded}, $B$ is a block-diagonal matrix, where each block $B_{S \times S}$ corresponds to an equivalence class $S \subseteq [n]$ of $\sim$. Therefore, it suffices to separately bound the spectral norm of each block $B_{S \times S}$. To do this, we introduce the following event:

\begin{lemma}[Theorem 2.3 of \cite{massoulie-STOC}]\label{lem:neighborhoodgrowthbound} Let $E_2 = E_2(C)$ be the event that for all vertices $i \in V(G)$, for all $t \in [r]$, the following holds:
\begin{equation} \label{ineq:neighborhoodgrowthbound} |\{j : d_{G}(i,j) \leq t\}| \leq C (\log n)^2 \alpha^t.\end{equation} There is large enough $C$ that $E_2(C)$ holds with high probability.
\end{lemma}

Informally, $E_2$ is the event that for all $t \in [r]$, each vertex's $t$-neighborhood is not much larger than $\alpha^t$.

From now on, also condition on $E_2$, since it holds with high probability. Suppose $S$ is the set of vertices in some equivalence class of $\sim$. Let $H \subseteq G$ be the cycle that is shared by the $r$-neighborhoods of the vertices in $S$. Let $e$ be an edge of $H$. Then for every $i,j \in S$ such that there are two length-$(\leq r)$ self-avoiding walks from $i$ to $j$, at least one of the paths must contain $e$. Otherwise, the cycle $H$ is not the only cycle in the $r$-neighborhood of $i$. So by item (ii) of Proposition \ref{claim:Misblockdiagonalandbounded}, $$|\{\{i,j\} \in S \mid B_{i,j} \neq 0\}| \leq |\{(\leq r)\mbox{-length paths } P \subseteq G \mid e \in E(P)\}|.$$
For any $t \in [r]$, $$|\{t\mbox{-length paths } P \subseteq G \mid e = (u,v) \in E(P)\}| \leq $$ $$\sum_{l=0}^r |\{l\mbox{-length paths } P \subseteq G \mid u \in V(P)\}| \cdot |\{(t-l-1)\mbox{-length paths } P \subseteq G \mid v \in V(P)\}|.$$

By item (iv) of Proposition \ref{claim:Misblockdiagonalandbounded} and by $E_2$, for any $u \in V(G)$, $l \in [r]$, $$|\{l\mbox{-length paths } P \subseteq G \mid u \in V(P)\}| \leq 2C(\log n)^2 \alpha^l,$$ so $$|\{\{i,j\} \in S \mid B_{i,j} \neq 0\}| \leq r (2C (\log n)^2)^2 \alpha^{r-1} = O(\alpha^r \log^5 n).$$ Therefore, by item (iii) of Proposition \ref{claim:Misblockdiagonalandbounded} $$\|B_{S \times S}\|_2 \leq \|B_{S \times S}\|_F = O(\alpha^{r/2} \log^{5/2} n),$$ as desired. ($\|\cdot \|_F$ denotes the Frobenius norm.) \qed

\begin{proof}[Proof of Proposition \ref{claim:Misblockdiagonalandbounded}]

Suppose $B_{i,j} \neq 0$. Since every vertex has at most one cycle in its $r$-neighborhood, there are at most 2 length-$(\leq r)$ paths between every pair of vertices (item (iv)), so $0 \leq A_{i,j}^{\{r\}}$. Also, since $A^{[r]}_{i,j} \in \{0,1\}$, the possible cases are:
\begin{enumerate}
        \item $A_{i,j}^{[r]} = 0$:
        \begin{enumerate}
            \item $A_{i,j}^{\{r\}} = 1$. There is a path of length $< r$ between $i$ and $j$, because otherwise $d_{G}(i,j) = r$. So there are two paths of length $\leq r$ between $i$ and $j$.
            \item $A_{i,j}^{\{r\}} = 2$. This case is impossible. There is no path of length $< r$ between $i$ and $j$, because there are at most two paths  of length $\leq r$ between $i$ and $j$, and $A_{i,j}^{\{r\}} = 2$ tells us that there are two paths of length $r$ between $i$ and $j$. Therefore, $d_{G}(i,j) = r$, so $A_{i,j}^{[r]} = 1$.
        \end{enumerate}
        
        \item $A_{i,j}^{[r]} = 1$:
        \begin{enumerate}
            \item $A_{i,j}^{\{r\}} = 0$. This case is impossible. The distance between $i$ and $j$ is $r$, so there should be a path of length $r$ between them.
            \item $A_{i,j}^{\{r\}} = 2$. There are two paths of length $r$ between $i$ and $j$.
        \end{enumerate}
\end{enumerate}
So if $B_{i,j} \neq 0$, then $|B_{i,j}| = 1$, and there are exactly two $(\leq r)$-length paths between $i$ and $j$. This case analysis proves items (ii) and (iii) of the claim.

The union of the two paths from $i$ to $j$ contains a simple cycle which is contained in the $r$-neighborhoods of both $i$ and $j$. Therefore $i \sim j$, proving item (i) of the claim.
\end{proof}

\subsection{Proof of Theorem \ref{ks2}}

The key to proving Theorem \ref{ks2} is the identity
$$A^{(r)} = \sum_{k=0}^{r} A^{[k]},$$ which is just another way to write $\1(d_{G}(i,j) \leq r) = \sum_{k=0}^r \1(d_{G}(i,j) = k)$. 

Informally, we know from Theorem \ref{ks1} that the top eigenvalues of the $A^{[k]}$ matrices have the desired separation properties. So if we can prove that the top eigenvectors of the $A^{[k]}$ matrices are all roughly equal, then they will be roughly equal to the top eigenvectors of $A^{(r)}$, essentially proving Theorem \ref{ks2}.

Keeping this intuition in mind, our first step is to reduce the problem of analyzing $A^{(r)}$ to the problem of analyzing a slightly simpler matrix: $$D = D(r) := \sum_{k=r/2}^r A^{\{k\}}.$$ We write $$A^{(r)} = D + A^{(r/2 - 1)} - \sum_{k=r/2}^r (A^{\{k\}} - A^{[k]}),$$
By the proof of Theorem \ref{ks2}, we know that conditioned on $E_1 \cap E_2$ (so with high probability, $$\|\sum_{k=r/2}^r A^{[k]} - A^{\{k\}}\|_2 \leq \sum_{k=0}^r \|A^{[k]} - A^{\{k\}}\|_2 = O(\alpha^{r/2} (\log n)^4).$$
And since by Lemma \ref{lem:neighborhoodgrowthbound} the neighborhoods of vertices do not grow too quickly, with high probability, $A^{(r/2 - 1)}$ is the adjacency matrix of a graph with maximum degree $O(\alpha^{r/2} (\log n)^2)$. Under this event, we also get the following bound, $$\|A^{(r/2 - 1)}\|_2 = O(\alpha^{r/2} (\log n)^2).$$

We can conclude by triangle inequality that $\|A^{(r)} - D\|_2 = O(\alpha^{r/2} (\log n)^4)$. Therefore, by the matrix perturbation arguments (Weyl's inequality and Davis-Kahan inequality) used to prove Theorem \ref{ks1}, it suffices to prove that the matrix $D$ has the spectral properties that we desire for $A^{(r)}$. Theorem \ref{ks2} will follow.

%Thankfully, we can 
%
%We can split this sum into two parts:
%$$A^{(r)} = \sum_{k=0}^{r/2} A^{[k]} + \sum_{k=r/2}^r A^{[k]}.$$ The first sum is negligible, because by a union bound over Theorem \ref{ks1} and by the triangle inequality, with probability $1 - O(n^{-\epsilon})$ $$\|\sum_{k=0}^{r/2} A^{[k]}\|_2 \leq \sum_{k=0}^{r/2} \|A^{[k]}\|_2 \leq \sum_{k=0}^{r/2} c_0 \alpha^{k} = O(\alpha^{l/2}).$$ As long as we can prove that $\sum_{k=r/2}^r A^{[k]}$ has the desired spectral properties, $\sum_{k=0}^{r/2} A^{[k]}$ will just be a small perturbation, like $B$ in the proof of Theorem \ref{ks1}, and we will get Theorem \ref{ks2} by Weyl's inequality and the Davis-Kahan theorem.

Our proof will now follow the proof of Theorem 2.1 in \cite{massoulie-STOC}, with some ``under the hood'' details modified. We will show that $D$ has a ``weak Ramanujan property'', similar to Theorem 2.4 of \cite{massoulie-STOC}:

\begin{lemma}\label{lem:weakramanujan} With high probability, $D$ satisfies the following weak Ramanujan property:
$$\sup_{\|u\|_2 = 1, u^T A^{\{r\}} 1 = u^T A^{\{r\}} X = 0} \|D u\|_2 = \alpha^{r/2} (\log n)^{O(1)},$$ where $1$ is the all-ones vector and $X \in \{-1,+1\}^n$ is the community label vector.
\end{lemma}

We will also need the following two lemmas. The first lemma tells us that the top eigenvectors of the $A^{\{k\}}$ matrices are pretty well aligned with the top eigenvectors of $A^{\{r\}}$:
\begin{lemma}\label{lem:BlBmalignmentlemma}
There are $c_0, \delta > 0$ such that with high probability, for all $k \in \{r/2,\ldots,r\}$, $$\left\|\frac{A^{\{r\}}1}{\|A^{\{r\}} 1\|_2} - \frac{A^{\{k\}}1}{\|A^{\{k\}} 1\|_2}\right\|_2 \leq c_0 n^{-\delta},$$ $$\left\|\frac{A^{\{r\}}X}{\|A^{\{r\}} X\|_2} - \frac{A^{\{k\}}X}{\|A^{\{k\}} X\|_2}\right\|_2 \leq c_0 n^{-\delta}.$$
\end{lemma}

\begin{lemma}[From \cite{massoulie-STOC}, with minor modification as in the Theorem \ref{ks3} remark]\label{lem:unionboundmassoulie}
There are $c_1, c_2 > 0$, and $g = o(1)$ such that with high probability, for all $k \in \{r/2,\ldots,r\}$:
\begin{enumerate}
\item[(i)] $c_1 \alpha^k < \|A^{\{k\}}\|_2 < c_2 \alpha^k$.
\item[(ii)] $A^{\{k\}} A^{\{k\}} 1 = \|A^{\{k\}}A^{\{k\}} 1\|_2\left(\frac{A^{\{k\}} 1}{\|A^{\{k\}} 1\|_2} + h_k\right)$ for a vector $h_k$ s.t. $\|h_k\|_2 < g = o(1)$.
\item[(iii)] $A^{\{k\}} A^{\{k\}} X = \|A^{\{k\}}A^{\{k\}} X\|_2\left(\frac{A^{\{k\}} X}{\|A^{\{k\}} X\|_2} + h'_k\right)$ for a vector $h'_k$ s.t. $\|h'_k\|_2 < g = o(1)$.
\end{enumerate}
\end{lemma}

We now show that Lemmas \ref{lem:weakramanujan}, \ref{lem:BlBmalignmentlemma}, and \ref{lem:unionboundmassoulie} imply Theorem \ref{ks2}.

Following the argument of Theorem 4.1 of \cite{massoulie-STOC}, it suffices to show that with high probability,
\begin{align}
\label{eq:DB1bounds} & \|D A^{\{r\}} 1\|_2 = \Theta(\alpha^r \|A^{\{r\}} 1\|_2)\qquad\mbox{and} \\ \label{eq:DBXbounds} &\|D A^{\{r\}} X\|_2 = \Theta(\beta^r \|A^{\{r\}} X\|_2).
\end{align} Since $A^{\{r\}} 1$ and $A^{\{r\}} X$ are asymptotically orthogonal (by Lemma 4.4 of \cite{massoulie-STOC}), and since $D$ has the weak Ramanujan property of Lemma \ref{lem:weakramanujan}, the variational definition of eigenvalues yields that the top two eigenvectors of $D$ will be asymptotically in the span of $A^{\{r\}} 1$ and $A^{\{r\}} X$.\footnote{Here we assume we are in the case $\beta^2 > \alpha$ of Theorem \ref{ks2}, since the other case is similar.} By the lower bound of \eqref{eq:DB1bounds} and the upper bound of \eqref{eq:DBXbounds}, the top eigenvalue of $D$ will be $\Theta(\alpha^{r})$, with eigenvector asymptotically parallel to $A^{\{r\}} 1$. Since $A^{\{r\}} X$ is asymptotically orthogonal to $A^{\{r\}} 1$, the second eigenvalue of $D$ will be $\Theta(\beta^r)$, with eigenvector asymptotically parallel to $A^{\{r\}} X$. This proves the theorem, since by Massouli\'e \cite{massoulie-STOC} $A^{\{r\}} 1$ and $A^{\{r\}} X$ are in fact asymptotically parallel to the top two eigenvectors of $A^{\{r\}}$.

The inequalities in \eqref{eq:DB1bounds} hold because with high probability, \begin{align}\frac{\|D A^{\{r\}} 1\|_2}{\|A^{\{r\}} 1\|_2} &=  \label{eq:spectralnormline} \left\|\sum_{k=r/2}^r A^{\{r\}} \frac{A^{\{r\}} 1}{\|A^{\{r\}} 1\|_2}\right\|_2
\\ &= \label{eq:triangleBlBe1} \left\|\sum_{k=r/2}^r A^{\{k\}} \frac{A^{\{k\}} 1}{\|A^{\{k\}} 1\|_2}\right\|_2 + O\left(\sum_{k=r/2}^r \|A^{\{k\}}\|_2 \left\|\frac{A^{\{r\}}1}{\|A^{\{r\}} 1\|_2} - \frac{A^{\{k\}}1}{\|A^{\{k\}} 1\|_2}\right\|_2\right)
\\ &= \label{eq:asymptoticparallel1} \left\|\sum_{k=r/2}^r A^{\{k\}} \frac{A^{\{k\}} 1}{\|A^{\{k\}} 1\|_2}\right\|_2 + O\left(\sum_{k=0}^r \alpha^{k} n^{-\delta}\right)
\\ &= \left\|\sum_{k=r/2}^r A^{\{k\}} \frac{A^{\{k\}} 1}{\|A^{\{k\}} 1\|_2}\right\|_2 + O(\alpha^{r} n^{-\delta})
\\ &= \label{eq:mastheorem} \left\|\sum_{k=r/2}^r (\frac{A^{\{k\}} 1}{\|A^{\{k\}} 1\|_2} + h_k) \frac{\|A^{\{k\}}A^{\{k\}} 1\|_2}{\|A^{\{k\}} 1\|_2}\right\|_2 + O(\alpha^{r}n^{-\delta})\quad\mbox{for }\|h_k\|_2 = o(1)
\\ &= \label{eq:triangleBlBe2}\left\|\sum_{k=r/2}^r (\frac{A^{\{r\}} 1}{\|A^{\{r\}} 1\|_2} + h_k) \frac{\|A^{\{k\}}A^{\{k\}} 1\|_2}{\|A^{\{k\}} 1\|_2}\right\|_2 \\ \nonumber &+ O\left(\sum_{k=r/2}^r \frac{\|A^{\{k\}}A^{\{k\}} 1\|_2}{\|A^{\{k\}} 1\|_2} \left\|\frac{A^{\{r\}} 1}{\|A^{\{r\}} 1\|_2} - \frac{A^{\{k\}} 1}{\|A^{\{k\}} 1\|_2}\right\|_2\right) + O(\alpha^{r} n^{-\delta})
\\ &= \label{eq:asymptoticparallel2}\left\|\sum_{k=r/2}^r (\frac{A^{\{r\}} 1}{\|A^{\{r\}} 1\|_2} + h_k) \frac{\|A^{\{k\}}A^{\{k\}} 1\|_2}{\|A^{\{k\}} 1\|_2}\right\|_2 + O(\alpha^{r}n^{-\delta})
\\ &= \left(\sum_{k=r/2}^r\frac{\|A^{\{k\}}A^{\{k\}} 1\|_2}{\|A^{\{k\}} 1\|_2}\right)(1 + o(1)) + O(\alpha^{r}n^{-\delta}).
\end{align}
\eqref{eq:triangleBlBe1} and \eqref{eq:triangleBlBe2} are derived by triangle inequality. \eqref{eq:asymptoticparallel1} and \eqref{eq:asymptoticparallel2} are consequences of Lemma \ref{lem:BlBmalignmentlemma}. \eqref{eq:mastheorem} follows by Lemma \ref{lem:unionboundmassoulie}.
Plugging in the bound on $\frac{\|A^{\{k\}} A^{\{k\}} 1\|_2}{\|A^{\{k\}} 1\|_2}$ from item (i) of Lemma \ref{lem:unionboundmassoulie} gives $$\|D A^{\{r\}} 1\|_2 = \Theta(\|A^{\{r\}} 1\|_2),$$ proving \eqref{eq:DB1bounds}. A similar argument proves \eqref{eq:DBXbounds}.
\qed

\subsubsection{Proof of Lemma \ref{lem:weakramanujan}}
Let $\delta > 0$. By triangle inequality and the definition $D = \sum_{k=r/2}^r A^{\{k\}}$, it suffices to prove that for all $k \in \{r/2,\ldots,r\}$, $$\sup_{\|u\|_2 = 1, u^T A^{\{r\}} 1 = u^T A^{\{r\}} X = 0} \|A^{\{k\}} u\|_2 \leq (\log n)^{O(1)} \alpha^{r/2},$$ where the bound is uniform over $k$. This follows from a union bound over Theorem 2.4 in \cite{massoulie-STOC} (the weak Ramanujan property for $A^{\{r\}}$).

\subsubsection{Proof of Lemma \ref{lem:BlBmalignmentlemma}}

Let $\mathcal{B} \subset V(G)$ denote the set of vertices $v$ such that there is a cycle in the $r$-neighborhood of $v$. Lemma 4.3 of \cite{massoulie-STOC} gives us the following bounds on entries of $A^{\{k\}}1$, $A^{\{k\}}X$ for all $k \in [r]$:
\begin{equation}\label{eq:B1notincBbound} v \not\in \mathcal{B} \implies (A^{\{k\}} 1)_v = \alpha^{k-r} (A^{\{r\}} 1)_v + O(\log n) + O(\sqrt{\log(n) \alpha^k})
\end{equation}
\begin{equation}\label{eq:BxnotincBbound} v \not\in \mathcal{B} \implies (A^{\{k\}} X)_v = \beta^{k-r} (A^{\{r\}} X)_v + O(\log n) + O(\sqrt{\log(n) \alpha^k})
\end{equation}
\begin{equation}\label{eq:B1incBbound}
v \in \mathcal{B} \implies \|(A^{\{k\}} 1)_v\|_2 = O(\alpha^k \log(n))
\end{equation}
\begin{equation}\label{eq:BxincBbound}
v \in \mathcal{B} \implies \|(A^{\{k\}} X)_v\|_2 = O(\beta^k \log(n))
\end{equation}

By Lemma 4.2 of \cite{massoulie-STOC}, $|\mathcal{B}| = O(\alpha^{2r}\log^4 n)$ with high probability, so by \eqref{eq:BxnotincBbound} and \eqref{eq:BxincBbound},

\begin{align*}\langle A^{\{r\}} X, A^{\{k\}} X \rangle &= \left(\sum_{v \not\in \mathcal{B}} (A^{\{r\}} X)_v^2 \beta^{k-r} + (A^{\{r\}} X)_v (O(\log n + \sqrt{\alpha^k \log n}))\right) + \sum_{v \in \mathcal{B}} O(\beta^{k+r} \log n) \\ \label{eq:end1alignBlsigmaBesigma} &=  \beta^{k-r}\|A^{\{r\}} X\|_2^2 + O(n \beta^r \sqrt{\alpha^k} \log^2 n) + O(\alpha^{2r}\beta^{k+r} \log^4 n).
\end{align*}
Noting that $\beta^2 > \alpha$, and  that $\|A^{\{r\}} X\|_2^2 = \Theta(n\beta^{2r})$ up to a factor of $\log^2 n$, there is $\delta > 0$ such that:
\begin{align}
\langle A^{\{r\}} X, A^{\{k\}} X \rangle &= \beta^{k-r}\|A^{\{r\}} X\|_2^2(1+o(n^{-\delta})).	
\end{align}
Similarly, we can show that $\|A^{\{k\}} X\|_2 = \beta^{k-r} \|A^{\{r\}} X\|_2(1 + o(n^{-\delta})).$
This implies that $$\left\|\frac{A^{\{r\}}X}{\|A^{\{r\}} X\|_2} - \frac{A^{\{k\}}X}{\|A^{\{k\}} X\|_2}\right\|^2_2 = 2 - 2\frac{\langle A^{\{r\}} X, A^{\{k\}} X \rangle}{\|A^{\{k\}} X\|_2\|A^{\{r\}} X\|_2} = 2 - 2(1 + o(n^{-\delta'})) = o(n^{-\delta'})$$ for some $\delta' > 0$. Similar arguments, using \eqref{eq:B1notincBbound} and \eqref{eq:B1incBbound} prove the analogous result for $\langle A^{\{k\}} 1, A^{\{r\}} 1\rangle$. \qed

\subsubsection{Proof of Lemma \ref{lem:unionboundmassoulie}} Item (i) is the statement from Theorem \ref{ks3} that $\lambda_1(A^{\{k\}}) = \Theta(\alpha^k)$, with the additonal subtletly that we can choose uniform constants in the $\Theta$ notation for $k \in \{r/2,\ldots,r\}$. Items (ii) and (iii) are equivalent to stating that for all $k \in \{l/2,\ldots,l\}$, $A^{\{k\}}A^{\{k\}} 1$ is asymptotically in the same direction as $A^{\{k\}} 1$, and $A^{\{k\}}A^{\{k\}} X$ is asymptotically in the same direction as $A^{\{k\}} X$. A union bound over Theorem 4.1 of \cite{massoulie-STOC} implies this is true for all $k \in \{r/2,\ldots,r\}$. \qed

\subsection{Proof of Theorem \ref{abp}}\label{proof_abp}

We will prove two sub-theorems, and the result of Theorem \ref{abp} is a consequence.  The proof is based on  counting closed walks, as in the approach of J. Friedman for the Alon-Boppana Theorem \cite{friedman}.

\begin{definition}$t_{2k}^{(r)}$ is the minimum, taken over all vertices $x\in V(G)$, of the number of closed walks of length $2k$ in $G^{(r)}$ terminating at $x$.\end{definition}

First, we will bound $\lambda_2(G^{(r)})$ in terms of $t_{2k}^{(r)}$.

\begin{theorem}\label{thm_lambda}
Let $G$ be a graph and $r\geq 1$.  Let $D$ be the diameter of $G$ and let $k$ satisfy $2k < \lceil D/r\rceil$.  Then $$\lambda_2(G^{(r)})^{2k}\geq t_{2k}^{(r)}.$$
\end{theorem}

\begin{proof}

Let $A^{(r)}$ be the adjacency matrix for $G^{(r)}$.
Let $f_1$ be an eigenfunction satisfying $A^{(r)}f_1 = \lambda_1(G^{(r)})f_1$.  By the Perron-Frobenius theorem, we can choose $f_1$ so that $f_1(z) > 0$ for all $z\in V$.

Because $A^{(r)}$ is symmetric, we can express $\lambda_2(G^{(r)})$ by the Rayleigh quotient \begin{align*}\lambda_2(G^{(r)})^{2k} = \sup_{f\perp f_1}\frac{\langle f, \parens{A^{(r)}}^{2k} f\rangle}{\langle f, f\rangle},\end{align*} where $k$ is any non-negative integer.  To obtain a lower bound on $\lambda_2(G^{(r)})$ we will set a test-function $f$ for this quotient.

For vertices $x,y\in V(G)$, set $f_{xy}(x) = f_1(y)$, $f_{xy}(y) = -f_1(x)$ and $f_{xy} \equiv 0$ otherwise.  Clearly $f_{xy}\perp f_1$.  It is well-known that $[\parens{A^{(r)}}^{2k}]_{ij}$ counts the number of walks in $G^{(r)}$ of length $2k$ that start at $i$ and end at $j$.  If $z\in V$, \begin{align*}\left(\parens{A^{(r)}}^{2k}f_{xy}\right)(z) = f_1(y)\left[\parens{A^{(r)}}^{2k}\right]_{xz}-f_1(x)\left[\parens{A^{(r)}}^{2k}\right]_{yz}.
\end{align*}

Let $D$ be the diameter of $G$, and choose $x$ and $y$ to be vertices with $d_G(x,y) = D$.  It follows that $d_{G^{(r)}}(x,y) = \lceil D/r \rceil$.  Choose $k$ so that $2k < \lceil D/r \rceil$.  There are no $2k$-walks in $G^{(r)}$ from $x$ to $y$ (or vice-versa), so that $\left[\parens{A^{(r)}}^{2k}\right]_{xy} = \left[\parens{A^{(r)}}^{2k}\right]_{yx} = 0$.  Now, our expression for $\lambda_2$ simplifies to \begin{align*}\lambda_2^{2k}\geq \frac{\langle f_{xy},\parens{A^{(r)}}^{2k}f_{xy}\rangle }{\langle f_{xy},f_{xy}\rangle} = \frac{f_1(y)^2\left[\parens{A^{(r)}}^{2k}\right]_{xx} + f_1(x)^2\left[\parens{A^{(r)}}^{2k}\right]_{yy}}{f_1(y)^2+f_1(x)^2} \geq t_{2k}^{(r)},\end{align*} where the last inequality holds because $t_{2k}^{(r)}\leq \left[\parens{A^{(r)}}^{2k}\right]_{zz}$ for all $z\in V$.
\end{proof}

Second, we will derive a lower bound on $t_{2k}^{(r)}$ in terms of the modified minimum degrees $\delta^{(i)}:0\leq i\leq r$.

\begin{theorem}\label{thm_tree}
Let $r$ and $k$ be positive integers. \begin{align*}\parens{t^{(r)}_{2k}}^{1/(2k)}\geq \parens{1-o(1)} \sum_{j = 0}^{r}\sqrt{\delta^{(j)}\delta^{(r-j)}}. \end{align*}
\end{theorem}  Here, the $o$ notation is in terms of $k$, treating $d$ and $r$ as constants.  

Fix $x$ to be vertex with the minimum number of closed walks of length $2k$ in $G^{(r)}$ terminating at $x$.  By definition, that number is $t_{2k}^{(r)}$.

For exposition, we will start by outlining Friedman's proof of this result for simple graphs.  

\begin{theorem}[Friedman,\cite{friedman}]

Let $G$ be a $d$-regular graph and $x\in V(G)$.  Consider $t_{2k}$, the number of closed walks on $G$ of length $2k$ terminating at $x$.  Then $$\parens{t_{2k}}^{1/(2k)}\geq \parens{1-o(1)}2\sqrt{d-1}.$$

\end{theorem}

The following is the outline of Friedman's proof:

\begin{itemize}
\item We wish to obtain a lower bound on the number of closed walks of length $2k$ from $x$ to $x$ in $G$.
\item Consider only those walks that are "tree-like", i.e., walks with the following description.   We start a walk at the root $x$.  We construct a tree labelled with vertices of $G$: if at some step we are at $y$, we can either move to the parent of $y$ or generate a child of $y$ by moving to some other neighbor.  We want to count only those walks that trace all edges of the tree twice and so terminate at the original root $x$; for example:
\begin{itemize}
\item $x,y,z,y,x$ is tree-like.
\item $x,y,z,y,z,y,x$ is also tree-like : in this case $z$ is the label for two children of the same vertex.
\item $x,y,z,x$ is closed but not tree-like, because it ends on a node of depth $3$ rather than the root.
\end{itemize}
(Another characterization is that tree-like walks correspond to closed walks on the cover graph of $G$.)
\item In a tree-like walk of length $2k$, the sequence of moving to a child / moving to a parent makes a Dyck word of length $2k$, of which there are the $k$-th Catalan number $C_k$.
\item For a Dyck word, there are at least $d-1$ choices of which child to generate at each of those $k$ steps, and exactly $1$ choice at each step that we return to the parent node.  Note here that $d-1 = \delta^{(1)}$ for a $d$-regular graph.
\item In total we find that $$t_{2k}\geq (d-1)^kC_k\gtrapprox \frac{\parens{2\sqrt{d-1}}^{2k}}{k^{3/2}},$$ the bound follows.
\end{itemize}

We will to use a similar process to bound $t_{2k}^{(r)}$.  A simple suggestion is to consider walks in $G^{(r)}$ that correspond to closed walks on the cover graph of $G^{(r)}$.  The argument outlined above works (as it will for the cover graph of \textit{any} graph) and we see that $$t_{2k}^{(r)}\geq (d^{(r)}-1)C_k\approx \parens{2\sqrt{d^{(r)}-1}}^{2k},$$ where $d^{(r)}$ is the minimum degree in $G^{(r)}$.

The problem is that this is not tight: if we approximate $\delta^{(i)} \approx d^{(i)} \approx d^i$ (as is the case in a $d$-regular graph with girth larger than $2i$), this result gives $$\parens{t_{2k}^{(r)}}^{1/(2k)} \gtrapprox 2d^{r/2},$$  while our theorem has the result
$$\parens{t_{2k}^{(r)}}^{1/(2k)} \gtrapprox (r+1)d^{r/2}.$$  In order to prove the theorem we must improve this method.  Consider a class of walks we are not counting:

Let $r = 2$, suppose $x,y,z$ are vertices with a common neighbor $w$.  $x,y,z,x$ is a closed walk in $G^{(2)}$ but it does not correspond to a closed walk in the cover graph of $G^{(2)}$.  However, the underlying $G$-walk $x,w,y,w,z,w,x$ \textit{does} correspond to a closed walk on the cover graph of $G$.  This is the type of walk that we want to count, but failed to do so in our first attempt.

So instead we can try to count the number of walks in the cover graph of $G$ that are the underlying walk for some $G^{(r)}$-walk.  But this suggestion introduces a new issue: a given $G^{(r)}$-walk may correspond to several different underlying walks on the cover graph of $G$.  For instance (again in the setting $r=2$), if $x\sim y\sim z\sim w\sim x$, then the $G^{(2)}$-walk $x,z,x$ corresponds to both $x,y,z,y,x$ and $x,w,z,w,x$.  If we want to count the underlying walks as a lower bound on the number of $G^{(2)}$-walks, we must disregard all but at most $1$ of the underlying walks that correspond to a given $G^{(2)}$-walk.  For this reason we introduce a system of canonical paths between neighbors in $G^{(r)}$ that we follow when generating an underlying walk on $G$.

\begin{definition}[Canonical $i$-path]For all ordered pairs of distinct vertices $(v,w)$ with $d_G(v,w) = i$, we arbitrarily choose a canonical path of length $i$ from $v$ to $w$.\end{definition}

We now define the set of tree-like walks in $G^{(r)}$, which we will later count.  The tree-like walks are analogous to walks on the cover graph of $G$.  We first define the related concept of a sequence of canonically constructed walks.

\begin{definition}[Sequence of $r$-canonically constructed walks]
Let $k>0$.  A sequence of $r$-canonically constructed walks is a sequence $W_0,\dots,W_{2k}$ of walks that is constructed according to the following process:

Assume $W_i$ has length at least $r$, so it can be expressed as $W_i = \dots, v_0,v_1,\dots, v_r$.  We can then make a move of type $m$ to obtain $W_{i+1}$, where $m$ is an integer satisfying $0\leq m\leq r$.

In a move of type $m$, we start by removing the last $r-m$ vertices from $W_i$, leaving $\dots v_0,v_1,\dots, v_m$.  Let $y$ be a vertex with $d_G(v_m,y) = m$ and $d_G(v_{m-1},y)\geq m$.  We can append to our sequence the canonical walk from $v_m$ to $y$ of length $m$ to obtain $W_{i+1}$.  Observe that $d_G(v_r,y)\leq d_G(v_r,v_m) + d_G(v_m,y) = (r-m)+m = r$, so that $v_r$ and $y$ are neighbors in $G^{(r)}$.

We require $W_0 = W_{2k} = x$.  To find $W_1$, we are allowed to make a move of type $r$ (removing no vertices from $W_0$).  When $0 < i < 2k$, we require that $W_i$ has length at least $r$.
\end{definition}

The motivation for this definition is that we are taking a $G^{(r)}$-walk on the endpoints of $W_0,\dots,W_{2k}$: observe that in such a sequence of walks, the endpoints of $W_i$ and $W_{i+1}$ have graph distance at most $r$.

Note that the final move from $W_{2k-1}$ to $W_{2k}$ must be of type $0$, where $W_{2k-1}$ has length $r$ and $W_{2k}$ has length $0$.

Note that if there is a move of type $m$ from $W_i$ to $W_{i+1}$ where $W_{i+1}$ ends in $y$, no move of type $<m$ from $W_i$ can end in $y$.  Suppose on the contrary there is $j<m$ so that some move of type $j$ ends in $y$.  Then $d_G(v_{m-1},y)\leq d_G(v_{m-1},v_j) + d_G(v_j,y) \leq (m-1-j)+j = m-1$.  This contradicts that $d_G(v_{m-1},y) \geq m$.

Also note that, given $W_i$ is long enough that moves of type $m$ are allowed, the number of possible moves of type $m$ that can generate distinct possibilities for $W_{i+1}$ is at least $\delta^{(m)}$.

\begin{definition}[Tree-like walk in $G^{(r)}$]
Let $k>0$.  A tree-like walk of length $2k$ in $G^{(r)}$ is the closed walk made by the endpoints of an sequence of $r$-canonically constructed walks starting at $x$ and having length $2k$.
\end{definition}

We will count the number of tree-like walks in $G^{(r)}$.  First, it is useful to prove an equivalence between the tree-like walk $x_0,\dots, x_{2k}$ and underlying sequence of $r$-canonically constructed walks $W_0,\dots, W_{2k}$.

\begin{theorem}\label{thm_treeseq}
Let $k > 0$.  Let $\mathcal{W}$ be the set of sequences of $r$-canonically constructed walks.  Let $\mathcal{X}$ be the set of tree-like walks in $G^{(r)}$.  Then the function that takes an sequence of $r$-canonically constructed walks and outputs the tree-like walk consisting of its final vertices is a bijection between $\mathcal{W}$ and $\mathcal{X}$.
\end{theorem}

\begin{proof}
It is clear from the definition of tree-like walks that this function is surjective.  It remains to show that it is injective.  We will argue that it is possible to recover the sequence of $r$-canonically constructed walks $W_0,\dots, W_{2k}$ given the tree-like walk $x_0,\dots, x_{2k}$.

Start with $W_0 = x = x_0$.  Given $x_0,\dots x_k$ and $W_i$ we wish to find $W_{i+1}$.  We first determine the type of the move from $W_i$ to $W_{i+1}$:

Write $W_i = \dots, v_0,v_1,\dots v_r$.  Let $m$ be the least integer so that $d_G(x_{i+1},v_m) \leq m$.  Then the move from $W_i$ to $W_{i+1}$ must be of type $m$, for the following reasons:  assume for contradiction that the type is not $m$.

If the move is in fact of type $j < m$, then $d_G(x_{i+1},v_j) = j$, and so $m$ is not the least integer satisfying $d_G(x_{i+1},v_m) \leq m$, this is a contradiction.

Instead the move may be of type $j > m$.  But then, $d_G(v_{j-1}, x_{i+1})\leq d_G(v_{j-1},v_m) + d_G(v_m, x_{i+1})\leq (j-1-m)+m = j-1$.  This contradicts the assumption that $d_G(v_{j-1},x_{i+1})\geq j$ in a move of type $j$.

So the move must be of type $m$.  We can determine $W_{i+1}$ by removing all vertices after $v_m$ from $W_i$ and then appending the canonical walk from $v_m$ to $x_{i+1}$.
\end{proof}

Because every tree-like walk is a walk, $t_{2k}^{(r)}\geq \abs{\mathcal{X}}$.  It follows from Theorem \ref{thm_treeseq} that $t_{2k}^{(r)}\geq \abs{\mathcal{W}}$.  In order to prove Theorem \ref{thm_tree}, we will look for a lower bound on $|\mathcal{W}|$.

First, given $W_0,\dots, W_i$, we count the number of ways that we can make a move of type $m$ to generate a walk $W_{i+1}$.  Assuming $W_i = \dots, v_0,v_1,\dots v_r$, the last vertex of $W_{i+1}$ can be any vertex $y$ so that $d_G(v_m,y) = m$ and $d_G(v_{m-1},y) \geq m$.  Because $(v_m, v_{m-1})\in E$, the number of such vertices is at least $\delta^{(m)}$.

Now, suppose $m_1,\dots, m_{2k}$ is a legal sequence of move types that generate a closed walk. 
Here, $m_i$ represents the move type between $W_{i-1}$ and $W_i$.  The number of tree-like walks with such a sequence is at least $$\prod_{i=1}^{2k} \delta^{(m_i)}.$$

The remaining difficulty is to describe the allowed sequences of move types.  It is straightforward to see that $m_1 = r$ and $m_{2k} = 0$ for any allowed sequence.  It is also true that, summing over all moves, we add $\sum_i m_i$ edges and remove $\sum_i (r-m_i)$ edges.  Because we start at $W_0 = x$ and end at $W_{2k} = x$ with exactly $0$ edges, we see that these quantities must be equal: in other terms, $\sum_i m_i = \tfrac{1}{2}\sum_i r = kr$.  But apart from such observations, we have the difficult problem of making sure that $W_i$ has length at least $r$ after every move (besides $i = 0$ and $i = 2k$).  For example, if $r = 6$ we cannot take the move type sequence $6,1,3,3,5,0$: in this case, $W_2$ has length $2$.  Then to obtain $W_3$ we should start by removing the last $3$ edges of $W_2$, but this is impossible!  For this reason we require that $W_2$ has length at least $6$.

These restrictions make counting $\mathcal{W}$ difficult.  Instead we will count a subset of $\mathcal{W}$.  We will allow only those sequences of move types in which:

\begin{itemize}
\item $m_1 = r$ and $m_{2k} = 0$.
\item For every $j > r/2$, the subsequence that consists of only moves of types $j$ and $r-j$ is a Dyck word starting with type $j$.
\item The subsequence of moves of types $r$ and $0$ is still a Dyck word if the entries $m_1$ and $m_{2k}$ are removed.
\item If $r$ is even, there may also be moves of type $r/2$.
\end{itemize}

In any tree-like walk with such a sequence of move types, $W_i$ has length $\geq r$ whenever $0 < i < {2k}$, so that all moves are in fact legal.  In addition, because the number of moves of types $j$ and $r-j$ are always equal, the walk is necessarily closed; i.e, in total over all moves, we will add and remove the same number of edges.  

In order to count the walks with such move sequences, we need to prove a lemma.  The well-known binomial identity reveals that \begin{align*}
\sum_{i_1 + \dots i_k = n}\binom{n}{i_1,\dots, i_k}x_1^{i_1}\dots x_k^{i_k} = \parens{\sum_{i=1}^k x_i}^n.\end{align*}  In our results, we will use similar sums, except, instead of summing over all partitions of $n$, we sum over only even partitions of $n$.  (We require that $n$ be even so that such a partition is possible.)  In this lemma we bound the sum over even partitions in terms of the binomial identity.

\begin{lemma} \label{lemma_binom}
Suppose $x_1,\dots,x_k \geq 0$ and $2n$ is a non-negative even integer.  Then

\begin{align*}\sum_{2m_1+\dots 2m_k = 2n} \binom{2n}{2m_1,\dots, 2m_k}x_1^{2m_1}\dots x_k^{2m_k}\geq \frac{1}{2^{k-1}} \parens{\sum_{i=1}^k x_i}^{2n},\end{align*}

where the first sum is over all $k$-tuples of non-negative even integers that sum to $2n$.

\end{lemma}

\begin{proof}

Consider vectors $j\in \{-1,1\}^k$.

\begin{align*}\sum_j \parens{\sum_{i=1}^k j_ix_i}^{2n} = 2^k\sum_{2m_1+\dots 2m_k = 2n} \binom{2n}{2m_1,\dots, 2m_k}x_1^{2m_1}\dots x_k^{2m_k}.\end{align*} Here, each term in the right-hand sum is counted once for each of the $2^k$ choices of $j$.  Any term of the binomial expansion that does not correspond to an even partition of $2n$ is positive for exactly $2^{k-1}$ values of $j$ and negative for the other $2^{k-1}$.  In the sum over $j$ such a term cancels.

Also, \begin{align*}\sum_j \parens{\sum_{i=1}^k j_ix_i}^{2n} \geq 2\parens{\sum_{i=1}^k x_i}^{2n},\end{align*} because the left-hand side is a sum of non-negative quantities that contains the right hand side twice: when $j \equiv 1$ and $j\equiv -1$.  Combining we find \begin{align*}2^k\sum_{2m_1+\dots 2m_k = 2n} \binom{2n}{2m_1,\dots, 2m_k}x_1^{2m_1}\dots x_k^{2m_k} \geq 2\parens{\sum_{i=1}^k x_i}^{2n},\end{align*} the result immediately follows.

\end{proof}

Finally we are able to give a lower bound on $t_{2k}^{(r)}$.  We will prove the bound separately in the cases that $r$ is odd and $r$ is even: the difference is that when $r$ is even we must consider moves of type exactly $r/2$.

\begin{proof}[Proof of Theorem \ref{thm_tree}]

First, suppose $r$ is odd.  In this case, the number of accepted sequences of move types is \begin{align*}
\sum_{n_0+\dots +n_{(r-1)/2} = k-1}\binom{2k-2}{2n_1,\dots, 2n_{(r-1/2)}}\prod_{j=0}^{(r-1)/2}C_{n_j},\end{align*} where $n_j$ is the number of times types $j$ (and $r-j$) appear in the sequence.  The total number of tree-like walks on $G^{(r)}$ that correspond to such sequences is
\begin{align*}
t_{2k}^{(r)}&\geq \sum_{n_0+\dots +n_{(r-1)/2}\atop = k-1}\binom{2k-2}{2n_0,\dots, 2n_{(r-1/2)}}\prod_{j=0}^{(r-1)/2}C_{n_j}\parens{\delta^{(j)}}^{n_j}\parens{\delta^{(r-j)}}^{n_j}\\
&=\parens{1+O\parens{\tfrac{1}{k}}} \sum_{n_01+\dots +n_{(r-1)/2} \atop = k-1}\binom{2k-2}{2n_0,\dots, 2n_{(r-1/2)}}\prod_{j=0}^{(r-1)/2}\frac{\parens{4\delta^{(j)}\delta^{(r-j)}}^{n_j}}{n_j^{3/2}\sqrt{\pi}}\\
& \geq \parens{1+O\parens{\tfrac{1}{k}}}\parens{\frac{r^3}{8k^3\sqrt{\pi}}}^{r/2}\sum_{n_0+\dots +n_{(r-1)/2}\atop = k-1}\binom{2k-2}{2n_0,\dots, 2n_{(r-1/2)}}\prod_{j=0}^{(r-1)/2}\parens{2\sqrt{\delta^{(j)}\delta^{(r-j)}}}^{2n_j}\\
& \geq \parens{1+O\parens{\tfrac{1}{k}}}\parens{\frac{r^3}{16k^3\sqrt{\pi}}}^{r/2}\parens{\sum_{j = 0}^{(r-1)/2}2\sqrt{\delta^{(j)}\delta^{(r-j)}}}^{2k-2}\\
& = \parens{1+O\parens{\tfrac{1}{k}}}\parens{\frac{r^3}{16k^3\sqrt{\pi}}}^{r/2}\parens{\sum_{j = 0}^{r}\sqrt{\delta^{(j)}\delta^{(r-j)}}}^{2k-2}.
\end{align*}

Here, the first equality is a standard approximation for $C_n$ and the last inequality is the first statement of Lemma \ref{lemma_binom}.  The result of Theorem \ref{thm_tree} follows (though it remains to prove the theorem when $r$ is even): \begin{align*}\parens{t_{2k}^{(r)}}^{1/(2k)}\geq \parens{1-o(1)}\sum_{j = 0}^{r}\sqrt{\delta^{(j)}\delta^{(r-j)}}.\end{align*}

If instead $r$ is even, the computations are complicated by the existence of moves of type $r/2$.  The number of accepted move type sequences is \begin{align*}
\sum_{2n_0+\dots +2n_{r/2-1} = 2k-2-n_{r/2}}\binom{2k-2}{2n_1,\dots, 2n_{r/2-1},n_{r/2}}\prod_{j=0}^{r/2-1}C_{n_j},\end{align*} where $n_j$ is the number of times types $j$ (and $r-j$) appear in the sequence.  Using almost the same argument as before, the number of tree-like walks on $G^{(r)}$ that correspond to such sequences is
\begin{align*}
t_{2k}^{(r)}&\geq \sum_{n_0+\dots +n_{r/2-1}\atop = k-1-(n_{r/2})/2}\binom{2k-2}{2n_0,\dots, 2n_{r/2-1},n_{r/2}}\parens{\prod_{j=0}^{r/2-1}C_{n_j}\parens{\delta^{(j)}}^{n_j}\parens{\delta^{(r-j)}}^{n_j}}\parens{\delta^{(r/2)}}^{n_{r/2}}\\
&=\parens{1+O\parens{\tfrac{1}{k}}} \sum_{n_0+\dots +n_{r/2-1} \atop = k-1-(n_{r/2})/2}\binom{2k-2}{2n_0,\dots, 2n_{r/2-1},n_{r/2}}\parens{\prod_{j=0}^{r/2-1}\frac{\parens{4\delta^{(j)}\delta^{(r-j)}}^{n_j}}{n_j^{3/2}\sqrt{\pi}}}\parens{\delta^{(r/2)}}^{n_{r/2}}\\
& \geq \parens{1+O\parens{\tfrac{1}{k}}}\parens{\frac{r^3}{8k^3\sqrt{\pi}}}^{r/2}\sum_{n_0+\dots +n_{r/2-1}\atop = k-1-(n_{r/2})/2}\binom{2k-2}{2n_0,\dots, 2n_{r/2-1},n_{r/2}}\\ &\parens{\prod_{j=0}^{r/2-1}\parens{2\sqrt{\delta^{(j)}\delta^{(r-j)}}}^{2n_j}}\parens{\delta^{(r/2)}}^{n_{r/2}}\\
& \geq \parens{1+O\parens{\tfrac{1}{k}}}\parens{\frac{r^3}{16k^3\sqrt{\pi}}}^{r/2}\parens{\delta^{(r/2)}+\sum_{j = 0}^{r/2-1}2\sqrt{\delta^{(j)}\delta^{(r-j)}}}^{2k-2}\\
& = \parens{1+O\parens{\tfrac{1}{k}}}\parens{\frac{r^3}{16k^3\sqrt{\pi}}}^{r/2}\parens{\sum_{j = 0}^{r}\sqrt{\delta^{(j)}\delta^{(r-j)}}}^{2k-2}.
\end{align*}

As in the case where $r$ is odd, the method is to use the standard approximation on $C_n$ and then to cite Lemma \ref{lemma_binom} in order to express the sum as a binomial expansion.  The result of Theorem \ref{thm_tree} follows: \begin{align*}\parens{t_{2k}^{(r)}}^{1/(2k)}\geq \parens{1-o(1)}\sum_{j = 0}^{r}\sqrt{\delta^{(j)}\delta^{(r-j)}}.\end{align*}

\end{proof}

Now, combining Theorems \ref{thm_lambda} and \ref{thm_tree}, we see that \begin{align*} \lambda_2(G^{(r)})\geq \parens{1-o(1)}\sum_{j = 0}^{r}\sqrt{\delta^{(j)}\delta^{(r-j)}} = \parens{1-o(1)}(r+1)\hat{d}_r\,^{r/2}(G),\end{align*} where the $o$ notation is in terms of $k = \tfrac{1}{2}\lceil \text{diam}(G)/r \rceil$.  Under the hypotheses of Theorem \ref{abp}, we have $\text{diam}(G_n) = \omega(1)$ and $r_n = \epsilon\cdot\text{diam}(G_n)$, so that $k_n = \tfrac{1}{2}\lceil 1/\epsilon \rceil$.  Letting $\epsilon \to 0$, we have $k_n = \omega(1)$.  The result of Theorem \ref{abp} follows.

\subsection{Proof of Lemma \ref{lemma_reg} and Lemma \ref{lemma_girth}}

\begin{proof}[Proof of Lemma \ref{lemma_reg}]

Let $G$ be a random $d$-regular graph on $n$ vertices.  Let $r>0$ and let $x\in V$.  We want to bound the probabilty of the induced subgraph $B_r(x)$ containing two different cycles.

Consider the set $W$ of nonbacktracking walks in $G$ that start at $x$ and have length at most $r$.  There are $|W| = O(d^r)$ such walks.  A cycle exists in $B_r(x)$ if there are two such walks - neither of which is an extension of the other by one step - whose endpoints are adjacent.

Take $w_1,\dots, w_4\in W$.  We approximate the probability of the event $E_{w_1,w_2,w_3,w_4}$: the event that $B_r(x)$ has two cycles: one cycle because $w_1$ and $w_2$ have adjacent endpoints and a different cycle because $w_3$ and $w_4$ have adjacent endpoints.  Given that $E_{w_1,w_2,w_3,w_4}$ requires two pairs of vertices to be adjacent, with the the endpoints of $w_3$ and $w_4$ making a different cycle than that formed from the union of $w_1$ and $w_2$, $P(E_{w_1,w_2,w_3,w_4}) = O(d^2/n^2)$.  Let $E_x$ be the union of $E_{w_1,w_2,w_3,w_4}$ over all quadruples in $W$, i.e., the event that two cycles exist in the $r$-neighborhood of $x$ in $G$.  As $|W| = O(d^r)$, $P(E_x) = O(d^{4r+2}/n^2)$.  Let $E$ be the probability of any vertex in $G$ having $2$ cycles in its $r$-neighborhood, i.e., $E = \bigcup_x E_x$.  It is clear that $P(E) = O(d^{4r+2}/n)$.  Choose $r$ so that $r < \tfrac{\log n - 2}{4\log d }$.  Following this choice, $P(E) = o_n(1)$, so with high probability no ball $B_r(x)$ around any vertex $x$ contains two cycles.
 
%\Pnote{I could write this part better.\\ UPDATE: I have used some more notation.  Some arguments are still not totally spelled out.}

 In the high probability case $\overline{E}$, we can tightly bound $\delta^{(r)}$.  The assumption that $B_r(x)$ contains at most one cycle means that it closely resembles a tree - this makes the following computation of $\delta^{(r)}$ straightforward.
 
 Let $x\in V$.  There are $d(d-1)^{r-1}$ nonbacktracking walks of length $r$ starting at $x$.  For any $y\sim x$, exactly $(d-1)^{r-1}$ of those walks start $x,y,\dots$.  Because a vertex at distance $r$ from $x$ and $r-1$ from $y$ must be the endpoint of such a walk, there are at most $(d-1)^{r-1}$ of those vertices.  A cycle can occur in $B_r(x)$ in one of two ways:
 
 First, an even cycle means that there is some vertex $z$ so that there are $2$ nonbacktracking walks of length $d_G(x,z)$ from $x$ to $z$.  Write $i = d_G(x,z);$ there are $(d-2)(d-1)^{r-i-1}$ ways to extend those paths to length $r$ so that the endpoint is in $S_r(x)$.  For any of the $d(d-1)^{i-1}-2$ other nonbacktracking walks of length $i$, there are $(d-1)^{r-i}$ ways to extend them to distance $r$, each corresponding to a unique vertex in $S_r(x)$.  In total we find $|S_r(x)| = d(d-1)^{r-1}-d(d-1)^{r-i-1}$.  Because $i\geq 2$, $|S_r(x)| \geq d(d-1)^{r-1}-d(d-1)^{r-3}$ if there is an even cycle.
 
 Second, an odd cycle means that there are two adjacent vertices $z_1,z_2$ with $j = d_G(x,z_1) = d_G(x,z_2)$.  Each of those vertices can be extended to a nonbacktracking walk ending in a unique vertex of $S_r(x)$ in $(d-2)(d-1)^{r-j-1}$ different ways.  For any of the $d(d-1)^{j-1}-2$ other nonbacktracking walks of length $j$, there are $(d-1)^{r-j}$ ways to extend them to distance $r$, each corresponding to a unique vertex in $S_r(x)$.  In total we have $|S_r(x)| = d(d-1)^{r-1}-2(d-1)^{r-j-1}$. Because $j\geq 1$, it follows that $|S_r(x)|\geq d(d-1)^{r-1}-2(d-1)^{r-2}$.
 
 As there is at most one cycle which is either odd or even, we find the overall bound $|S_r(x)|\geq d(d-1)^{r-1}-2(d-1)^{r-2}$.  Because there are at most $(d-1)^{r-1}$ vertices in $S_r(x)$ at distance $r-1$ from $y$, we are left with $\delta^{(r)}\geq (d-1)^{r}-2(d-1)^{r-2}$.  It is clearly the case that $\delta^{(r)}\leq (d-1)^{r}$, so with high probability $\delta^{(r)} = \parens{1+o_d(1)}(d-1)^r$.

The result $\hat{d}_r(G) = \parens{1+o_d(1)}d$ follows.

\end{proof}

\begin{proof}[Proof of Lemma \ref{lemma_girth}]

Let $G$ be a $d$-regular Ramanujan graph with girth $g$.  Assume that $2r < g$.

It is straightforward to compute that $\delta^{(i)} = (d-1)^i$ for all values $0\leq i \leq r$.  Our results show that  $\lambda_2(A^{(r)})\geq (1-o(1))(r+1)(d-1)^{r/2}$.  We will show that this is tight up to a factor $1+o_d(1)$.  To do so, we will find an upper bound for $\lambda_2(A^{(r)})$.
Because $G$ is isomorphic to a $d$-regular tree in any $r$-neighborhood, there is a recursive formula for $A^{(r)}$:

$$A^{(r)} = AA^{(r-1)}-(d-1)A^{r-2}$$ if $2\leq r$.  The base cases are $A^{(0)} = I$ and $A^{(1)} = A+I$.

We briefly justify this recursion:  $[AA^{(r-1)}]_{ij}$ counts the number of neighbors $v$ of $j$ satsifying $d_G(i,v)\leq r-1$.  Because of the girth bound, this number is $1$ if $d_G(i,j) = r$ or $r-1$.  It is $0$ if $d_G(i,j) > r$ and $d$ if $d_G(i,j) \leq r-2$.  Because $[A^{(r)}]_{ij} = 1$ iff $d_G(i,j)\leq r$, we subtract $(d-1)A^{(r-2)}$ from the previous term.

It is easy to see that there is a sequence of polynomials $p^{(r)}$ so that $A^{(r)} = p^{(r)}(A)$ (though it requires some effort to compute $p^{(r)}$.)  If $\lambda_1\geq \lambda_2\geq\dots$ are the eigenvalues of $A$, then $p^{(r)}(\lambda_1), p^{(r)}(\lambda_2),\dots$ are the eigenvalues of $A^{(r)}$ (but note that these values are not necessarily ordered.)  The largest eigenvalue of $A^{(r)}$ is $p^{(r)}(d)$, achieved by eigenvector $v\equiv 1$.  Because all other eigenvalues of $A^{(r)}$ are $p^{(r)}(\lambda)$ for some $|\lambda |< 2\sqrt{d-1}$, $$\lambda_2(A^{(r)}) \leq \max_{|x|< 2\sqrt{d-1}} |p^{(r)}(x)|.$$

We must now compute $p^{(r)}$.  Given the recursive formula $p^{(r)}(x) = xp^{(r-1)}(x) - (d-1)p^{(r-2)}(x)$ with base cases $p^{(0)}(x) = 1, p^{(1)}(x) = x+1$, it is straightforward to write the solution:  if $\abs{x}< 2\sqrt{d-1}$, \begin{align*}p^{(r)}(x) &= \parens{\frac{1}{2} - \frac{i}{2}\frac{x/2 + 1}{\sqrt{(d-1)-x^2/4}}}\parens{\frac{x}{2} + i\sqrt{(d-1)-x^2/4}}^r\\ &+\parens{\frac{1}{2} + \frac{i}{2}\frac{x/2 + 1}{\sqrt{(d-1)-x^2/4}}}\parens{\frac{x}{2} - i\sqrt{(d-1)-x^2/4}}^r\end{align*}

Define $\theta = arccos\parens{\frac{x/2}{\sqrt{d-1}}}$.

\begin{align*}
p^{(r)}(x) &= \frac{1}{2}\parens{1-i\frac{1}{\sqrt{d-1}\sin\theta}-i\frac{\cos\theta}{\sin\theta}}\parens{d-1}^{r/2}\parens{\cos\theta + i\sin\theta}^r\\ &+ \frac{1}{2}\parens{1+i\frac{1}{\sqrt{d-1}\sin\theta}+i\frac{\cos\theta}{\sin\theta}}\parens{d-1}^{r/2}\parens{\cos\theta - i\sin\theta}^r.\\
&= \frac{(d-1)^{r/2}}{2}\parens{1-i\frac{1}{\sqrt{d-1}\sin\theta}-i\frac{\cos\theta}{\sin\theta}}\parens{\cos(r\theta)+ i\sin(r\theta)}\\&+\frac{(d-1)^{r/2}}{2}\parens{1-i\frac{1}{\sqrt{d-1}\sin(-\theta)}-i\frac{\cos(-\theta)}{\sin(-\theta)}}\parens{\cos(-r\theta)+ i\sin(-r\theta)}\\
& = \parens{d-1}^{r/2}\parens{\cos(r\theta)+\frac{\sin(r\theta)}{\sqrt{d-1}\sin\theta}+ \frac{\cos\theta\sin(r\theta)}{\sin\theta}}.
\end{align*}

It follows that \begin{align*}\abs{p^{(r)}(x)}\leq \parens{d-1}^{r/2}\parens{1+ \frac{r}{\sqrt{d-1}} + r} = \parens{1+o_d(1)}(r+1)(d-1)^{r/2}.\end{align*}  Holding $r$ as a constant, our lower bound $\lambda_2(A^{(r)})= (1+o_d(1)) (r+1)d^{r/2}$ is tight.  This completes the proof.
\end{proof}

%\subsection{Proof plan for Conjecture \ref{conj_gbm}}\label{plan}

\newcommand{\etalchar}[1]{$^{#1}$}
\providecommand{\bysame}{\leavevmode\hbox to3em{\hrulefill}\thinspace}
\providecommand{\MR}{\relax\ifhmode\unskip\space\fi MR }
\providecommand{\MRhref}[2]{%
  \href{http://www.ams.org/mathscinet-getitem?mr=#1}{#2}
}
\providecommand{\href}[2]{#2}

\end{document}